\documentclass[10pt]{IEEEtran}

\usepackage{amsmath}
\usepackage{graphicx}
\usepackage{enumerate}
\usepackage{xurl}
\usepackage{mathtools}
\usepackage{subfigure}
\usepackage{algorithm}
\usepackage{algpseudocode}
\usepackage{stfloats}
\usepackage{lineno}
\usepackage{hyperref}
\usepackage{bm}
\usepackage{bbm}
\usepackage{amssymb}
\usepackage{enumerate}
\usepackage{xcolor}
\usepackage{float}
\usepackage{cite}
\usepackage{tabularx}
\usepackage{tabu}
\usepackage{multicol}
\usepackage{multirow}
\usepackage{colortbl,booktabs,threeparttable}
\usepackage{dcolumn}
\usepackage{flushend}
\usepackage{soul}
\usepackage{CJKutf8}
\usepackage{ragged2e}
\usepackage{hanging}

\graphicspath{{Figures/}}

\newcommand{\bmb}[1]{\bar{\bm{#1}}}

\newcommand{\bmt}[1]{\tilde{\bm{#1}}}

\renewcommand{\vec}[1]{\bm #1}
\newcommand{\rvec}[1]{\mathbf{#1}}
\newcommand{\mat}[1]{\bm #1}
\newcommand{\matb}[1]{\bmb #1}
\newcommand{\matt}[1]{\bmt #1}
\newcommand{\rmat}[1]{\mathbf{#1}}
\newcommand{\opvec}{\operatorname{vec}}

\newcommand{\bb}[1]{\mathbb{#1}}
\renewcommand{\cal}[1]{\mathcal{#1}}
\newcommand{\bfit}[1]{\textbf{\textit{#1}}}

\newcommand{\Tr}{\operatorname{Tr}}
\newcommand{\ip}[2]{\left<{#1},{#2}\right>}

\newcommand{\R}{\mathbb{R}}
\newcommand{\C}{\mathbb{C}}

\renewcommand{\th}{\text{th}}
\newcommand{\T}{\mathsf{T}}
\renewcommand{\H}{\mathsf{H}}
\renewcommand{\st}{\text{s.t.}}

\newcommand{\stp}{\hfill $\square$}

\newcommand{\captext}[1]{\texorpdfstring{#1}{}}

\definecolor{hl-bg-color}{RGB}{255,255,215}
\sethlcolor{hl-bg-color}
\definecolor{new-magenta}{RGB}{255,0,255}
\soulregister\cite7
\soulregister\citep7
\soulregister\citet7
\soulregister\ref7
\soulregister\eqref7

\newtheorem{theorem}{{Theorem}}
\newtheorem{proposition}{{Proposition}}

\newtheorem{lemma}{{Lemma}}
\newtheorem{fact}{{Fact}}
\newtheorem{remark}{{Remark}}

\newtheorem{insight}{{Insight}}
\newtheorem{example}{{Example}}

\newtheorem{method}{{Method}}

\newcommand{\tabincell}[2]{\begin{tabular}{@{}#1@{}}#2\end{tabular}}
\newcommand{\quotemark}[1]{“#1”}

\newcommand{\defeq}{\coloneqq}

\DeclareMathOperator*{\argmax}{argmax}
\DeclareMathOperator*{\argmin}{argmin}

\DeclareMathOperator*{\diag}{diag}

\begin{document}

\newpage

\title{Robust Waveform Design for Integrated Sensing and Communication}

\author{Shixiong Wang, 
Wei Dai,
Haowei Wang,
and Geoffrey Ye Li,~\IEEEmembership{Fellow,~IEEE}
\thanks{

S. Wang, W. Dai, and G. Li are with the Department of Electrical and Electronic Engineering, Imperial College London, London SW7 2AZ, United Kingdom (E-mail: s.wang@u.nus.edu; wei.dai1@imperial.ac.uk; geoffrey.li@imperial.ac.uk).

H. Wang is with Rice-Rick Digitalization, Singapore 308900 (E-mail: haowei\_wang@ricerick.com).
}
\thanks{This work is supported by the UK Department for Science, Innovation and Technology under the Future Open Networks Research Challenge project TUDOR (Towards Ubiquitous 3D Open Resilient Network). 
}
}

\maketitle

\begin{abstract}
Integrated sensing and communication (ISAC), which enables hardware, resources (e.g., spectra), and waveforms sharing, is becoming a key feature in future-generation communication systems.
This paper investigates performance characterization and waveform design for ISAC systems when the underlying true communication channels are not accurately known. 
With uncertainty in a nominal communication channel, the nominal Pareto frontier of the sensing and communication performances cannot represent the true performance trade-off of a real-world operating ISAC system.
Therefore, this paper portrays the robust (i.e., conservative) Pareto frontier considering the uncertainty in the communication channel. 
To be specific, the lower bound of the true (but unknown) Pareto frontier is investigated, technically by studying robust waveform design problems that find the best waveforms under the worst-case channels. 
The robust waveform design problems examined in this paper are shown to be non-convex and high-dimensional, which cannot be solved using existing optimization techniques. 
As such, we propose a computationally efficient solution framework to approximately solve them. 
Simulation results show that by solving the robust waveform design problems, the lower bound of the true but unknown Pareto frontier, which characterizes the sensing-communication performance trade-off under communication channel uncertainty, can be obtained.
\end{abstract}


\section{Introduction}\label{sec:introdction}
\IEEEPARstart{I}{{ntegrated}} sensing and communication (ISAC) is one of the enabling technologies for the sixth-generation (6G) communications. It uses one single hardware system to simultaneously realize the sensing and communication functions. This integration is able to improve spectral efficiency, reduce platform size, and control power consumption. From the signal processing perspective, one of the main features of ISAC is that the same transmitted waveform, called dual-functional waveform, is used for both sensing and communication functions \cite{sturm2011waveform,liu2020joint,zhang2021overview,liu2022survey,zhou2022integrated}. 
This paper is concerned with performance characterization and waveform design for ISAC systems.

\subsection{Performance Characterization and Waveform Design}
Performance characterization for ISAC contains two connotations: 1) to evaluate the performances of an ISAC system, in terms of sensing and communication, at given dual-functional waveforms; 2) to explore the performance boundary by depicting the Pareto frontier of sensing versus communication. Typical performance measures for sensing include the Cramér--Rao bounds for the estimated parameters \cite{liu2022survey,ferguson1996course} and the sensing mutual information between received signals and sensing channels (i.e., sensing rate) \cite{ouyang2022performance,ouyang2023integrated,xie2024sensing}, while those for communication include the distortion minimum mean-squared error \cite{kumari2019adaptive} and the mutual information between the transmitted and the received signals (i.e., channel capacity) \cite{liu2022survey,cover2006elements}. Waveform design for ISAC, as a technical aspect of performance characterization, aims to depict the Pareto frontier (i.e., to achieve the performance boundary) through finding optimal dual-functional waveforms at different trade-off levels because an optimal waveform for sensing is not necessarily optimal for communication \cite{liu2022survey,ouyang2023integrated,xiong2023fundamental}, and vice versa. Supposing that an ISAC system aims to minimize the sensing metric while maximizing the communication metric, e.g., the Cramér--Rao bound for sensing and the achievable sum-rate for communication, the performance Pareto frontier is shown in Fig. \ref{fig:pareto-frontiers}; cf., e.g., \cite[Fig.~3]{xiong2023fundamental}. Every point on the Pareto frontier is associated with at least one optimal dual-functional waveform that balances the performances between sensing and communication.
\begin{figure}[htbp]
	\centering
	\includegraphics[height=3.55cm]{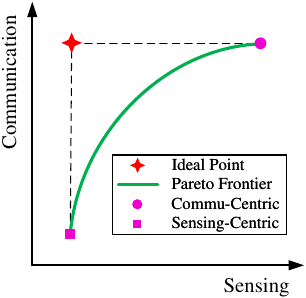}
	\caption{The Pareto frontier represents the performance trade-off between sensing and communication. For a given ISAC system and specified sensing-communication performance measures, the Pareto frontier is fixed, which is obtained through waveform design. For communication-centric waveform design, the communication performance reaches its best level, while for sensing-centric waveform design, the sensing performance arrives at its limit.}
	\label{fig:pareto-frontiers}
\end{figure}

According to the performance priority, as well as the underlying system and signal features, existing ISAC waveform design methods can be categorized into the following three classes \cite{mishra2019toward,zhang2021overview}, \cite[Fig.~1]{zhou2022integrated}.
\begin{itemize}
    \item \textit{Communication-centric waveform design}: In this class, the communication performance is guaranteed to be (nearly) perfect, with a secondary focus on optimizing the sensing performance. Typically, communication base stations and waveforms are used to realize the sensing function because the sensing information can be extracted from communication electromagnetic waves bounced back from environments and targets. While the communication performance can be virtually unaffected, the sensing performance may be limited and difficult to further enhance. Classical representatives of this category encompass the orthogonal frequency division multiplexing (OFDM) waveforms \cite{aditya2022sensing,wei2023waveform}, the orthogonal time frequency space (OTFS) waveforms \cite{gaudio2020effectiveness}, etc.
    
    \item \textit{Sensing-centric waveform design}: In this class, the sensing performance is guaranteed to be (nearly) perfect, with a secondary focus on optimizing the communication performance. Typically, sensing base stations (e.g., radar) and waveforms are used to realize the communication function; to be specific, communication information is modulated onto sensing waveforms. Since these waveforms are (almost) optimal for sensing, the sensing performance can remain approximately uninfluenced, however, the communication performance may be restricted; for example, the communication data rate cannot be sufficiently high. Classical representatives of this category include the chirp waveforms \cite{ouyang2016orthogonal}, the index-modulation-based waveforms \cite{huang2020majorcom}, etc; see also, e.g., \cite{ahmed2023sensing}.
    
    \item \textit{Joint waveform design}: To fully balance the communication and sensing performances, i.e., to eliminate the inherent performance limits on either of them, ISAC systems can be independently designed and realized without relying on existing hardware and waveforms; see, e.g., \cite{liu2018toward,xiong2023fundamental}. For an extensive review, see, e.g., \cite[Sec.~V]{zhang2021overview}, \cite[Sec.~IV]{zhou2022integrated}.
\end{itemize}
Intuitively, joint waveform design portrays the Pareto frontier in Fig. \ref{fig:pareto-frontiers}, while communication-centric and sensing-centric designs delineate the two extreme points on the frontier. In terms of design methodology, the existing methods can also be categorized into two main streams. The first stream directly designs the dual-functional waveforms through optimizing and constraining various performance measures \cite{xiong2023fundamental,guo2023bistatic,wei2023waveform}, while the second stream finds the balanced waveforms through similarity matching with known benchmark waveforms \cite{liu2018toward}.

\subsection{Problem Statement, Research Aims, and Related Works}\label{subsec:problem-statement}
The existing ISAC waveform design methods face the following major drawback: In practice, the \textbf{true} communication channel is unknown, and only the \textbf{nominal} (i.e., estimated or approximated) channel is available. This point can be understood from three aspects. 
\begin{itemize}
    \item First, the communication channel is time-selective (i.e., time-varying) in a frame. 
    \item Second, the channel model used might be inexact. 
    \item Third, any statistical method cannot give exact channel estimation when training samples are limited. 
\end{itemize}
As a consequence, the nominal Pareto frontier obtained under the nominal communication channel cannot characterize the actual sensing-communication performance trade-off; see Fig. \ref{fig:true-frontiers-a}. To be specific, the true Pareto frontier associated with the actual communication channel may lie above, across, or below the nominal Pareto frontier. When the actual channel varies over time, the true Pareto frontier may even not remain fixed. Since the true channel is unknown and even random (e.g., time-varying), characterizing the upper and lower bounds of the actual Pareto frontier(s) is of natural importance and significance; see Fig. \ref{fig:true-frontiers-b}. However, such uncertainty-aware performance characterization has not been discussed in the existing ISAC literature.

\begin{figure}[htbp]
	\centering
	\subfigure[Nominal and true frontiers]{
	 	\includegraphics[height=3.55cm]{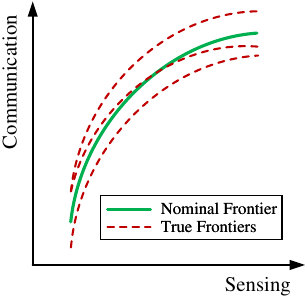}
        \label{fig:true-frontiers-a}
	}
	\subfigure[Bounds of true frontiers]{
	 	\includegraphics[height=3.55cm]{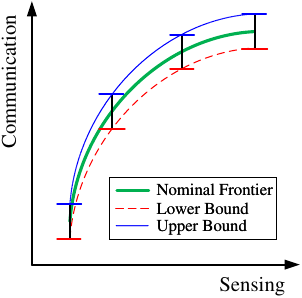}
        \label{fig:true-frontiers-b}
	}
	\caption{Pareto frontiers of sensing-communication performances. The true Pareto frontier is unknown and even time-varying so the nominal frontier cannot characterize the actual performance trade-off; NB: every realization of the random communication channel corresponds to one true Pareto frontier. However, we can identify the lower and upper bounds of the true frontier(s).}
	\label{fig:true-frontiers}
\end{figure}

As an initiative work on ISAC performance characterization under communication channel uncertainty, this paper focuses on exploring the lower bound of the actual Pareto frontier(s), that is, the \textbf{robust} (i.e., \textbf{conservative}) performance boundary. To achieve such conservativeness, robust waveform design methods that utilize the worst-case information of the communication channel are leveraged.

Literature on robust dual-functional waveform design for ISAC is rather lacking, and the highly related records include only \cite{liao2023robust} and \cite{li2023optimal}. In \cite{liao2023robust}, a robust beamforming method for ISAC that handles the uncertainties in the communication channel and the target angles is studied, which, however, limits the ISAC waveforms to linear combinations of the communication symbols. Another ISAC waveform design work considering the uncertainty in target angles is \cite{li2023optimal}, which, however, does not deal with the uncertainty in the communication channel. Other weakly-related works considering imperfect channel information encompass, e.g., robust beamforming for ISAC \cite{liu2017robust,zhao2022joint,ren2023robust,luan2023robust}, which, as in \cite{liao2023robust}, require the ISAC waveforms to be linear transforms of the communication symbols. For general ISAC waveform design, nevertheless, the technical methodology is beyond linear beamforming; see \cite{liu2018toward,guo2023bistatic}. The main issue of existing works on ISAC robust waveform design is that they do not reveal the relation between conservative performance characterization (i.e., lower bound discovery) and robust waveform design; cf. Fig. \ref{fig:true-frontiers}. Technically, this is because, facing the uncertainty in the communication channel, these works focus on guaranteeing the satisfaction of the worst-case communication performance rather than optimizing it; see, e.g., \cite[Eq.~(7)]{liao2023robust}, \cite[Eq.~(17b)]{zhao2022joint}.

\begin{remark}[Specificity of Robust Design]
For any given existing ISAC waveform design method, in principle, we can propose a corresponding robust counterpart that is insensitive to the potential uncertainties in the involved parameters and/or models. That is, the mathematical formulations of robust counterparts are specific to the adopted original (i.e., non-robust) ISAC waveform design frameworks because different waveform design philosophies give different waveform design formulations; cf. \cite{liu2018toward,zhang2021overview,zhou2022integrated}. 
To develop this paper, as an illustration of uncertainty-aware performance characterization and robust design, we exemplify the waveform design framework proposed in \cite{liu2018toward}, where balanced waveforms that are simultaneously close to the optimal waveforms for sensing and communication are studied. The future papers following this work can consider robustifying other existing ISAC waveform design methods. 
\stp
\end{remark}

\subsection{Contributions}
In line with the basic system setups in \cite{liu2018toward}, the contributions of this paper can be highlighted as follows:
\begin{enumerate}
    \item A robust dual-functional waveform design framework that combats the uncertainty of the communication channel is proposed; see Subsection \ref{subsec:robust-design}. The framework is established leveraging the min-max optimization, where the minimization of the objective is to design the optimal waveforms and the maximization of the objective is to find the worst-case information of the communication channel; see Problems \eqref{eq:robust-waveform-design} and \eqref{eq:robust-waveform-design-tradeoff}.

    \item The concept of ISAC performance characterization under channel uncertainty is proposed; see Fig. \ref{fig:true-frontiers}. We show that finding the lower bound of the true Pareto frontier(s), i.e., the conservative performance boundary, is related to solving the robust waveform design problems; see Section \ref{sec:robust-formulation}, especially Facts \ref{fact:upperbound} and \ref{fact:unachievability}, and Remarks \ref{rem:robustness-benefit} and \ref{rem:robustness-benefit-AASR}.
    
    \item To tackle the two-layer min-max robust waveform design problems \eqref{eq:robust-waveform-design} and \eqref{eq:robust-waveform-design-tradeoff}, an approximate solution framework is proposed; see Section \ref{sec:solve-robust-problem}, especially Lemmas \ref{lem:max-form} and \ref{lem:gap}, Insight \ref{insight:approx-solution}, Theorems \ref{thm:max-form} and \ref{thm:max-form-tradeoff}, Methods \ref{method:max-form-solution}, \ref{method:max-form-solution-another}, and \ref{method:robust-joint-design}, and Table \ref{tab:solution-method}.
\end{enumerate}
The first two contributions are new to the ISAC community, whereas the third contribution also enriches the general min-max optimization theory and practice (e.g., robustness theory, game theory, etc.).

\subsection{Notations}
Vectors (column by default) and matrices are denoted by bold symbols and are written in lowercase and uppercase letters, respectively; e.g., $\vec x$ is a vector while $\mat X$ is a matrix. Deterministic quantities are denoted by Italic symbols (e.g., $\vec x$, $\mat X$) while random quantities are denoted by normal symbols (e.g., $\rvec x$, $\rmat X$). Sets are denoted by calligraphic symbols, e.g., $\cal X$. Let $\C^d$, $\R^d$ stand for a $d$-dimensional complex, real space, respectively. Let $\mat X^\T$, $\mat X^\H$, $\Tr[\mat X]$, $\|\mat X\|$, $\operatorname{vec}(\mat X)$, $\operatorname{diag}(\mat X)$, and $\mat X^{-1}$ denote the transpose, the conjugate transpose, the trace, a norm (e.g., the Frobenius norm $\|\mat X\|_F$), the column vectorization, the diagonal entries, and the inverse of the matrix $\mat X$, respectively. When the inverse of $\mat X$ does not exist, we use $\mat X^{-1}$ to denote its Moore--Penrose pseudo-inverse. We use $\cal{CN}(\vec 0, \mat \Sigma)$ to denote the circularly-symmetric complex normal distribution with zero mean, covariance matrix $\mat \Sigma$, and zero pseudo-covariance matrix. Let the running index set induced by an integer $I$ be denoted by $[I] \defeq \{1,2,\ldots, I\}$ and the $N$-dimensional identity matrix by $\mat I_N$.

\section{Problem Formulation}\label{sec:problem-formulation}

\subsection{Basic Setups and Traditional Waveform Design}
Following the basic setups in \cite{liu2018toward}, we consider an ISAC system with an $N$-antenna uniform linear array (ULA). The system simultaneously senses targets and communicates with $K$ single-antenna downlink users using shared dual-functional waveform $\mat X \in \C^{N \times L}$, where $L$ denotes the length of the communication frame (from the perspective of communication) and the number of snapshots (from the perspective of sensing). Letting $\vec x_l \in \C^N$ denote the transmitted $N$-dimensional signal vector at time stamp $l \in [L]$, we have $\mat X \defeq [\vec x_1,\vec x_2,\ldots,\vec x_l,\ldots,\vec x_L]$.

On the communication side, $\mat X$ is the output of pre-coding and conveys the information of communication symbols. For downlink communication with $K$ users, the base-band communication model is
\begin{equation}\label{eq:communication-model}
    \rmat Y = \mat H \mat X + \rmat W,
\end{equation}
where $\rmat Y \in \C^{K \times L}$ denotes the received symbol matrix at $K$ users, $\mat H \defeq [\vec h_1, \vec h_2, \ldots, \vec h_K]^\T \in \C^{K \times N}$ the communication channel matrix, and $\rmat W \defeq [\rvec w_1, \rvec w_2, \ldots, \rvec w_L] \in \C^{K \times L}$ the channel noise; for every $k \in [K]$, $\vec h_k$ denotes the channel model of the $k^\th$ user; for every $l \in [L]$, $\rvec w_l \sim \cal{CN}(\vec 0, N_0 \mat I_K)$ denotes the channel noise at the $l^\th$ data unit and $N_0$ denotes the noise level. Note that given a waveform matrix $\mat X$ and a channel $\mat H$, the randomness of $\rmat Y$ is from the randomness of $\rmat W$. Suppose that the expected constellation-point matrix for the $K$ downlink users is $\mat S \in \C^{K \times L}$. We have the multi-user interference (MUI) signals $\mat H \mat X - \mat S$ and the total MUI energy can be calculated as $\|\mat H \mat X - \mat S\|^2_F$; see \cite{liu2018toward,bazzi2023integrated}. Therefore, from the perspective of communication, waveform design is to find a mapping $\cal F_{\mat H}:\mat S \to \mat X$ such that the total MUI energy $\|\mat H \cdot \cal F_{\mat H} (\mat S) - \mat S\|^2_F$ is minimized, where $\cal F_{\mat H}$ depends on the communication channel information $\mat H$. That is, the \textbf{perfect communication} problem can be stated as 
\begin{equation}\label{eq:perfect-communication}
    \begin{array}{cl}
      \displaystyle \min_{\mat X}  &  \|\mat H \mat X - \mat S\|^2_F \\
       \st  &  \mat X \text{ satisfies some constraints.}
    \end{array}
\end{equation}
The constraint can be, for example, to limit the total transmit power, i.e., $\|\mat X\|^2_F / L \le P_{\T}$, where $P_{\T}$ denotes the maximum allowed total transmit power. Note that the smaller the MUI energy, the higher the communication performance. To be specific, for a communication system, the average achievable sum-rate (AASR; unit: bps/Hz/user) $R_{\mat H, \mat X}$ is defined as \cite{liu2018toward}
\begin{equation}\label{eq:AASR}
R_{\mat H, \mat X} \defeq \frac{1}{K} \sum_{k=1}^K \log _2\left(1+\gamma_{\mat H, \mat X, k}\right)    
\end{equation}
and the average signal-to-interference-plus-noise ratio (SINR) $\gamma_{\mat H, \mat X, k}$ for the $k^\th$ user is defined as
\begin{equation}\label{eq:SINR}
\gamma_{\mat H, \mat X, k} \defeq 
\frac{\sum^L_{l = 1} \left|s_{k, l}\right|^2 / L}{ \sum^L_{l = 1} \left|\vec h^\T_k \vec x_l - s_{k, l}\right|^2 / L + N_0},
\end{equation}
where $s_{k,l}$ is the $(k,l)$-entry of the constellation $\mat S$. As a result, minimizing MUI energy $\|\mat H \mat X - \mat S\|^2_F$ in \eqref{eq:perfect-communication} leads to the improvement of AASR in \eqref{eq:AASR} because for a power-fixed constellation $\mat S$, the numerator in \eqref{eq:SINR} is fixed.

\begin{remark}
The extensions of \eqref{eq:perfect-communication} to the case of multi-antenna users and the case of multi-carrier are technically trivial. For details, see Appendix A in supplementary materials. 
 \stp
\end{remark}

On the sensing side, the key is to design perfect beam patterns for target searching and/or tracking. According to \cite{fuhrmann2008transmit}, designing a beam pattern is equivalent to designing the cross-correlation matrix $\mat R \in \C^{N \times N}$ of the probing signals $\{\vec x_l\}_{l \in [L]}$. Therefore, from the perspective of sensing, waveform design is to find a mapping $\cal F: \mat R \to \mat X$ such that $\mat X \mat X^\H/L = \cal F(\mat R) \cal F^\H(\mat R)/L = \mat R$. For an omnidirectional beam pattern, 
which is suitable for target searching, 
$\mat R \defeq \frac{P_\T}{N} \mat I_N$ should be used. Therefore, the \textbf{perfect sensing} problem can be stated as 
\begin{equation}\label{eq:perfect-sensing}
    \begin{array}{cl}
    \operatorname{find}  &  \mat X \\
       \st  & \mat X \mat X^\H = L \mat R \\
            & \mat X \text{ satisfies some constraints.}
    \end{array}
\end{equation}
The constraint can be, for example, to limit the peak-to-average-power ratio (PAPR) \cite[Eq.~(18)]{stoica2008waveform}.

For an ISAC system, \textbf{dual-functional waveform design} can be stated as finding a mapping $\cal F_{\mat H}: (\mat S, \mat R) \to \mat X$ such that the total MUI energy $\|\mat H \cdot \cal F_{\mat H} (\mat S, \mat R) - \mat S\|^2_F$ is minimized and the cross-correlation matrix $\cal F_{\mat H} (\mat S, \mat R) \cal F^\H_{\mat H} (\mat S, \mat R) / L$ is as close as possible to the desired $\mat R$.
In \cite[Eq. (10)]{liu2018toward}, the following \textbf{sensing-centric} ISAC waveform design problem is defined:
\begin{equation}\label{eq:waveform-design}
    \begin{array}{cl}
       \displaystyle \min_{\mat X}  &  \|\mat H \mat X - \mat S\|^2_F \\
       \st  & \displaystyle \frac{1}{L} \mat X \mat X^\H = \mat R,
    \end{array}
\end{equation}
which aims to minimize the communication MUI energy while guaranteeing the perfect sensing performance. If we desire a balanced performance between communication and sensing, we can consider a \textbf{joint design} problem:
\begin{equation}\label{eq:joint-design}
    \begin{array}{cl}
     \displaystyle  \min_{\mat X}  & \rho \|\mat H \mat X - \mat S\|^2_F + (1 - \rho) \| \mat X \mat X^\H - L \mat R \|^2_F \\
       \st  &  \mat X \text{ satisfies some constraints,}
    \end{array}
\end{equation}
where $\rho \in [0, 1]$ is the trade-off coefficient. Problem \eqref{eq:joint-design} is technically difficult to solve because the objective function is quartic in terms of the waveform matrix $\mat X$. Therefore, we need to formulate simplifications. In \cite{liu2018toward}, the alternative objective function $\rho \|\mat H \mat X - \mat S\|^2_F  + (1 - \rho) \|\mat X - \mat X_s\|^2_F$ is intensively considered, which relaxes \eqref{eq:joint-design} to
\begin{equation}\label{eq:waveform-design-tradeoff}
    \begin{array}{cl}
     \displaystyle  \min_{\mat X}  & \rho \|\mat H \mat X - \mat S\|^2_F + (1 - \rho) \| \mat X - \mat X_s \|^2_F \\
       \st  &  \mat X \text{ satisfies some constraints,}
    \end{array}
\end{equation}
where $\mat X_s$ is a desired sensing waveform, for example, the waveform in \eqref{eq:perfect-sensing} for perfect sensing regardless of communications. Compared with \eqref{eq:waveform-design}, the hard constraint $\mat X \mat X^\H = L \mat R$ for perfect sensing is relaxed to penalizing the mismatch $\mat X \mat X^\H - L \mat R$ in \eqref{eq:joint-design} and is further relaxed to penalizing the discrepancy $\mat X - \mat X_s$ in \eqref{eq:waveform-design-tradeoff} due to the technical tractability in developing the solution method, where $\mat X_s \mat X^\H_s = L \mat R$. Note that the objective of \eqref{eq:waveform-design-tradeoff} is quadratic (not quartic) in $\mat X$.

\subsection{Robust Waveform Design}\label{subsec:robust-design}
\subsubsection{Robust Counterpart of Sensing-Centric Design \captext{\eqref{eq:waveform-design}}}
Considering the modeling errors in the communication channel $\mat H$, the robust counterpart of \eqref{eq:waveform-design} can be formulated as
\begin{equation}\label{eq:robust-waveform-design}
    \begin{array}{cl}
       \displaystyle \min_{\mat X} \max_{\mat H}  &  \|\mat H \mat X - \mat S\|^2_F \\
       \st  & \displaystyle \frac{1}{L} \mat X \mat X^\H = \mat R\\
            & \|\mat H - \bar{\mat H}\| \le \theta,
    \end{array}
\end{equation}
where $\bar{\mat H}$ is a nominal value (e.g., an estimate) of the true channel matrix, $\|\cdot\|$ denotes any suitable matrix norm (e.g., the Frobenius norm $\|\cdot\|_F$), and $\theta \ge 0$ is our trust level of the nominal $\bar{\mat H}$; the smaller the $\theta$, the more we trust $\bar{\mat H}$. Problem \eqref{eq:robust-waveform-design} is non-convex in $\mat X$ and non-concave in $\mat H$. 

\subsubsection{Robust Counterpart of Joint Design \captext{\eqref{eq:waveform-design-tradeoff}}}
Similarly, the robust counterpart of \eqref{eq:waveform-design-tradeoff} is
\begin{equation}\label{eq:robust-waveform-design-tradeoff}
    \begin{array}{cl}
     \displaystyle  \min_{\mat X} \max_{\mat H}  & \rho \|\mat H \mat X - \mat S\|^2_F + (1 - \rho) \| \mat X - \mat X_s \|^2_F \\
       \st  &  \mat X \text{ satisfies some constraints} \\
       & \|\mat H - \bar{\mat H}\| \le \theta,
    \end{array}
\end{equation}
where the constraints of $\mat X$ can specifically be
\begin{equation}\label{eq:total-power}
    \frac{1}{L} \|\mat X\|^2_F = P_\T,
\end{equation}
for the total power constraint (TPC), or
\begin{equation}\label{eq:per-antenna-power}
    \diag(\mat X \mat X^\H) = \frac{L \cdot P_\T}{N} \mat I_N,
\end{equation}
for the per-antenna power constraint (PAPC). Problem \eqref{eq:robust-waveform-design-tradeoff}, particularized by either \eqref{eq:total-power} or \eqref{eq:per-antenna-power}, is non-convex in $\mat X$ and non-concave in $\mat H$.

\subsubsection{Remarks on Robust Counterparts}
From the viewpoint of optimization theory, the min-max formulations \eqref{eq:robust-waveform-design} and \eqref{eq:robust-waveform-design-tradeoff} are new and technically challenging in their own right. In the authors' knowledge, no literature working on similar optimization problems can be found. The only related research is \cite{ahmed2012relaxation}, which somehow has a subtle connection with \eqref{eq:robust-waveform-design}, however, did not solve the same problem.

Suppose that the communication performance of the investigated ISAC system is measured by AASR. To obtain the lower bound of AASR in Fig. \ref{fig:true-frontiers-b}, we need to find the upper bound of the MUI energy, which justifies why we maximize the MUI energy $\|\mat H \mat X - \mat S\|^2_F$ in \eqref{eq:robust-waveform-design} and \eqref{eq:robust-waveform-design-tradeoff} over the admissible channels $\mat H$. For detailed analyses on min-max modeling, see Section \ref{sec:robust-formulation}.

\section{Robust Waveform Design: Property Analyses}\label{sec:robust-formulation}

This section studies the rationale behind the min-max robust formulations \eqref{eq:robust-waveform-design} and \eqref{eq:robust-waveform-design-tradeoff}. Without loss of generality, we exemplify using \eqref{eq:robust-waveform-design}; the analysis on \eqref{eq:robust-waveform-design-tradeoff} is similar, and therefore, omitted. 

\subsection{Interpretation of Min-Max Formulation \captext{\eqref{eq:robust-waveform-design}}}\label{sec:model-interpretation}
The min-max robust formulation \eqref{eq:robust-waveform-design} can be seen as a game between a waveform designer selecting the waveform $\mat X$ and a fictitious adversary selecting the channel state $\mat H$; the adversary tries to deteriorate the communication performance by lifting the MUI energy $\|\mat H \mat X - \mat S\|^2_F$, while the waveform designer tries to improve the communication performance by diminishing the MUI energy. Therefore, min-max robust waveforms perform best under worst-possible channel states.

Let the true communication channel be $\mat H_0$. According to \eqref{eq:waveform-design}, the problem that we really want to solve should be
\begin{equation}\label{eq:waveform-design-true}
  \min_{\mat X:~\mat X \mat X^\H = L \mat R}  \|\mat H_0 \mat X - \mat S\|^2_F
\end{equation}
which, however, cannot be conducted in practice because $\mat H_0$ is unknown. Nevertheless, we can assume that $\mat H_0$ is included in an \textbf{uncertainty set} 
\begin{equation}\label{eq:uncertainty-set}
\cal H \defeq \{\mat H:~\| \mat H - \bar{\mat H}\| \le \theta\}
\end{equation}
for some matrix norm $\|\cdot\|$ where $\bar{\mat H}$ is an estimate of $\mat H_0$ and $\theta$ denotes the trust level; we suppress the dependence of $\cal H$ on $(\bar{\mat H}, \theta, \|\cdot\|)$ to avoid notational clutter. This assumption is practically reasonable because the estimate $\bar{\mat H}$ of the ground truth $\mat H_0$ can be close to $\mat H_0$. In addition, we denote $\cal X$ the collection of all waveforms for perfect sensing, i.e., 
\begin{equation}\label{eq:feasible-region}
    \cal X \defeq \{\mat X:\mat X \mat X^\H = L \mat R\}.
\end{equation}

Facing the unavailability of $\mat H_0$, the pragmatic strategy is to find a practically available upper bound $\mathcal{UB}(\mat X)$ for the true cost function $\|\mat H_0 \mat X - \mat S\|^2_F$ for every $\mat X \in \cal X$, i.e.,
\begin{equation}\label{eq:upper-bound-inequality}
\|\mat H_0 \mat X - \mat S\|^2_F \le \mathcal{UB}(\mat X),~~~\forall \mat X \in \cal X.
\end{equation}
As a result, when the upper bound $\mathcal{UB}(\mat X)$ is minimized by some $\mat X^*$, the \textbf{true cost} of the waveform $\mat X^*$ evaluated at the true channel state $\mat H_0$, i.e.,
\[
\|\mat H_0 \mat X^* - \mat S\|^2_F,
\]
which defines the true communication performance in practical operation, is also reduced. In consideration of the uncertainty set $\cal H$, $\mathcal{UB}(\mat X)$ can be constructed as
\begin{equation}\label{eq:upper-bound}
\mathcal{UB}(\mat X) \defeq \displaystyle \max_{\mat H \in \cal H} \|\mat H \mat X - \mat S\|^2_F,~~~ \forall \mat X \in \cal X.
\end{equation}
The upper bound $\mathcal{UB}(\mat X)$ above is tight in the sense that when the uncertainty set $\cal H$ contains only $\mat H_0$ (i.e., when $\theta = 0$), the equality in \eqref{eq:upper-bound-inequality} holds. In other words, provided that $\mat H_0 \in \cal H$, the more accurate the $\bar{\mat H}$ (i.e., the smaller the $\theta$), the tighter the upper bound \eqref{eq:upper-bound}.

The fact below explains the rationale of the robust waveform design model \eqref{eq:robust-waveform-design}.

\begin{fact}[Achievability of Robust Estimate]\label{fact:upperbound}
By minimizing the practically-available upper bound $\max_{\mat H \in \cal H} \|\mat H \mat X - \mat S\|^2_F$, we can also upper bound the \textbf{truly optimal cost} 
\begin{equation}\label{eq:truly-opt-cost}
\min_{\mat X \in \cal X} \|\mat H_0 \mat X - \mat S\|^2_F
\end{equation}
and reduce the \textbf{true cost} of the \textbf{robust waveform} $\mat X^*$ evaluated at the true channel $\mat H_0$, i.e.,
\begin{equation}\label{eq:true-cost-at-robust}
\|\mat H_0 \mat X^* - \mat S\|^2_F
\end{equation} 
because
\[
\min_{\mat X \in \cal X} \|\mat H_0 \mat X - \mat S\|^2_F \le \|\mat H_0 \mat X^* - \mat S\|^2_F  
\le \|\mat H^* \mat X^* - \mat S\|^2_F,
\]
where 
\[
(\mat X^*, \mat H^*) \in \argmin_{\mat X \in \cal X} \argmax_{\mat H \in \cal H} \|\mat H \mat X - \mat S\|^2_F
\] 
and $\mat H^*$ is the \textbf{worst-case channel state} associated with $\mat X^*$. Note that in real-world operation, the true communication performance of a specified waveform $\mat X$ is defined by the true cost $\|\mat H_0 \mat X - \mat S\|^2_F$ evaluated at $\mat H_0$; cf. \eqref{eq:true-cost-at-robust}.
\stp
\end{fact}

In contrast, the nominally optimal waveform(s) cannot provide such a performance guarantee.

\begin{fact}[Unachievability of Nominal Estimate]\label{fact:unachievability}
The \textbf{nominally optimal cost} 
\begin{equation}\label{eq:nominally-opt-cost}
    \min_{\mat X \in \cal X} \|\bar{\mat H} \mat X - \mat S\|^2_F
\end{equation}
\textbf{cannot} upper bound the \textbf{truly optimal cost} 
\[
\min_{\mat X \in \cal X} \|\mat H_0 \mat X - \mat S\|^2_F
\]
and reduce the true cost of the \textbf{nominally optimal waveform} $\bar{\mat X}$ evaluated at the true channel $\mat H_0$, i.e.,
\begin{equation}\label{eq:true-cost-at-nominal}
\|\mat H_0 \bar{\mat X} - \mat S\|^2_F
\end{equation}
 because
\[
\min_{\mat X \in \cal X} \|\mat H_0 \mat X - \mat S\|^2_F \le \|\mat H_0 \bar{\mat X} - \mat S\|^2_F
\stackrel{\text{\scriptsize ?}}{\le} \|\bar{\mat H} \bar{\mat X} - \mat S\|^2_F,
\]
where $\bar{\mat X} \in \argmin_{\mat X \in \cal X} \|\bar{\mat H} \mat X - \mat S\|^2_F$ and the symbol $\stackrel{\text{\scriptsize ?}}{\le}$ means that the relation is not guaranteed. \stp
\end{fact}

Let $\mat X_0 \in \argmin_{\mat X \in \cal X} \|\mat H_0 \mat X - \mat S\|^2_F$ denote a \textbf{truly optimal waveform}. The motivation and the necessity of the robust waveform design are summarized in the remark below. 

\begin{remark}[Benefit of Robust Waveform Design]\label{rem:robustness-benefit}
In the practice of communications, finding the tight upper bound of the truly optimal cost $\|\mat H_0 \mat X_0 - \mat S\|^2_F$ and the true cost $\|\mat H_0 \mat X - \mat S\|^2_F$ of the specified waveform $\mat X$ is critical. Suppose that an ISAC system can tolerate at most $C$ MUI energies for communications: i.e., for an employed waveform $\mat X^\prime$, we must have $\|\mat H_0 \mat X^\prime - \mat S\|^2_F \le C$. However, if we use the nominally optimal waveform $\bar{\mat X}$ designed at $\bar{\mat H}$, the designed nominally optimal cost $\|\bar{\mat H} \bar{\mat X} - \mat S\|^2_F \le C$ does not necessarily imply the true cost $\|\mat H_0 \bar{\mat X} - \mat S\|^2_F \le C$. This may lead to significant dissatisfaction (i.e., bad user experiences) among communication users because the announced system performance $C$ cannot be guaranteed, and therefore, a serious reliability issue in communications arises. In contrast, if we employ the robust waveform design in \eqref{eq:robust-waveform-design}, the dissatisfaction issue can be fixed because the robust cost $\|\mat H^* \mat X^* - \mat S\|^2_F \le C$ can imply the true cost $\|\mat H_0 \mat X^* - \mat S\|^2_F \le C$; recall Fact \ref{fact:upperbound}.
\stp
\end{remark}

To further clarify the benefit of robust design, in addition to Remark \ref{rem:robustness-benefit}, the remark below exemplifies the average achievable sum-rate defined in \eqref{eq:AASR}. 
\begin{remark}\label{rem:robustness-benefit-AASR}
Given a communication system \eqref{eq:communication-model}, the performance is typically characterized using the AASR $R_{\mat H, \mat X}$, which is specific to the channel state $\mat H$ and the waveform $\mat X$. When the true channel state $\mat H_0$ is exactly known, the communication system design is to find the optimal waveform $\mat X_0$ such that
\[
R_{\mat H_0, \mat X} \le R_{\mat H_0, \mat X_0},~~~\forall \mat X \in \cal X.
\]
As a result, the truly optimal AASR $R_{\mat H_0, \mat X_0}$ is the performance frontier of this communication system; cf. Fig. \ref{fig:pareto-frontiers}. In practice, $\mat H_0$ is unknown, and the nominal channel $\bar{\mat H}$ acts as a surrogate of $\mat H_0$. However, the nominally optimal AASR $R_{\matb H, \matb X}$ cannot serve as a performance indicator of this communication system because
\[
R_{\matb H, \matb X} \stackrel{\text{\scriptsize ?}}{\le} R_{\mat H_0, \matb X} \le R_{\mat H_0, \mat X_0};
\]
cf. Fig. \ref{fig:true-frontiers-a}. In contrast, when the robust waveform $\mat X^*$ is utilized, we have 
\[
R_{\mat H^*, \mat X^*} \le R_{\mat H_0, \mat X^*} \le R_{\mat H_0, \mat X_0},
\]
which means that $R_{\mat H^*, \mat X^*}$ defines the worst-case performance bound of this communication system: that is, in real-world operations, the system's true running performance is no worse than $R_{\mat H^*, \mat X^*}$. Hence, $R_{\mat H^*, \mat X^*}$ characterizes the robust or conservative performance frontier in terms of communication; cf. Fig. \ref{fig:true-frontiers-b}.
\stp
\end{remark}

\subsection{Price of Robustness}\label{subsec:price-robustness}
As we can see from Fact \ref{fact:upperbound}, the robust design specifies an upper bound for the truly optimal cost and the true cost. In practice, however, this upper bound may be overly loose, i.e., 
\[
\min_{\mat X \in \cal X} \|\mat H_0 \mat X - \mat S\|^2_F \le \|\mat H_0 \mat X^* - \mat S\|^2_F  
\ll \|\mat H^* \mat X^* - \mat S\|^2_F.
\]
To be specific, the robust waveform $\mat X^*$ that solves the min-max robust problem may induce a loose estimate of the truly optimal cost and the true cost, and therefore, the optimality of the communication performance might be compromised. This is because the worst-case channel(s) does not necessarily frequently occur in real-world operations. As a result, the robustness-optimality trade-off is raised: to obtain robustness under uncertain conditions (i.e., when $\mat H_0$ is not exactly known), the price to pay is to sacrifice optimality under perfect conditions (i.e., when $\mat H_0$ is perfectly known).

\subsection{Budget of Uncertainty Set}\label{subsec:budget-uncertainty-set}
To reduce the conservativeness of the robust design, a practical trick is to limit the size of the uncertainty set $\cal H$. To be specific, we employ a \textbf{shrunken uncertainty set} $\cal H_\beta \defeq \{\mat H:~\| \mat H - \bar{\mat H}\| \le \beta \theta\}$ where $\beta \in [0,1]$ is called the \textbf{budget of the uncertainty set}. Since
\[
\min_{\mat X \in \cal X} \max_{\mat H \in \cal H_\beta} \|\mat H \mat X - \mat S\|^2_F \le \min_{\mat X \in \cal X} \max_{\mat H \in \cal H} \|\mat H \mat X - \mat S\|^2_F,
\]
a practically tighter upper bound (i.e., the left-hand side of the above display) can be suggested, which however \textbf{cannot} theoretically serve as an upper bound for the truly optimal cost and the true cost. The motivation is that in practice, $\mat H_0$ is still included in the shrunken set $\cal H_\beta$ with high probability. Nevertheless, the price is to meet the worst-case: $\mat H_0$ might be outside of $\cal H_\beta$. The fact below practically lifts Fact \ref{fact:upperbound}.
\begin{fact}[Achievability in Probability]\label{fact:bound-in-prob}
    Suppose that the true channel $\mat H_0$ is included in $\cal H_\beta$ with probability $\eta$ (e.g., $\eta = 95\%$). Then the robust cost $\min_{\mat X \in \cal X} \max_{\mat H \in \cal H_\beta} \|\mat H \mat X - \mat S\|^2_F$ upper bounds the truly optimal cost and the true cost with probability $\eta$; cf. Fact \ref{fact:upperbound}. \stp
\end{fact}

\section{Robust Waveform Design: Solution Methods}\label{sec:solve-robust-problem}
This section studies the solution methods to the robust counterparts \eqref{eq:robust-waveform-design} and \eqref{eq:robust-waveform-design-tradeoff}. We start with the following preparatory and motivational lemma.

\begin{lemma}\label{lem:max-form}
Let $\cal X$ define a generic feasible region of waveforms. Suppose that $\phi: \cal H \times \cal X \to \R$ is the performance objective function of an ISAC system, e.g., as in \eqref{eq:robust-waveform-design} and \eqref{eq:robust-waveform-design-tradeoff}. For every fixed point $\mat X^* \in \cal X$, we have
\begin{equation}\label{eq:upper-bounds}
\begin{array}{cl}
    \displaystyle \max_{\mat H  \in \cal H} \min_{\mat X \in \cal X}  \phi(\mat H, \mat X) &\le \displaystyle \min_{\mat X \in \cal X} \max_{\mat H  \in \cal H}  \phi(\mat H, \mat X) \\
    &\le \displaystyle \max_{\mat H  \in \cal H} \phi(\mat H, \mat X^*).
\end{array}
\end{equation}
In addition, if there exists an $\mat X^* \in \cal X$ such that 
\begin{equation}\label{eq:upper-bounds-tightness}
\begin{array}{cl}
    \displaystyle \max_{\mat H  \in \cal H} \min_{\mat X \in \cal X}  \phi(\mat H, \mat X) = \displaystyle \max_{\mat H  \in \cal H} \phi(\mat H, \mat X^*),
\end{array}
\end{equation}
then the strong min-max property holds for the min-max game on $(\phi,\cal H,\cal X)$: i.e.,
\begin{equation}\label{eq:min-max-equality}
\displaystyle \max_{\mat H  \in \cal H} \min_{\mat X \in \cal X}  \phi(\mat H, \mat X) = \displaystyle \min_{\mat X \in \cal X} \max_{\mat H  \in \cal H}  \phi(\mat H, \mat X).
\end{equation}
\end{lemma}
\begin{proof}
    See Appendix B in supplementary materials. 
    \stp
\end{proof}

In the context of robustness analysis, the max-min optimization $\max_{\mat H} \min_{\mat X}  \phi(\mat H, \mat X)$ can be easier to solve than the original min-max robust counterpart $\min_{\mat X} \max_{\mat H}  \phi(\mat H, \mat X)$ because for every fixed channel $\mat H$, the solution to the waveform design problem $\min_{\mat X}  \phi(\mat H, \mat X)$ is usually available in literature. Therefore, if the strong mix-max property \eqref{eq:min-max-equality} holds, the robustification problem $\min_{\mat X} \max_{\mat H}  \phi(\mat H, \mat X)$ can be simplified to the problem of finding the worst-case channel through $\max_{\mat H  \in \cal H} \phi(\mat H, \mat X^*_{\mat H})$ where $\mat X^*_{\mat H}$ solves $\min_{\mat X}  \phi(\mat H, \mat X)$ for every $\mat H$.

In practice, the condition \eqref{eq:upper-bounds-tightness} might be harsh to satisfy. A compromise, however, can be made if the gap between $\max_{\mat H} \phi(\mat H, \mat X^*)$ and $\max_{\mat H} \min_{\mat X}  \phi(\mat H, \mat X)$ in \eqref{eq:upper-bounds} is small. The lemma below characterizes this gap in a special case.

\begin{lemma}\label{lem:gap}
Suppose that $\matb X$ solves the nominal problem $\min_{\mat X \in \cal X}  \phi(\matb H, \mat X)$. If the function 
\begin{equation}\label{eq:max-form-abstract}
\mat H \mapsto \min_{\mat X \in \cal X}  \phi(\mat H, \mat X),~~~\forall \mat H \in \cal H
\end{equation}
is $L_1$-Lipschitz continuous in $\mat H$ on $\cal H$ and the function
\begin{equation}\label{eq:upper-bound-abstract}
\mat H \mapsto \phi(\mat H, \matb X),~~~\forall \mat H \in \cal H
\end{equation}
is $L_2$-Lipschitz continuous in $\mat H$ on $\cal H$, both with respect to the norm $\|\cdot\|$ used in defining $\cal H$ in \eqref{eq:uncertainty-set}, then
\[
    \max_{\mat H \in \cal H} \phi(\mat H, \matb X) - \max_{\mat H \in \cal H} \min_{\mat X \in \cal X}  \phi(\mat H, \mat X) \le (L_1 + L_2) \cdot \theta.
\]
\end{lemma}
\begin{proof}
    See Appendix C in supplementary materials. 
    \stp
\end{proof}

Lemma \ref{lem:gap} suggests that, if the functions in \eqref{eq:max-form-abstract} and \eqref{eq:upper-bound-abstract} are Lipschitz continuous, and the radius $\theta$ of $\cal H$ is small, then the gap between $\max_{\mat H} \phi(\mat H, \matb X)$ and $\max_{\mat H} \min_{\mat X}  \phi(\mat H, \mat X)$ would be naturally small. Lemma \ref{lem:gap} further implies that if we let $\mat X^* \defeq \matb X$ in \eqref{eq:upper-bounds}, then the extent of breaching the equality in \eqref{eq:upper-bounds-tightness} can be inherently slight. The following insight is useful for designing approximate solution methods to original min-max problems.

\begin{insight}[Solve Max-Min Counterpart]\label{insight:approx-solution}
When employing the max-min counterpart $\max_{\mat H  \in \cal H} \min_{\mat X \in \cal X}  \phi(\mat H, \mat X)$ for designing the approximate solution method to the original min-max problem $\min_{\mat X \in \cal X} \max_{\mat H  \in \cal H}  \phi(\mat H, \mat X)$, the key is to choose an appropriate $\mat X^*$ such that the discrepancy between $\max_{\mat H \in \cal H} \phi(\mat H, \mat X^*)$ and $\max_{\mat H \in \cal H} \min_{\mat X \in \cal X}  \phi(\mat H, \mat X)$ can be minimized to its utmost extent. \stp
\end{insight}

In what follows, based on Lemmas \ref{lem:max-form} and \ref{lem:gap} and Insight \ref{insight:approx-solution}, we discuss the solution methods to the robust counterparts \eqref{eq:robust-waveform-design} and \eqref{eq:robust-waveform-design-tradeoff}, respectively in Subsections \ref{subsec:sol-robust-sensing-centric} and \ref{subsec:sol-robust-joint}. In Subsection \ref{subsubsec:remarks}, the proposed methods and their computational complexities are summarized; see Tables \ref{tab:solution-method} and \ref{tab:computational-complexity}.

\subsection{Solution Method to Robust Counterpart  \captext{\eqref{eq:robust-waveform-design}}}\label{subsec:sol-robust-sensing-centric}

This subsection designs the solution method to the robust counterpart \eqref{eq:robust-waveform-design} of the sensing-centric waveform design problem \eqref{eq:waveform-design}.

\subsubsection{Model Reformulation}\label{sec:max-form}
Since Problem {\eqref{eq:robust-waveform-design} is non-convex in $\mat X$ and non-concave in $\mat H$, conventional solution methods (e.g., alternating descent, Lagrangian duality) to min-max problems are not applicable.

Equipped with Lemma \ref{lem:max-form}, the theorem below transforms the robust counterpart \eqref{eq:robust-waveform-design} into a tractable equivalent.
\begin{theorem}\label{thm:max-form}
    Let $\cal X \defeq \{\mat X:\mat X \mat X^\H = L \mat R\}$. Suppose that the beampattern-inducing matrix $\mat R$ is positive definite and $\mat F$ is invertible such that $\mat R = \mat F \mat F^\H$ (e.g., Cholesky decomposition).\footnote{Note that given $\mat R$, there may exist multiple $\mat F$.} 
    If there exists ${\mat X}^* \in \cal X$ such that
    \begin{equation}\label{eq:sensing-centric-thm-condition}
        \max_{\mat H \in \cal H} \min_{\mat X \in \cal X} \|\mat H \mat X - \mat S\|^2_F = \max_{\mat H \in \cal H} \|\mat H {\mat X}^* - \mat S\|^2_F,
    \end{equation}
    then Problem \eqref{eq:robust-waveform-design} is equivalent to
    \begin{equation}\label{eq:max-form}
        \begin{array}{cl}
            \displaystyle \max_{\mat H} & \|\sqrt{L} \cdot \mat H \cdot \mat F \cdot \mat U \mat I_{N \times L} \mat V^\H - \mat S\|^2_F \\
            \st & \|\mat H - \bar{\mat H}\| \le \theta, \\
            & \mat U \mat \Sigma \mat V^\H \overset{\text{SVD}}{=} \mat F^\H \mat H^\H\mat S,
        \end{array}
    \end{equation}
    where $\mat U \mat \Sigma \mat V^\H \overset{\text{SVD}}{=} \mat F^\H \mat H^\H\mat S$ means the singular value decomposition (SVD) of $\mat F^\H \mat H^\H\mat S$ and $\mat I_{N \times L} \defeq [\mat I_N, \mat 0_{N \times (L - N)}]$; the $N \times (L - N)$ zero matrix is denoted by $\mat 0_{N \times (L - N)}$. 
\end{theorem}
\begin{proof}
    See Appendix D in supplementary materials. \stp
\end{proof}

Theorem \ref{thm:max-form} implies that the original two-layer min-max problem \eqref{eq:robust-waveform-design} can be conditionally equivalently solved by the single-layer maximization problem \eqref{eq:max-form}. The equivalence stems from the condition \eqref{eq:sensing-centric-thm-condition} and the existence of the closed-form solution $\mat X^*_{\mat H}$ for every given $\mat H$ \cite[Eq.~(15)]{liu2018toward}: 
\begin{equation}\label{eq:opt-X-given-H}
\mat X^*_{\mat H} = \sqrt{L} \mat F \mat U \mat I_{N \times L} \mat V^\H,~~~\forall \mat H \in \cal H.
\end{equation}  
Condition \eqref{eq:sensing-centric-thm-condition}, which is motivated by Lemma \ref{lem:max-form}, enforces the strong min-max property between $\min_{\mat X \in \cal X} \max_{\mat H  \in \cal H}  \|\mat H \mat X - \mat S\|^2_F$ and $\max_{\mat H  \in \cal H} \min_{\mat X \in \cal X}  \|\mat H \mat X - \mat S\|^2_F$. Since directly attacking the original problem \eqref{eq:robust-waveform-design} is technically challenging, this paper solves the reformulated problem \eqref{eq:max-form}, wherein the condition \eqref{eq:sensing-centric-thm-condition} needs to be stringently satisfied to ensure the equivalence between \eqref{eq:robust-waveform-design} and \eqref{eq:max-form}.

Although the condition \eqref{eq:sensing-centric-thm-condition} is particularized from the general case \eqref{eq:upper-bounds-tightness}, it is still technically difficult to verify the existence of $\mat X^*$ in \eqref{eq:sensing-centric-thm-condition}. Therefore, approximate solutions to the robust counterpart \eqref{eq:robust-waveform-design} can be obtained if the strict equality in the condition \eqref{eq:sensing-centric-thm-condition} can be compromised; cf. Insight \ref{insight:approx-solution}. To be specific, we aim to seek a solution ${\mat X}^*$ to \eqref{eq:max-form} such that $\max_{\mat H \in \cal H} \|\mat H {\mat X}^* - \mat S\|^2_F$ can approach or equal its lower bound $\max_{\mat H \in \cal H} \min_{\mat X \in \cal X} \|\mat H \mat X - \mat S\|^2_F$ from above; note that the closer the right-hand side of \eqref{eq:sensing-centric-thm-condition} is to the left-hand side, the smaller the gap between \eqref{eq:robust-waveform-design} and \eqref{eq:max-form}; cf. Lemma \ref{lem:max-form}.

We start with examining the properties of the maximization problem \eqref{eq:max-form}.

\subsubsection{Property Analysis of Maximization Problem \captext{\eqref{eq:max-form}}}\label{sec:property-analysis}

Denote the objective function in Problem \eqref{eq:max-form} by 
\begin{equation}\label{eq:obj-f}
    f(\mat H) \defeq \|\sqrt{L} \cdot \mat H \cdot \mat F \cdot \mat U \mat I_{N \times L} \mat V^\H - \mat S\|^2_F,
\end{equation}
which is a particularization of \eqref{eq:max-form-abstract}. The proposition below establishes the Lipschitz continuity of the function $f$ on $\cal H$.
\begin{proposition}\label{prop:f-continuous}
    The function $f$ is $L_f$-Lipschitz continuous in $\mat H$ on $\cal H$ with respect to the norm $\|\cdot\|$ used in defining $\cal H$, where
    \begin{equation}\label{eq:L_f}
        L_f = 2 B LP_\T \cdot (\|\bar{\mat H}\|_F + B\theta) + 2 B \sqrt{LP_\T} \cdot \|\mat S\|_F
    \end{equation}
    and $B$ is a finite positive real number such that $\|\mat H\|_F \le B\|\mat H\|$ for every $\mat H \in \cal \C^{K \times N}$.
\end{proposition}
\begin{proof}
See Appendix E in supplementary materials. \stp
\end{proof}

It is well believed that a Lipschitz continuous objective function is much easier to be globally maximized \cite[p.~7]{sergeyev2013introduction}. An extreme case is that the Lipschitz constant is zero on $\cal H$ so $f$ is constant on $\cal H$. In this case, we just need to evaluate only one point on $\cal H$ to globally maximize $f$. In addition, the Lipschitz continuity of $f$ is also important to control the extent of breaching the equality in \eqref{eq:sensing-centric-thm-condition}; recall Lemma \ref{lem:gap}. Since $f$ is Lipschitz continuous, it is almost everywhere differentiable. However, $f$ is non-convex and non-concave in $\mat H$ on $\cal H$.

\begin{proposition}\label{prop:f-other-properties}
The function $f$ is neither convex nor concave in $\mat H$ on $\cal H$.
\end{proposition}
\begin{proof}
See Appendix F in supplementary materials.
\stp
\end{proof}

Nevertheless, the objective function $f(\mat H)$ has tight upper bound $\overline f(\mat H)$ and lower bound $\underline f(\mat H)$. In addition, $\overline f(\mat H)$ is positive-definite quadratic (thus convex).
\begin{proposition}\label{prop:f-upper-bound}
    The function $f(\mat H)$ is upper bounded by 
    \[
        \overline f(\mat H) \defeq \|\mat H \bar{\mat X} - \mat S\|^2_F
    \]
    and lower bounded by
    \[
        \underline f(\mat H) \defeq \Big[ \sqrt{L\Tr[\mat H^\H \mat H \mat R]} - \|\mat S\|_F\Big]^2
    \]
    for all possible $\mat H \in \C^{K \times N}$, not necessarily on $\cal H$,
    where 
    \[
    \bar{\mat X} \defeq \sqrt{L} \mat F \bar{\mat U} \mat I_{N \times L} \bar{\mat V}^\H, 
    \]
    $\bar{\mat U} \bar{\mat \Sigma} \bar{\mat V}^\H \overset{\text{SVD}}{=} \mat F^\H \bar{\mat H}^\H\mat S$, and $\bar{\mat H}$ is the center of $\cal H$. In addition, the following are true.
    \begin{enumerate}
        \item The upper bound $\overline f(\mat H)$ is positive-definite quadratic (thus convex) in $\opvec(\mat H)$.
        \item The upper bound $\overline f(\mat H)$ and the lower bound $\underline f(\mat H)$ are both tight in the sense that the two bounds can be reached for some $\mat H$. 
        \item The upper bound $\overline f(\mat H)$ is $L_f$-Lipschitz continuous.
        \item The difference between $\overline f(\mat H)$ and $f(\mat H)$ is uniformly bounded by $2L_f \cdot \theta$ on $\cal H$, i.e.,
        \[
            \overline f(\mat H) - f(\mat H) \le 2L_f \cdot \theta.
        \]
        As a consequence, the positive-definite quadratic function $\overline f(\mat H) - 2L_f \cdot \theta$ is also a convex lower bound of $f(\mat H)$ on $\cal H$: i.e., $\underline f(\mat H) \defeq \overline f(\mat H) - 2L_f \cdot \theta$ is also a possible choice.
    \end{enumerate}
\end{proposition}
\begin{proof}
    See Appendix G in supplementary materials. 
    \stp
\end{proof}

Note that $\bar f(\mat H)$ is a particularization of \eqref{eq:upper-bound-abstract}, due to which the Lipschitz continuity is crucial. Since the Lipschitz continuity constant in \eqref{eq:L_f} is rather loose, the upper bound $2L_f \cdot \theta$ of the difference between $\bar f(\mat H)$ and $f(\mat H)$ can be much smaller in practice.

Propositions \ref{prop:f-other-properties} and \ref{prop:f-upper-bound} suggest that the landscape of the objective function $f(\mat H)$ is \quotemark{globally positive-definite quadratic} but \quotemark{locally non-convex and non-concave}. An illustration is shown in Fig. \ref{fig:f-landscape}, in whose caption the implications from Propositions \ref{prop:f-other-properties} and \ref{prop:f-upper-bound} are clarified.
\begin{figure}[!htbp]
    \centering
    \includegraphics[height=3.5cm]{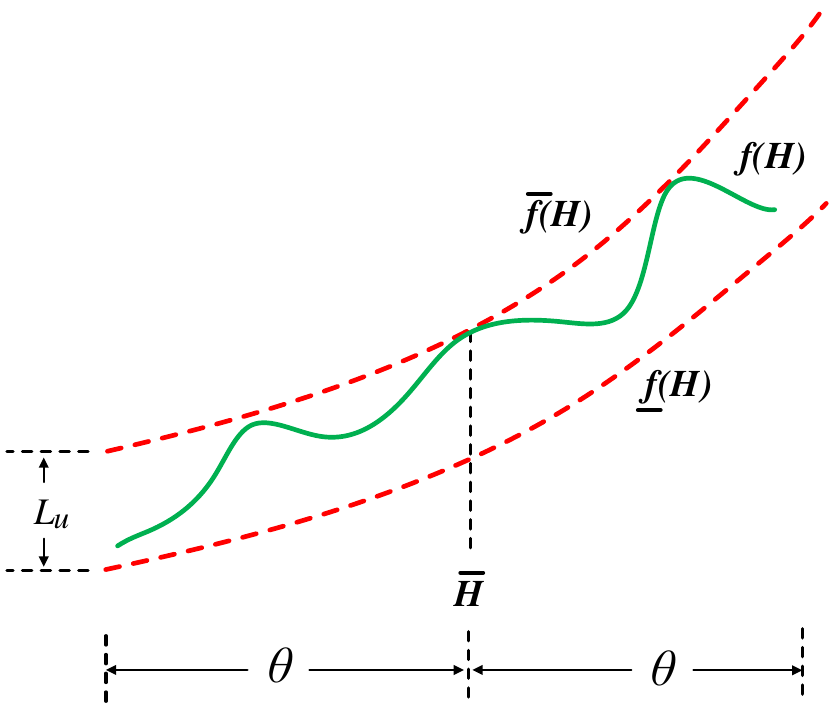}
    \caption{An one-dimensional illustration of the landscape of the objective function $f(\mat H)$ on $\cal H$. The upper bound and the lower bound of $f$ are $\overline{f}$ and $\underline{f}$, respectively; $\underline{f} \defeq \overline{f} - L_u$, where $L_u$ is the uniform gap between $\overline{f}$ and $\underline{f}$. Therefore, the upper bound $\overline{f}$ and the lower bound $\underline{f}$ are both positive-definite quadratic and thus convex. Note that $L_u \le 2L_f\cdot \theta$. The center and the scale of the feasible region $\cal H$ are $\bar{\mat H}$ and $\theta$, respectively. At the center $\bar{\mat H}$, $f(\bar{\mat H})$ reaches its upper bound $\overline{f}(\bar{\mat H})$. Since the three functions $f$, $\overline{f}$, and $\underline{f}$ share the same Lipschitz continuity constant, the objective function $f$ cannot be overly rugged (i.e., it is relatively flat) in between the two bounds $\overline{f}$ and $\underline{f}$. Note that $f$ is almost everywhere differentiable on $\cal H$.}
    \label{fig:f-landscape}
\end{figure}

In consideration of the landscape and other properties (i.e., non-convexity, non-concavity, Lipschitz continuity, etc.) of the objective function $f(\mat H)$ in Problem \eqref{eq:max-form}, the zero-order optimization methods such as heuristic optimization \cite{desale2015heuristic,mart2018handbook,kaveh2019metaheuristics,taillard2023design}, Bayesian optimization \cite{vazquez2010convergence,bull2011convergence,shahriari2015taking}, etc., are the last possible choices. However, a zero-order method is globally optimal for continuous objective functions if and only if the evaluation points governed by the searching algorithm are dense in the feasible region \cite{jones1998efficient}. Nevertheless, the searching space of Problem \eqref{eq:max-form} is extremely large; when the number $K$ of communication users and the number $N$ of transmit antennas are large, the dimension of $\mat H$ will be large because $\mat H \in \C^{K\times N}$. Recall that for large-scale multiple-input multiple-output (MIMO) communication systems, the antenna number $N$ is expected to be large \cite{bjornson2019massive}. To be specific, there are in total $2KN$ real numbers involved. Even if we only query $I$ points in each dimension, we must evaluate the function $f(\mat H)$ at least $I^{2KN}$ times to ensure global optimality for Problem \eqref{eq:max-form}. This is a well-identified dilemma known as \quotemark{the curse of dimensionality} in global optimization; for rigorous complexity results in numerical accuracy and numerical computation, see, e.g., \cite[p.~10]{sergeyev2013introduction}, \cite{eriksson2019scalable}, \cite[Table~1]{malherbe2017global}.

In summary, Problem \eqref{eq:max-form} has challenging properties as a global optimization problem: large dimensionality, non-convexity, and non-concavity. Therefore, an ad-hoc and efficient algorithm for approximately solving Problem \eqref{eq:max-form} is to be designed. 

\subsubsection{Approximate Solution Method to \captext{\eqref{eq:robust-waveform-design}}}\label{sec:solution-max-form}

This subsection studies an approximate solution method to the robust counterpart \eqref{eq:robust-waveform-design} by leveraging the maximization problem \eqref{eq:max-form}. The approximation arises from sacrificing the strict equality in the condition \eqref{eq:sensing-centric-thm-condition}.

By introducing auxiliary variables $\mat A$, $\mat U$, $\mat \Sigma$, and $\mat V$, we consider an optimization equivalent of Problem \eqref{eq:max-form}:
\begin{equation}\label{eq:max-form-trans}
    \begin{array}{cl}
        \displaystyle \max_{\mat H, \mat A, \mat U, \mat \Sigma, \mat V} & \|\sqrt{L} \mat H \mat F \mat A - \mat S\|^2_F \\
        \st & \|\mat H - \bar{\mat H}\| \le \theta,\\
        & \mat U \mat \Sigma \mat V^\H = \mat F^\H \mat H^\H\mat S, \\
        & \mat A = \mat U \mat I_{N \times L} \mat V^\H,~\mat U \mat U^\H = \mat I_N,~\mat V \mat V^\H = \mat I_L, \\
        & \mat \Sigma \text{~is diagonal, non-negative, and real};
    \end{array}
\end{equation}
note from Theorem \ref{thm:max-form} that for every given $\mat H$, the corresponding optimal waveform is  
\begin{equation}\label{eq:optimal-X-given-H}
    \mat X^*_{\mat H} = \sqrt{L} \mat F \mat A.
\end{equation}
In terms of the variables $\mat U$, $\mat \Sigma$, and $\mat V$, the transformed problem \eqref{eq:max-form-trans} is just a feasibility problem because the three variables are not involved in the objective function. Since the rank of the matrix $\mat F^\H \mat H^\H\mat S$ is no larger than $K$ where $K \le N \le L$, the SVD of $\mat F^\H \mat H^\H\mat S$ is not unique for every feasible $\mat H$; to be specific, the feasible values of $\mat U$ and $\mat V$ are not unique. As a result, the feasible value of $\mat A$ is also not unique given $\mat H$. 

In terms of the variable $\mat H$, the vectorized optimization equivalent of \eqref{eq:max-form-trans} is
\begin{equation}\label{eq:max-form-trans-vec}
    \begin{array}{cl}
        \displaystyle \max_{\mat H} \max_{\mat A, \mat U, \mat \Sigma, \mat V} & \|\sqrt{L}(( \mat F \mat A)^\T \otimes \mat I_K) \opvec(\mat H) - \opvec(\mat S)\|^2_2\\
        \st & (\mat F^\T \otimes \mat S^\H) \opvec(\mat H) = \opvec(\mat V \mat \Sigma^\H \mat U^\H), \\
        & \|\opvec(\mat H) - \opvec(\bar{\mat H})\| \le \theta, \\
        & \mat A = \mat U \mat I_{N \times L} \mat V^\H, \\ 
        &\mat U \mat U^\H = \mat I_N,~\mat V \mat V^\H = \mat I_L, \\
        & \mat \Sigma \text{~is diagonal, non-negative, and real}.
    \end{array}
\end{equation}
The key properties of the problem \eqref{eq:max-form-trans-vec} are given below.

\begin{proposition}\label{prop:max-form-solution}
Consider the reformulated problem \eqref{eq:max-form-trans-vec}. The following is true.
\begin{enumerate}
    \item The objective function is positive-definite quadratic in $\opvec(\mat H)$.
    \item In terms of both $\mat H$ and $\mat A$, the objective function is convex, and the constraints are also convex.\footnote{However, note that the optimization problem \eqref{eq:max-form-trans-vec} is not convex in $\mat H$ and $\mat A$ because it is a maximization problem.}
\end{enumerate}
\end{proposition}
\begin{proof}
    See Appendix H in supplementary materials.
    \stp
\end{proof}

To limit the technical complexity of, and obtain an elegant solution to, Problem \eqref{eq:max-form-trans-vec}, we impose an additional constraint 
\begin{equation}\label{eq:strategy}
\mat U \mat I_{N\times L} \mat V^\H = \bar{\mat U} \mat I_{N\times L} \bar{\mat V}^\H    
\end{equation}
to limit the size and shape of the feasible region of $\mat H$, and thus, simplify the problem. Intuitively speaking, this proposal is to let the nominally optimal solution $\bar{\mat X}$ \textbf{simultaneously} solve the robust problem, that is,
\begin{equation}\label{eq:trick}
\max_{\mat H} \min_{\mat X} \|\mat H \mat X - \mat S\|^2_F = \max_{\mat H} \|\mat H \bar{\mat X} - \mat S\|^2_F,
\end{equation} 
which adheres to the condition \eqref{eq:sensing-centric-thm-condition}; cf. Lemma \ref{lem:gap}.
Suppose that $\mat H^*$ solves \eqref{eq:max-form} (or equivalently \eqref{eq:max-form-trans} and \eqref{eq:max-form-trans-vec}) and $\mat X^*$ is associated with $\mat H^*$ through \eqref{eq:opt-X-given-H}; that is
\[
(\mat H^*, \mat X^*) \in \argmax_{\mat H \in \cal H} \argmin_{\mat X \in \cal X} \|\mat H \mat X - \mat S\|^2_F.
\]
By employing the strategies \eqref{eq:strategy} and \eqref{eq:trick}, the relation of 
\begin{equation}\label{eq:implication-of-strategy}
\mat X^* = \bar{\mat X}
\end{equation}
can be guaranteed where $\bar{\mat X} = \sqrt{L} \mat F \bar{\mat A} = \sqrt{L} \mat F \bar{\mat U} \mat I_{N\times L} \bar{\mat V}^\H$. We do not directly consider the constraint $\sqrt{L} \mat F \mat U \mat I_{N\times L} \mat V^\H = \bar{\mat X}$ (i.e., $\mat X = \matb X$) because $\mat F$ is invertible, and therefore, the two ways are equivalent. 

The method below formalizes the above reasoning, which approximately solves \eqref{eq:max-form} and compromises the exact equality constraint in \eqref{eq:sensing-centric-thm-condition}.

\begin{method}\label{method:max-form-solution}
Suppose a SVD of $\mat F^\H \bar{\mat H}^\H\mat S$ is $\bar{\mat U} \bar{\mat \Sigma} \bar{\mat V}^\H \overset{\text{SVD}}{=} \mat F^\H \bar{\mat H}^\H\mat S$. Let $\bar{\mat A} \defeq \bar{\mat U} \mat I_{N \times L} \bar{\mat V}^\H$. If we consider $\mat U$ and $\mat V$ such that $\mat U \mat \Sigma \mat V^\H \overset{\text{SVD}}{=} \mat F^\H \mat H^\H\mat S$ and $\mat U \mat I_{N\times L} \mat V^\H = \bar{\mat U} \mat I_{N\times L} \bar{\mat V}^\H$, then Problem \eqref{eq:max-form-trans-vec} is approximately solved by $(\mat H^*, \bar{\mat A}, {\mat U}^*, \mat \Sigma^*, {\mat V}^*)$ where $\mat H^*$ is a global maximum of the upper bound function $\overline{f}(\mat H)$, i.e.,
    \begin{equation}\label{eq:max-form-trans-vec-sol-H}
        \begin{array}{ccl}
            \mat H^* &\in \displaystyle \argmax_{\mat H}  & \|\sqrt{L}(( \mat F \bar{\mat A})^\T \otimes \mat I_K) \opvec(\mat H) - \opvec(\mat S)\|^2_2\\
            & \st & \|\opvec(\mat H) - \opvec(\bar{\mat H})\| \le \theta, 
        \end{array}
    \end{equation}
and $({\mat U}^*, \mat \Sigma^*, {\mat V}^*)$ is a global minimum of
    \begin{equation}\label{eq:max-form-trans-vec-sol-UV}
        \begin{array}{cl}
            \displaystyle \min_{\mat U, \mat \Sigma, \mat V}  & \alpha \|\mat U \mat \Sigma \mat V^\H - \mat F^\H \mat H^{*\H} \mat S\|^2_F + \|\mat U \mat I_{N\times L} \mat V^\H - \bar{\mat A} \|^2_F \\
            \st & \mat U \mat U^\H = \mat I_N,~\mat V \mat V^\H = \mat I_L, \\
            & \mat \Sigma \text{~is diagonal, non-negative, and real},
        \end{array}
    \end{equation}
where $\alpha \ge 0$ is a large real number to numerically ensure that $\mat U \mat \Sigma \mat V^\H$ is a SVD of $\mat F^\H \mat H^{*\H} \mat S$. \stp
\end{method}

In Method \ref{method:max-form-solution}, the exact equality constraint in \eqref{eq:strategy}, and therefore in \eqref{eq:trick}, is sacrificed; cf. Condition \eqref{eq:sensing-centric-thm-condition} and Insight \ref{insight:approx-solution}. Specifically, we just require $\mat U \mat I_{N\times L} \mat V^\H$ to be as close as possible to $\bar{\mat A}$. This relaxation brings two benefits: 
\begin{enumerate}
    \item First, the relaxed problem \eqref{eq:max-form-trans-vec-sol-UV} can be technically readily solved over $\mat U$, $\mat \Sigma$, and $\mat V$ (see Proposition \ref{thm:solution-remedy-problem});
    
    \item Second, the values of $\mat U$ and $\mat V$ that strictly satisfy $\mat U \mat I_{N\times L} \mat V^\H = \bar{\mat A}$ might not exist for every specified $\mat H$; some technical regularization conditions are needed to ensure the existence. However, if we use the soft-constraint counterpart as in \eqref{eq:max-form-trans-vec-sol-UV}, we do not need to explicitly derive the regularization conditions.
\end{enumerate}
As a result of Method \ref{method:max-form-solution}, compared to the strict equality in \eqref{eq:implication-of-strategy}, an approximation is achieved:
    \begin{equation}\label{eq:real-implementation-of-strategy}
        \mat X^* \approx \bar{\mat X}.
    \end{equation}
For an intuitive interpretation of Method \ref{method:max-form-solution}, see Appendix I in supplementary materials.

As indicated by Method \ref{method:max-form-solution}, the key to approximately solving the robust counterpart \eqref{eq:robust-waveform-design} is to solve the two sub-problems \eqref{eq:max-form-trans-vec-sol-H} and \eqref{eq:max-form-trans-vec-sol-UV}. 

\textbf{Solution to Problem \eqref{eq:max-form-trans-vec-sol-H}}: Since Problem \eqref{eq:max-form-trans-vec-sol-H} is not convex, it is not straightforward to solve. Nevertheless, Problem \eqref{eq:max-form-trans-vec-sol-H} has nice properties to benefit the design of a globally optimal algorithm: it is positive-definite quadratic in $\opvec(\mat H)$ and has convex constraints for $\opvec(\mat H)$. We rewrite \eqref{eq:max-form-trans-vec-sol-H} in real spaces in a compact form:
\begin{equation}\label{eq:H-compact}
    \begin{array}{cl}
        \displaystyle \max_{\vec h \in \R^{2KN}} & \|\mat C \vec h - \vec s\|^2_2, \\
        \st & \|\vec h - \bar{\vec h}\| \le \theta,
    \end{array}
\end{equation}
where 
\[
\vec h \defeq \left[
\begin{array}{c}
\operatorname{real}(\opvec(\mat H)) \\
\operatorname{imag}(\opvec(\mat H))
\end{array}
\right] \in \R^{2KN}
\]
is a real-valued vector constructed by stacking the real and imaginary components of $\opvec(\mat H)$; the quantities 
\begin{quote}
$\mat C \in \R^{2KL \times 2KN}$,~~~~${\vec s} \in \R^{2KL}$,~~~~and $\bar{\vec h} \in \R^{2KN}$
\end{quote}
are constructed from \eqref{eq:max-form-trans-vec-sol-H} in a similar way, during which the lemma below is useful.
\begin{lemma}\label{lem:stacking}
    Suppose $\mat \Xi \defeq \mat \Gamma + \mat \Theta j$ and $\vec \xi \defeq \vec a + \vec b j$ where $j$ denotes the imaginary unit; $\mat \Gamma$ and $\mat \Theta$ are real matrices and $\vec a$ and $\vec b$ are real vectors. We have
    \[
        \left[
        \begin{array}{cc}
            \operatorname{real}(\mat \Xi \vec \xi) \\
            \operatorname{imag}(\mat \Xi \vec \xi)
        \end{array}
        \right] = 
        \left[
        \begin{array}{cc}
            \mat \Gamma  &  - \mat \Theta \\
            \mat \Theta & \mat \Gamma
        \end{array}
        \right]
        \left[
        \begin{array}{c}
            \vec a \\
            \vec b
        \end{array}
        \right].
    \]
\end{lemma}
\begin{proof}
    See Appendix J in supplementary materials. \stp
\end{proof}

Let the objective function of \eqref{eq:H-compact} be $p(\vec h)$. We can further show that $p(\vec h)$ is positive-definite and strongly convex.
\begin{lemma}\label{lem:positive-definite}
    The objective function $p(\vec h)$ of \eqref{eq:H-compact} is positive definite and strongly convex in $\vec h$.
\end{lemma}
\begin{proof}
    See Appendix K in supplementary materials.
    \stp
\end{proof}

Equipped with Lemma \ref{lem:positive-definite}, Proposition \ref{prop:convex-quadratic-max} below provides the globally optimal solution(s) to Problem \eqref{eq:H-compact}, and hence, to Problem \eqref{eq:max-form-trans-vec-sol-H}.
\begin{proposition}[Solution to Upper-Bound Problem \eqref{eq:H-compact}]\label{prop:convex-quadratic-max}
The globally optimal solution to the upper-bound problem \eqref{eq:H-compact}, and therefore to \eqref{eq:max-form-trans-vec-sol-H}, is given by the following algorithmic iteration process.
\begin{enumerate}[Step a)]
    \item Initialization: Let the iteration count $k = 0$. Choose $\vec h_0$ from the domain $\{\vec h:~\|\vec h - \bar{\vec h}\| \le \theta\}$ but $\vec h_0 \neq (\mat C^\T \mat C)^{-1} \mat C^\T \vec s$.
    \item Solve
    \begin{equation}\label{eq:sub-prob-1}
    \begin{array}{ccl}
        \vec y_k =  & \displaystyle \argmax_{\vec y \in \R^{2KN}} &(\vec h^\T_{k-1} \mat C^\T \mat C - \vec s^\T \mat C) \cdot \vec y \\
        & \st &\|\mat C \vec y - \vec s\|^2_2 = \|\mat C \vec h_{k-1} - \vec s\|^2_2.
    \end{array}
    \end{equation}
    The constraint can be changed to $\|\mat C \vec y - \vec s\|^2_2 \le \|\mat C \vec h_{k-1} - \vec s\|^2_2$ without losing the optimality.
    \item Solve
    \begin{equation}\label{eq:sub-prob-2}
    \begin{array}{ccl}
        \vec h_{k} &= \displaystyle \argmax_{\vec h \in \R^{2KN}} & (\vec y^\T_k \mat C^\T \mat C - \vec s^\T \mat C) \cdot \vec h \\
        & \st & \|\vec h - \bar{\vec h}\| \le \theta.
    \end{array}
    \end{equation}
    \item Repeat Step b) and Step c) until 
    \begin{equation}\label{eq:converge-condition}
        \ip{\mat C^\T \mat C \vec y_k - \mat C^\T\vec s}{~\vec h_{k} - \vec y_k} \le 0
    \end{equation}
    where $\ip{\cdot}{\cdot}$ denotes the inner product in real spaces.
\end{enumerate}
When the iteration process terminates, $\vec h_{k}$ is a globally optimal solution to Problem \eqref{eq:H-compact}. In addition, for every $k$, it holds that $p(\vec h_{k}) > p(\vec h_{k-1})$; i.e., in every iteration round, the solution $\vec h_k$ is improved.
\end{proposition}
\begin{proof}
    See Appendix L in supplementary materials.
    \stp
\end{proof}

In addition to Proposition \ref{prop:convex-quadratic-max}, \cite[Algo.~2]{ben2022algorithm} is an alternative method that has been empirically shown to be computationally efficient and almost globally optimal for convex-quadratic maximization. Problem \eqref{eq:sub-prob-1} can be analytically solved.
\begin{proposition}[Solution to \eqref{eq:sub-prob-1}]\label{prop:sub-prob-1}
    Suppose that $\mat M^\T \mat M$ denotes a square-root decomposition (e.g., Cholesky) of $\mat C^\T \mat C$ and $\mat M$ is invertible. Let 
    \[
    \gamma \defeq \sqrt{\|\mat C \vec h_{k-1} - \vec s\|^2_2 + \vec s^\T \mat C (\mat C^\T \mat C)^{-1} \mat C^\T \vec s - \vec s^\T \vec s}. 
    \]
    Then, a maximum of \eqref{eq:sub-prob-1} is
    \[
        \mat M^{-1} \cdot \left[ \gamma \frac{\mat M^{-\T} (\mat C^\T \mat C \vec h_{k-1} -  \mat C^\T \vec s)}{\| \mat M^{-\T} (\mat C^\T \mat C \vec h_{k-1} -  \mat C^\T \vec s) \|_2} + \mat M^{-\T} \mat C^\T \vec s \right].
    \]
\end{proposition}
\begin{proof}
    See Appendix M in supplementary materials.
    \stp
\end{proof}

As for \eqref{eq:sub-prob-2}, when the norm constraint is defined using the infinity norm, it is particularized into a linear program, which can be efficiently solved using the simplex method. In addition, we prove that a closed-form solution exists for \eqref{eq:sub-prob-2} if the norm constraint is defined using the $2$-norm.
\begin{proposition}[Solution to \eqref{eq:sub-prob-2}]\label{prop:sub-prob-2}
    If the norm constraint is defined using the $2$-norm, then a maximum of \eqref{eq:sub-prob-2} is 
    \[
        \theta \frac{\mat C^\T \mat C \vec y_k - \mat C^\T \vec s}{\|\mat C^\T \mat C \vec y_k - \mat C^\T \vec s\|_2} + \bar{\vec h}.
    \]
\end{proposition}
\begin{proof}
    See Appendix N in supplementary materials.
    \stp
\end{proof}

After addressing Problem \eqref{eq:max-form-trans-vec-sol-H}, we solve Problem \eqref{eq:max-form-trans-vec-sol-UV}, which is termed a \textit{remedy} problem.

\textbf{Solution to Problem \eqref{eq:max-form-trans-vec-sol-UV}}: The solution to Problem \eqref{eq:max-form-trans-vec-sol-UV} is summarized in the proposition below.
\begin{proposition}[Solution to Remedy Problem \eqref{eq:max-form-trans-vec-sol-UV}]\label{thm:solution-remedy-problem}
The solution to Problem \eqref{eq:max-form-trans-vec-sol-UV} is given by the following iteration.
\begin{enumerate}[Step a)]
    \item Fix $\mat \Sigma$ and $\mat V$, and minimize over $\mat U$. The optimal solution is 
    $
    \mat U^* = \mat U_1 \mat V^\H_1
    $ 
    where $\mat U_1 \mat \Sigma_1 \mat V^\H_1 \overset{\text{SVD}}{=} \mat B_1 \mat A^\H_1$,
    $
        \mat A_1 \defeq [\mat I_{N \times L} \mat V^\H,~\sqrt{\alpha} \mat \Sigma \mat V^\H]
    $, 
    and  
    $
        \mat B_1 \defeq [\bar{\mat A},~\sqrt{\alpha} \mat F^\H \mat H^{*\H} \mat S]
    $.
    \item Fix $\mat \Sigma$ and $\mat U$, and minimize over $\mat V$.  The optimal solution is 
    $
    \mat V^* = \mat U_2 \mat V^\H_2
    $ 
    where $\mat U_2 \mat \Sigma_2 \mat V^\H_2 \overset{\text{SVD}}{=} \mat B_2 \mat A^\H_2$,
    $
        \mat A_2 \defeq [\mat I^\H_{N \times L} \mat U^\H,~\sqrt{\alpha} \mat \Sigma^\H \mat U^\H]
    $, 
    and 
    $
        \mat B_2 \defeq [\bar{\mat A}^\H,~\sqrt{\alpha} \mat S^\H \mat H^{*} \mat F]
    $.
    \item Fix $\mat U$ and $\mat V$, and minimize over $\mat \Sigma$. The optimal solution is 
    $
    \mat \Sigma^* = \operatorname{Proj}_{\Omega_{N \times L}}[\mat U^\H \mat F^\H \mat H^{*\H} \mat S \mat V]
    $
    where $\operatorname{Proj}$ denotes the projection operator and $\Omega_{N \times L}$ denotes the diagonal, non-negative, and real space with the dimension of ${N \times L}$.
    \item Repeat Steps a)-c) until convergence.
\end{enumerate}
Regarding the objective value, the iteration process is guaranteed to converge for any initial values of $(\mat U, \mat \Sigma, \mat V)$.
\end{proposition}
\begin{proof}
    See Appendix O in supplementary materials.
    \stp
\end{proof}

\begin{remark}[Initialization in Proposition \ref{thm:solution-remedy-problem}]
The iteration in Proposition \ref{thm:solution-remedy-problem} can be initialized using the SVD of $\mat F^\H \mat H^{*\H}\mat S$, that is, $\mat U \mat \Sigma \mat V^\H \overset{\text{SVD}}{:=} \mat F^\H \mat H^{*\H}\mat S$ . \stp
\end{remark}

\begin{remark}[Projection in Proposition \ref{thm:solution-remedy-problem}]
    For $\operatorname{Proj}_{\Omega_{N \times L}}[\cdot]$, one strategy is to use the singular value matrix of the argument. Another strategy is to only keep the real non-negative components in the diagonal entries and zero the rest. \stp
\end{remark}

\subsubsection{Another Approximate Solution Method to \captext{\eqref{eq:robust-waveform-design}}}\label{sec:another-solution-max-form}
This subsection proposes another approximate solution method to the robust counterpart \eqref{eq:robust-waveform-design}. Based on the design philosophy of Method \ref{method:max-form-solution}, the following alternative method can be motivated.
\begin{method}\label{method:max-form-solution-another}
In Method \ref{method:max-form-solution}, Problem \eqref{eq:max-form-trans-vec-sol-UV} can be replaced by
\begin{equation}\label{eq:max-form-trans-vec-sol-X}
    \begin{array}{cl}
        \displaystyle \min_{\mat X \in \cal X}  & \alpha \|\mat H^* \mat X - \mat S\|^2_F + \| \mat X - \matb X \|^2_F.
    \end{array}
\end{equation} 
Therefore, Problem \eqref{eq:max-form-trans-vec} can also be approximately solved by \eqref{eq:max-form-trans-vec-sol-H} and \eqref{eq:max-form-trans-vec-sol-X}.
\stp
\end{method}

Note that in the objective of Problem \eqref{eq:max-form-trans-vec-sol-UV}, the first term is to guarantee that $\mat X$ solves $\min_{\mat X \in \cal X} \|\mat H^* \mat X - \mat S\|^2_F$, while the second term is to reduce $\|\mat X - \matb X\|^2_F$. This is why Problem \eqref{eq:max-form-trans-vec-sol-X} can directly be an alternative to Problem \eqref{eq:max-form-trans-vec-sol-UV}. However, the two problems are not equivalent, although both perform effectively and efficiently in experiments; see Section \ref{sec:experiment}. The closed-form solution to \eqref{eq:max-form-trans-vec-sol-X} is given in the proposition below.

\begin{proposition}\label{prop:sol-remedy-another}
Problem \eqref{eq:max-form-trans-vec-sol-X} is analytically solved by
\[
\mat X^* = \sqrt{L} \mat F \matt U \mat I_{N \times L} \matt V^\H
\]
where
\[
\matt U \matt \Sigma \matt V^\H \overset{\text{SVD}}{=} \mat F^\H \matt H^\H \matt S,
\]
\[
\matt H \defeq \left[
\begin{array}{c}
    \sqrt{\alpha} \mat H^* \\
    \mat I_{N}
\end{array}
\right],
\text{~~~and~~~}
\matt S \defeq \left[
\begin{array}{c}
    \sqrt{\alpha} \mat S \\
    \matb X
\end{array}
\right].
\]
\end{proposition}
\begin{proof}
By defining $\matt H$ and $\matt S$, Problem \eqref{eq:max-form-trans-vec-sol-X} becomes $\min_{\mat X \in \cal X} \|\matt H \mat X - \matt S\|^2_F$, which is solved by $\mat X^*$ defined in the statement; cf. \cite[Eq.~(15)]{liu2018toward} and \eqref{eq:opt-X-given-H}.  \stp
\end{proof}

\subsection{Solution Method to Robust Counterpart  \captext{\eqref{eq:robust-waveform-design-tradeoff}}}\label{subsec:sol-robust-joint}

This subsection designs the solution method to the robust counterpart \eqref{eq:robust-waveform-design-tradeoff} of the joint waveform design problem \eqref{eq:waveform-design-tradeoff}.

\subsubsection{Model Reformulation}\label{sec:max-form-tradeoff} Depending on the design preference, let the feasible region $\cal X$ of waveforms $\mat X$ be defined as one of the follows:
\begin{itemize}
    \item TPC: $\cal X \defeq \{\mat X: \|\mat X\|^2_F = L P_\T$\}, for the total power constraint in \eqref{eq:total-power};
    \item PAPC: $\cal X \defeq \{\mat X: \diag(\mat X \mat X^\H) = \frac{L \cdot P_\T}{N} \mat I_N\}$, for the per-antenna power constraint in \eqref{eq:per-antenna-power}.
\end{itemize}
In analogy with Theorem \ref{thm:max-form} for Problem \eqref{eq:robust-waveform-design}, we give a reformulation of Problem \eqref{eq:robust-waveform-design-tradeoff}.
\begin{theorem}\label{thm:max-form-tradeoff} 
    Let $\phi(\mat H, \mat X) \defeq \rho \|\mat H \mat X - \mat S\|^2_F + (1 - \rho) \| \mat X - \mat X_s \|^2_F$.
    If there exists ${\mat X}^* \in \cal X$ such that
    \begin{equation}\label{eq:joint-thm-condition}
        \max_{\mat H \in \cal H} \min_{\mat X \in \cal X} \phi(\mat H, \mat X) = \max_{\mat H \in \cal H} \phi(\mat H, \mat X^*),
    \end{equation}
    then Problem \eqref{eq:robust-waveform-design-tradeoff} is equivalent to
    \begin{equation}\label{eq:max-form-tradeoff}
        \begin{array}{cl}
           \displaystyle \max_{\mat H}  &  \rho \|\mat H \mat X^*_{\mat H} - \mat S\|^2_F  + (1 - \rho) \|\mat X^*_{\mat H}  - \mat X_s\|^2_F \\
           \st & \|\mat H - \bar{\mat H}\| \le \theta,
        \end{array}
    \end{equation}
    where $\mat X^*_{\mat H}$, which depends on $\mat H$, is given by
    \begin{itemize}
        \item \cite[Algo.~1]{liu2018toward}, for the total power constraint \eqref{eq:total-power};
        \item \cite[Algo.~2]{liu2018toward}, for the per-antenna power constraint \eqref{eq:per-antenna-power}.
    \end{itemize}
\end{theorem}
\begin{proof}
    This is obvious from Lemma \ref{lem:max-form} and \cite{liu2018toward}. \stp
\end{proof}

Denote the objective function in Problem \eqref{eq:max-form-tradeoff} by 
\begin{equation}\label{eq:obj-g}
    g(\mat H) \defeq \rho \|\mat H \mat X^*_{\mat H} - \mat S\|^2_F  + (1 - \rho) \|\mat X^*_{\mat H}  - \mat X_s\|^2_F,
\end{equation}
which is a particularization of \eqref{eq:max-form-abstract}. Suppose that $\matb X$ solves the nominal joint waveform design problem \eqref{eq:waveform-design-tradeoff} under the nominal communication channel $\matb H$; for technical details, see \cite[Algos.~1, 2]{liu2018toward}.

Motivated by Lemma \ref{lem:gap}, the proposition below establishes the Lipschitz continuity of the function $g$ on $\cal H$.
\begin{proposition}\label{prop:g-continuous}
    The function $g(\mat H)$ defined in \eqref{eq:obj-g} is upper bounded by the function $\bar g(\mat H)$, for every $\mat H \in \cal H$, where
    \begin{equation}\label{eq:g-upper-bound}
            \bar g(\mat H) \defeq \rho \|\mat H \matb X - \mat S\|^2_F  + (1 - \rho) \|\matb X  - \mat X_s\|^2_F.
    \end{equation}
    In addition, $g$ and $\bar g$ are $L_g$-Lipschitz continuous in $\mat H$ on $\cal H$ with respect to the norm $\|\cdot\|$ used in defining $\cal H$, where 
    \begin{equation}\label{eq:L_g}
        \begin{array}{cl}
        L_g &= \rho \cdot L_f \\
            &= 2 \rho B L P_\T \cdot (\|\bar{\mat H}\|_F + B\theta) + 2\rho B \sqrt{LP_\T} \cdot \|\mat S\|_F
        \end{array}
    \end{equation}
    and $L_f$ is defined in \eqref{eq:L_f}; $B$ is a finite positive real number such that $\|\mat H\|_F \le B\|\mat H\|$ for every $\mat H \in \bb C^{K \times N}$.
\end{proposition}
\begin{proof}
The proof is routine given the techniques in the
proof of Proposition \ref{prop:f-continuous}. For details, see Appendix P in supplementary materials. \stp
\end{proof}

Note that $\bar g(\mat H)$ is a particularization of \eqref{eq:upper-bound-abstract}, due to which the Lipschitz continuity is crucial. Proposition \ref{prop:g-continuous} implies that the function $g$ in \eqref{eq:obj-g} is more \quotemark{flat} than $f$ in \eqref{eq:obj-f} since $\rho \le 1$, which can further reduce the gap studied in Lemma \ref{lem:gap}. In addition, from the perspective of global optimization, such flatness would benefit the design of solution methods, specifically, improving the efficiency of the search process of iterative algorithms; see \cite{hansen1995lipschitz,malherbe2017global}.

\subsubsection{Approximate Solution Method to \captext{\eqref{eq:robust-waveform-design-tradeoff}}}

Motivated by Method \ref{method:max-form-solution-another}, an approximate solution method to \eqref{eq:robust-waveform-design-tradeoff} can be summarized in the method below.

\begin{method}\label{method:robust-joint-design}
Problem \eqref{eq:robust-waveform-design-tradeoff} is approximately solved by $(\mat H^*, \mat X^*)$ where $\mat H^*$ is a global maximum of the upper bound function $\overline{g}(\mat H)$, that is,
    \begin{equation}\label{eq:max-form-trans-vec-sol-H-joint}
        \begin{array}{ccl}
            \mat H^* &\in \displaystyle \argmax_{\mat H}  & \|(\matb X^\T \otimes \mat I_K) \opvec(\mat H) - \opvec(\mat S)\|^2_2\\
            & \st & \|\opvec(\mat H) - \opvec(\bar{\mat H})\| \le \theta, 
        \end{array}
    \end{equation}
and $\mat X^*$ is a global minimum of
    \begin{equation}\label{eq:max-form-trans-vec-sol-X-joint}
        \begin{array}{cl}
            \displaystyle \min_{\mat X \in \cal X}  & \alpha \cdot \big[\rho \|\mat H^* \mat X - \mat S\|^2_F  + (1 - \rho) \|\mat X  - \mat X_s\|^2_F \big] + \\
            & \quad \quad \|\mat X - \matb X\|^2_F
        \end{array}
    \end{equation}
where $\alpha \ge 0$ is a large real number to numerically ensure that $\mat X$ solves 
\[
\min_{\mat X \in \cal X} \rho \|\mat H^* \mat X - \mat S\|^2_F  + (1 - \rho) \|\mat X  - \mat X_s\|^2_F. \tag*{$\square$}
\]
\end{method}
Problem \eqref{eq:max-form-trans-vec-sol-H-joint} can be transformed into \eqref{eq:H-compact} and solved by Proposition \ref{prop:convex-quadratic-max}. Problem \eqref{eq:max-form-trans-vec-sol-X-joint} can be transformed into
\begin{equation}\label{eq:max-form-trans-vec-sol-X-joint-trans}
\min_{\mat X \in \cal X} \frac{\alpha}{\alpha+1} \|\matt H \mat X - \matt S\|^2_F  + \frac{1}{\alpha+1} \|\mat X  - \matb X\|^2_F,    
\end{equation}
where
\[
\matt H \defeq \left[
\begin{array}{c}
    \sqrt{\rho} \mat H^* \\
    \sqrt{1-\rho} \mat I_{N}
\end{array}
\right],
\text{~~and~~}
\matt S \defeq \left[
\begin{array}{c}
    \sqrt{\rho} \mat S \\
    \sqrt{1-\rho} \mat X_s
\end{array}
\right].
\]
The transformed problem \eqref{eq:max-form-trans-vec-sol-X-joint-trans} can be solved by 
\begin{itemize}
    \item \cite[Algo.~1]{liu2018toward}, for the total power constraint \eqref{eq:total-power};
    \item \cite[Algo.~2]{liu2018toward}, for the per-antenna power constraint \eqref{eq:per-antenna-power}.
\end{itemize}

\subsection{Remarks on Solution Methods}\label{subsubsec:remarks}
This section studies three solution methods to robust counterparts: Methods \ref{method:max-form-solution} and \ref{method:max-form-solution-another} are designed for the robust counterpart \eqref{eq:robust-waveform-design}, whereas Method \ref{method:robust-joint-design} is for the robust counterpart \eqref{eq:robust-waveform-design-tradeoff}. 
As base-band digital signal processing methods, they are realized in digital signal processors of transmitters, for example, in distributed units (DUs) of wireless communication systems. 
The three methods are summarized in Table \ref{tab:solution-method}.
\begin{table}[!htbp]
\centering
\caption{Summary of Solution Methods}
\begin{tabular}{lllllllllll}
\hline
                \textbf{Problem}     & \textbf{Method}    & \textbf{Component}   & \textbf{Reformulation}  &  \textbf{Solution}           \\
\hline
\multirow{4}{*}{\eqref{eq:robust-waveform-design}} 
    &   \multirow{2}{*}{\ref{method:max-form-solution}}           &  \eqref{eq:max-form-trans-vec-sol-H}     &      \eqref{eq:H-compact}                                &    Prop. \ref{prop:convex-quadratic-max}      \\
    &                                                             &  \eqref{eq:max-form-trans-vec-sol-UV}    &      N/A    &    Prop. \ref{thm:solution-remedy-problem}      \\
    \cline{2-5}
    &   \multirow{2}{*}{\ref{method:max-form-solution-another}}   &  \eqref{eq:max-form-trans-vec-sol-H}     &      \eqref{eq:H-compact}                                &    Prop. \ref{prop:convex-quadratic-max}     \\
    &                                                             &  \eqref{eq:max-form-trans-vec-sol-X}     &      N/A    &    Prop. \ref{prop:sol-remedy-another}      \\
\hline
\multirow{2}{*}{\eqref{eq:robust-waveform-design-tradeoff}}  
    &  \multirow{2}{*}{\ref{method:robust-joint-design}}  &   \eqref{eq:max-form-trans-vec-sol-H-joint}      &      \eqref{eq:H-compact}                                &    Prop. \ref{prop:convex-quadratic-max}      \\
    &                                                     &   \eqref{eq:max-form-trans-vec-sol-X-joint}      &      \eqref{eq:max-form-trans-vec-sol-X-joint-trans}     &    \hspace{-0.5pt}\cite[Algo.~1,~2]{liu2018toward}      \\
\hline
\end{tabular}
\label{tab:solution-method}
\end{table}

All three methods involve a convex quadratic maximization problem \eqref{eq:H-compact}, specifically, \eqref{eq:max-form-trans-vec-sol-H} in Methods \ref{method:max-form-solution} and \ref{method:max-form-solution-another}, and \eqref{eq:max-form-trans-vec-sol-H-joint} in Method \ref{method:robust-joint-design}, which is solved using Proposition \ref{prop:convex-quadratic-max}. At each iteration in Proposition \ref{prop:convex-quadratic-max}, two dominating operations are included: \eqref{eq:sub-prob-1} and \eqref{eq:sub-prob-2}. However, both \eqref{eq:sub-prob-1} and \eqref{eq:sub-prob-2} can be analytically solved; see Propositions \ref{prop:sub-prob-1} and \ref{prop:sub-prob-2} where one Cholesky decomposition and several matrix inversion and multiplication operations are included. Problem \eqref{eq:max-form-trans-vec-sol-UV} is solved by the iteration process in Proposition \ref{thm:solution-remedy-problem}; at each iteration, three SVD and several matrix multiplication operations are contained. Problem \eqref{eq:max-form-trans-vec-sol-X} has the closed-form solution in Proposition \ref{prop:sol-remedy-another}, where one SVD and several matrix multiplication operations are encompassed. 

Table \ref{tab:computational-complexity} summarizes the asymptotically computational complexities of the involved solution methods, in terms of floating-point operation (FLOP). Note that in real-world operation, we have $K \le N \le L$. Therefore, the time complexity of the full SVD of an $N \times L$ matrix, which computes all SVD components, is $\cal O(L^2N + N^2L + N^3)$ using the Golub--Reinsch algorithm \cite[p.~493]{golub2013matrix}. For the Cholesky decomposition and the inversion of an $N \times N$ matrix, the time complexity is $\cal O(N^3)$. Since the variables in \eqref{eq:H-compact} are constructed from matrices, they are inherently large-dimensional. As a result, computing $\mat C^\T \mat C$ requires $\cal O(K^3 N^2 L)$ FLOPs and inverting $\mat C^\T \mat C$ requires $\cal O(K^3 N^3)$ FLOPs because $\mat C \in \R^{2KL \times 2KN}$.

\begin{table*}[!htbp]
\centering
\caption{Asymptotically Computational Complexity Analyses; cf. Table \ref{tab:solution-method}}
\begin{tabular}{l|l|l|l|lllllll}
\hline
                \textbf{Solution}                                    &       \textbf{Dominating Operation}                                                    &      \textbf{Time Complexity}                                   & \textbf{Solution Type} & \textbf{Remark} \\
\hline
\multirow{2}{*}{Prop. \ref{prop:convex-quadratic-max}}      &       \eqref{eq:sub-prob-1}, Prop. \ref{prop:sub-prob-1}, Chol \& MatInv \& MatMul      &      $\cal O(K^3N^3 + K^3N^2L + K^3L^2N)$    &   Closed-Form & \multirow{2}{*}{At Each Iteration}\\
                                                            &       \eqref{eq:sub-prob-2}, Prop. \ref{prop:sub-prob-2}, MatMul              &      $\cal O(K^3N^2L)$                                 &   Closed-Form &   \\
\hline

\multirow{3}{*}{Prop. \ref{thm:solution-remedy-problem}}    &       Step a, SVD \& MatMul                                                   &      $\cal O(N^3 + N^2L + N^2K + L^2N + NKL)$          &      Closed-Form & \multirow{3}{*}{At Each Iteration}\\
                                                            &       Step b, SVD \& MatMul                                                   &      $\cal O(L^3 + N^2L + N^2K + L^2N + NKL)$          &  Closed-Form &\\
                                                            &       Step c, SVD \& MatMul                                                   &      $\cal O(N^3 + N^2L + L^2 K + L^2N + NKL)$                       &  Closed-Form &\\
\hline

Prop. \ref{prop:sol-remedy-another}                         &       SVD \& MatMul                                                           &      $\cal O(N^3 + N^2L + N^2K + L^2N + NKL)$          & Closed-Form & Overall \\
\hline

\hspace{-0.5pt}\cite[Algo.~1]{liu2018toward}                &       See \hspace{-0.5pt}\cite[Sec. IV-E]{liu2018toward}                      &      $\cal O(N^3 + N^2L + N^2K + NKL)$                 & Iterative & Overall \\
\hline
\hspace{-0.5pt}\cite[Algo.~2]{liu2018toward}                &       See \hspace{-0.5pt}\cite[Sec. IV-E]{liu2018toward}                      &      $\cal O(N^2L + NKL)$                              & Iterative & At Each Iteration \\

\hline
\multicolumn{3}{l}{\tabincell{l}{Chol: Cholesky Decomposition;~~~~~~MatInv: Matrix Inversion;~~~~~~MatMul: Matrix Multiplication}} \\
\end{tabular}
\label{tab:computational-complexity}
\end{table*}

Table \ref{tab:computational-complexity} suggests that all the methods in this paper are solvable in polynomial times, and in this sense, they are computationally efficient. However, note that the time complexities of the algorithms in Propositions \ref{prop:convex-quadratic-max}, \ref{thm:solution-remedy-problem}, and \cite[Algo.~2]{liu2018toward} also depend on the number of numerical iterations. On the other hand, as the closed-form solution exists at every dominating operation, the actual running speeds of the proposed methods (i.e., Propositions \ref{prop:convex-quadratic-max}, \ref{prop:sub-prob-1}, \ref{prop:sub-prob-2}, \ref{thm:solution-remedy-problem}, and \ref{prop:sol-remedy-another}) are empirically very fast; see the experimental illustrations in Subsection \ref{subsubsec:convergence-computation}.

\section{Experiments}\label{sec:experiment}
The experiments are conducted using MATLAB 2023a on an $\text{HP}^\text{\textregistered}$ Desktop PC equipped with 12th Gen $\text{Intel}^\text{\textregistered}$ $\text{Core}^\text{\texttrademark}$ i7-12700K processor (3.60 GHz), 64.0 GB RAM, and 64-bit operating system. All the source data and codes are available online at GitHub: \url{https://github.com/Spratm-Asleaf/Robust-Waveform}. One may use them to reproduce the experimental results and verify the empirical claims in this paper. For each simulation case, we conduct $1000$ independent Monte--Carlo episodes, over which the reported results are averaged; for details, see Appendix Q in supplementary materials. 

\subsection{Experimental Setups}
Following the real-world experimental setups in Beijing, China, in \cite[Fig.~3]{zhang2020measuring}, we imagine a MIMO ISAC system placed at an intersection where two perpendicular streets cross each other, serving four downlink communication users while sensing two targets; see Fig. \ref{fig:scenario}. The key system parameters are configured as follows.
\begin{itemize}
    \item The carrier frequency is $5.9$ GHz; see \cite[Tab.~1]{zhang2020measuring};
    \item The total transmitting power is $P_{\T} = 34\text{ dBm} = 2.5$ watts; see \cite[Tab.~1]{zhang2020measuring};
    \item In \cite[p.~136]{zhang2020measuring}, the signal-to-noise ratio (SNR) is $24$ dB, that is, the noise power is $N_0 = P_{\T}/10^{(24/10)} = 0.01$ (watts). In our experiments, we use a lower SNR = 10 dB, i.e., $N_0 = 0.25$ watts, which is a more complex case; 
    \item The number of radio paths for each communication channel is in-between $60$ and $100$; the average is $80$; see \cite[Fig.~7]{zhang2020measuring}.
\end{itemize}
For sufficient generality and practicality, the mature MATLAB Phased Array System Toolbox and Communications Toolbox are employed to simulate the transmitters and the scattering-based communication channels.
\begin{figure}[!htbp]
	\centering
	\subfigure[Satellite View]{
	 	\includegraphics[height=3cm,width=4.1cm]{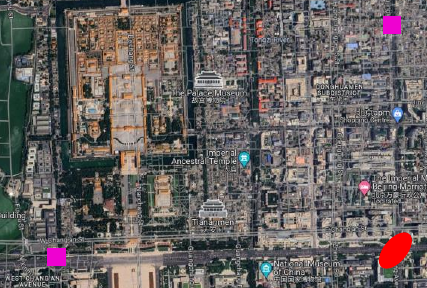}
        \label{fig:scenario_street_view}
	}
	\subfigure[Map View]{
	 	\includegraphics[height=3cm,width=4.1cm]{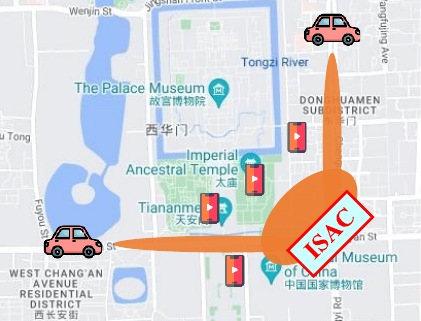}
        \label{fig:scenario_topology}
	}
    \caption{The ISAC system is placed at a crossing of two perpendicular streets in Beijing, China. It senses two targets (cars) and serves four downlink-communication users (phones). For adequate sensing qualities, two high-directional beams point to the two targets, respectively. In Fig. (a), the red ellipse represents the ISAC station, and magenta squares denote cars. In Fig. (b), the orange elongated ellipses illustrate beams. (Icon Credit: FLATICON.com; Photo Credit: Google Map.)}
    \label{fig:scenario}
\end{figure}

For other system parameters, we follow the setups in \cite{liu2018toward} for the convenience of comparison.
\begin{itemize}
    \item The ISAC base station is located at the origin $(0, 0)$;
    \item $K = 4$, $N = 16$, $L = 30$; cf. \eqref{eq:communication-model};
    \item The constellation $\mat S$ is constructed using $4$-point quadrature phase-shift keying (QPSK); 
    \item The two targets are located with azimuths of $(-45, 45)$ degrees relative to the vertical line of sight of the base station, and with randomly generated ranges of $(1.51, 1.39)$ kilometers (Km) from the base station, respectively;
    \item The expected beamwidth at each target direction is $10$ degrees;
    \item The locations of the four downlink communication users are randomly generated according to the uniform distribution on $[0,1] \times [-1, 1]$ Km$^2$.
\end{itemize}

We suppose that the nominal channel $\matb H$ and the true channel $\mat H_0$ are generated from the small perturbations of a reference channel $\mat H_{\text{ref}}$, 
that is,
\begin{equation}\label{eq:simulation-engine}
\begin{array}{cl}
\mat H_0 &= \mat H_{\text{ref}} + \epsilon \cdot \mat \Delta_1, \\
\matb H &= \mat H_{\text{ref}} + \epsilon \cdot \mat \Delta_2,
\end{array}
\end{equation}
where $\mat H_{\text{ref}}$ is simulated using the MATLAB toolboxes and fixed throughout the experiments; $\matb H$ is also fixed after channel estimation; $\epsilon \defeq 0.05$ is a small number and the perturbation matrices $\mat \Delta_1$ and $\mat \Delta_2$ are distributed according to entry-wise standard complex Gaussian. Note that the true channel $\mat H_0$ randomly varies from one Monte--Carlo test to another, however, it is close to the fixed nominal channel $\matb H$. Note also that waveform designs are based on the nominal channel $\matb H$, while performance tests are on the true channels $\mat H_0$ because, in simulations, $\mat H_0$ is known. In practice, the channel's uncertainty quantification parameter $\theta$ such that $\|\mat H_0 - \matb H\| \le \theta$ is unknown for a specified norm $\|\cdot\|$ because we only know $\matb H$ and all other information about the true channel $\mat H_0$ is missing.\footnote{The relation $\|\mat H_0 - \matb H\| \le \theta$ holds in probability due to \eqref{eq:simulation-engine}; cf. Fact \ref{fact:bound-in-prob}.} Hence, in real-world operation, $\theta$ can be a tuning parameter unless it is specified in the channel estimation stage. In the experiments, $\|\cdot\|$ is particularized to the matrix F-norm or the vector $2$-norm; cf. Proposition \ref{prop:sub-prob-2}.

Note that the parameter configurations described above do not essentially impact the experimental results in this paper. Therefore, these setups are merely employed as demonstrative examples; one may use the shared source codes to change the configurations to verify the claims in this section.

\subsection{Experimental Results}

\subsubsection{Perfect Waveforms for Sensing and Communication} 
The beam pattern for perfect sensing is shown in Fig. \ref{fig:perfect-sensing-waveform}, where two high-directional beams point to the two targets located with azimuths of $(-45, 45)$ degrees, respectively. For each target, the beam width is $10$ degrees. The theoretically perfect beam pattern for sensing is not practically achievable, and therefore, the mean-squared-error minimization method \cite{fuhrmann2008transmit} is used to design the practically perfect beam pattern (i.e., the beampattern-inducing matrix $\mat R$). When $\mat R$ is available, the cyclic algorithm \cite{stoica2008waveform} is used to design the perfect sensing waveform $\mat X_s$ with constrained peak-to-average power ratio (PAPR); recall \eqref{eq:perfect-sensing}.

\begin{figure}[!htbp]
    \centering
	\subfigure[Perfect Sensing]{
	 	\includegraphics[height=3cm,width=4.1cm]{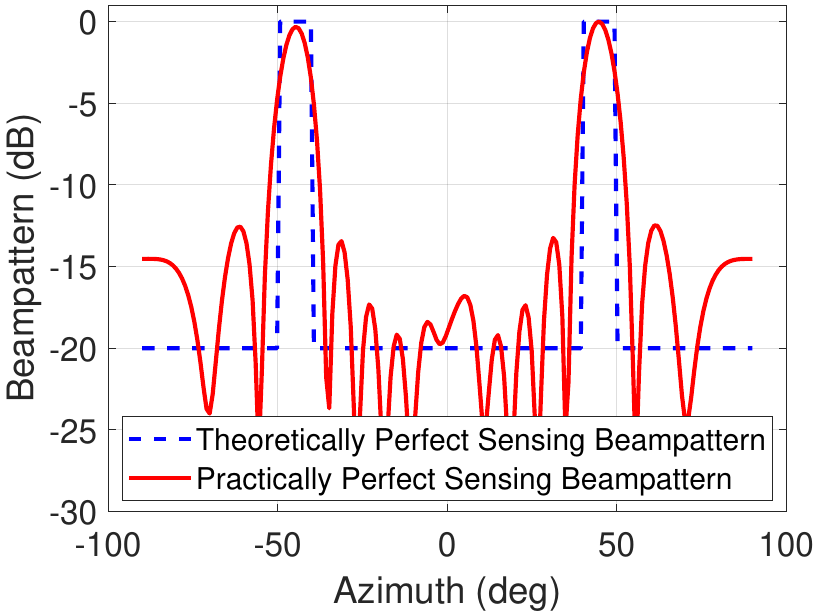}
        \label{fig:perfect-sensing-waveform}
	}
	\subfigure[Perfect Communication (for $\matb H$)]{
	 	\includegraphics[height=3cm,width=4.1cm]{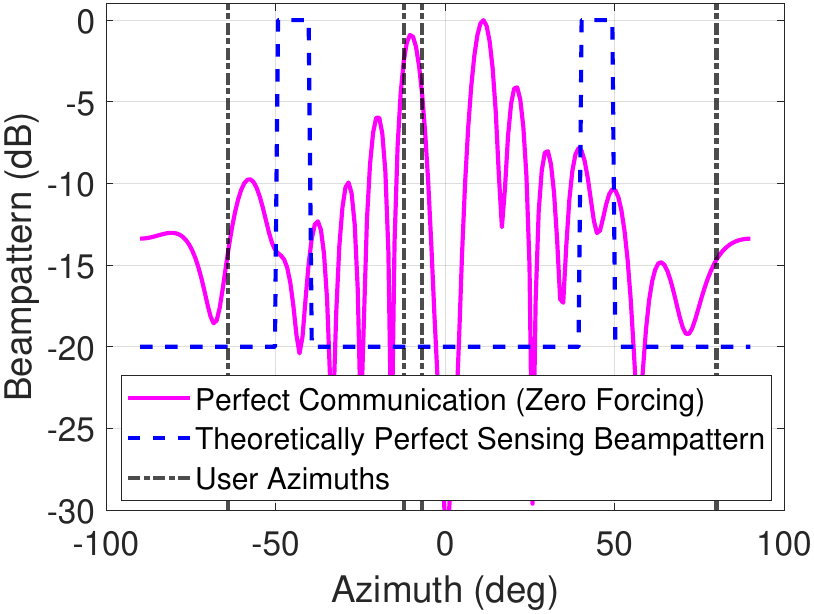}
        \label{fig:perfect-commu-waveform}
	}

    \caption{The beam patterns for perfect sensing and perfect communication, respectively. The perfect communication, obtained using zero-forcing precoding, is specific to the nominal channel $\matb H$.}
    \label{fig:perfect-waveforms}
\end{figure}

As for perfect communication, in the sense of maximum AASR $R_{\mat H_0, \mat X_0}$ defined in \eqref{eq:AASR}, it cannot be specified because $R_{\mat H_0, \mat X_0}$ depends on the true but unknown channel $\mat H_0$ and the associated truly optimal waveform $\mat X_0$. For the purpose of illustration, we show the perfect communication waveform under the nominal channel, that is, the zero-forcing waveform $\mat X_c \defeq \matb H^{\H} (\matb H \matb H^{\H})^{-1} \mat S$; see Fig. \ref{fig:perfect-commu-waveform}. The zero-forcing waveform is perfect for communication in the sense of interference mitigation (i.e., MUI cancellation). We see that to reduce MUI energy, the communication beams do not necessarily point to users. In addition, the optimal waveform for communication is indeed not necessarily the same as the optimal waveform for sensing; see Table \ref{tab:perfect-performances}. The detection probability of the target at the azimuth of $-45$ degrees is calculated using \cite[Eq.~(69)]{khawar2015target}.

\begin{table}[!htbp]
\centering
\caption{System Performances for Separate Waveform Design}
\begin{tabular}{l|l|l}
\hline
\textbf{Performance Metric}                                                                            & \textbf{Waveform} & \textbf{Value} \\ \hline
\multirow{2}{*}{\begin{tabular}[c]{@{}l@{}}Detection Probability\\ (for Sensing)\end{tabular}} & $\mat X_s$         &  1     \\
                                                                                  & $\mat X_c$         &  0.38     \\ \hline
\multirow{2}{*}{\begin{tabular}[c]{@{}l@{}}AASR\\ (for Communication)\end{tabular}}          & $\mat X_s$         &  0.30 (bps/Hz/user)     \\
                                                                                  & $\mat X_c$         &  3.07 (bps/Hz/user)     \\ \hline
\multicolumn{3}{l}{Azimuth of Target to Detect: $-45$ Degrees}
\end{tabular}
\label{tab:perfect-performances}
\end{table}

\subsubsection{Nominal Waveform Design}
In this subsection, we show the practical issues of the nominal waveform design methods \eqref{eq:waveform-design} and \eqref{eq:waveform-design-tradeoff} using $\matb H$. To be specific, we demonstrate the issues of the nominal Pareto frontiers (specified by nominal designs) and the necessity of uncertainty-aware performance characterization (i.e., robust/conservative Pareto frontier). Experimental results are shown in Fig. \ref{fig:nominal-issues}; cf. Figs. \ref{fig:pareto-frontiers} and \ref{fig:true-frontiers}.

\begin{figure}[!htbp]
    \centering
	\subfigure[Sensing-Centric Design]{
	 	\includegraphics[height=3cm]{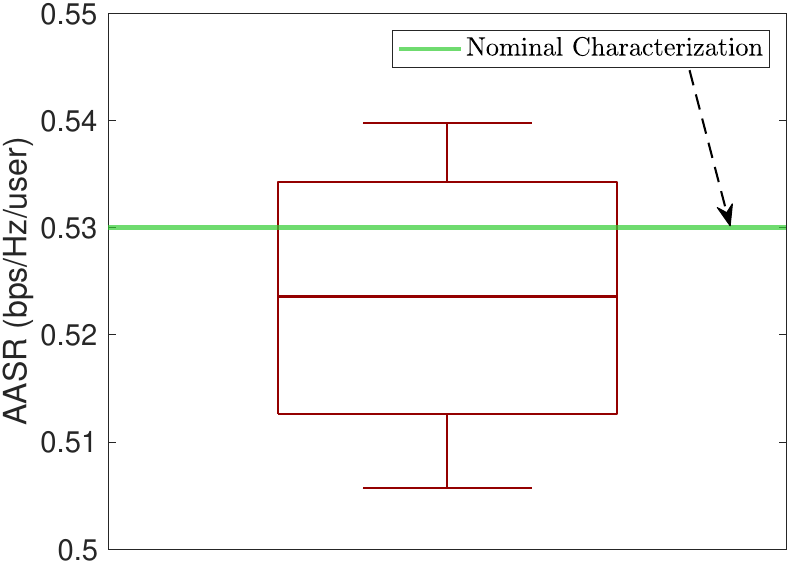}
        \label{fig:nominal-issues-sensingcentric}
	}
 
    \subfigure[Joint Design (TPC; $\rho \le 0.1$)]{
	 	\includegraphics[height=3cm]{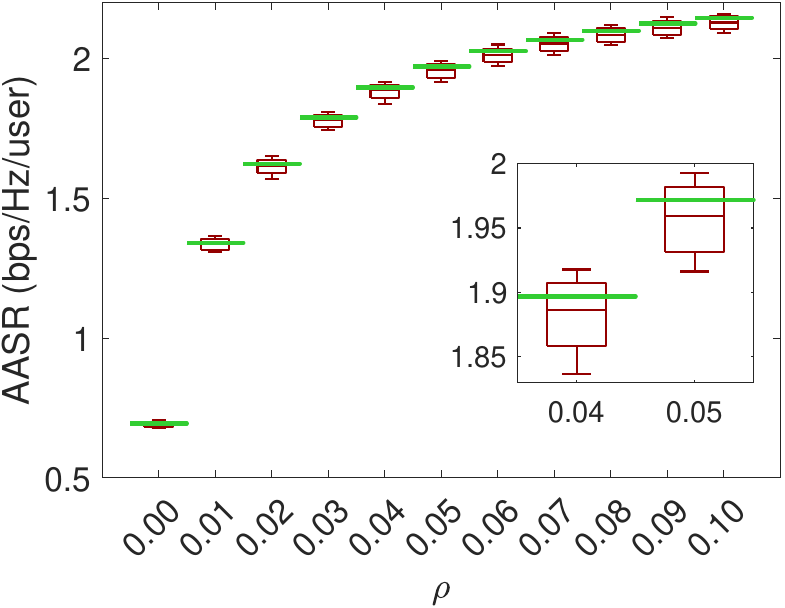}
        \label{fig:joint-TPC-nominal-issues-0.01}
	}
	\subfigure[Joint Design (TPC; $\rho \ge 0.1$)]{
	 	\includegraphics[height=3cm]{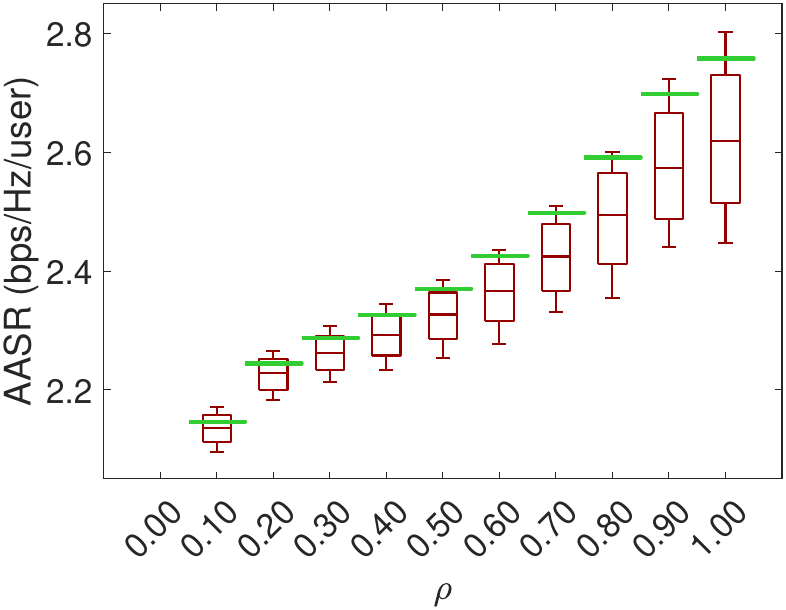}
        \label{fig:joint-TPC-nominal-issues-0.1}
	}
 
    \subfigure[Joint Design (PAPC; $\rho \le 0.1$)]{
	 	\includegraphics[height=3cm]{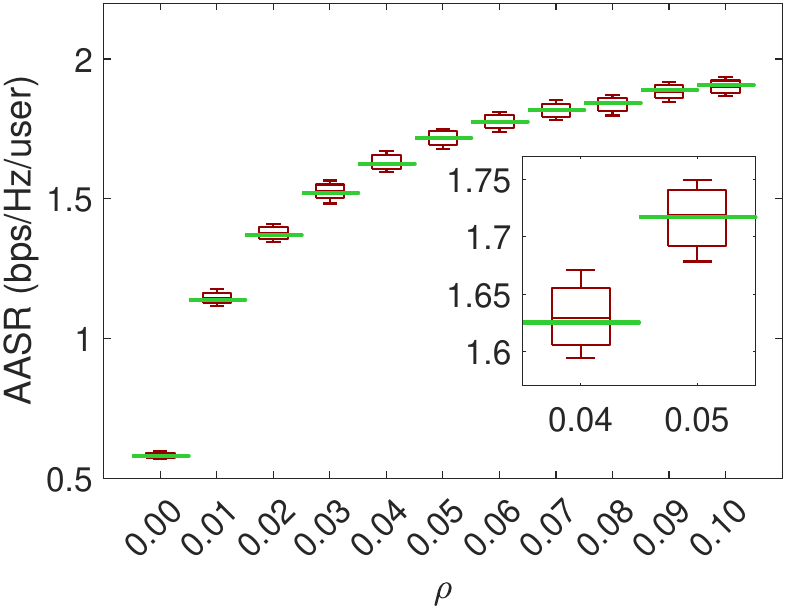}
        \label{fig:joint-PAPC-nominal-issues-0.01}
	}
	\subfigure[Joint Design (PAPC; $\rho \ge 0.1$)]{
	 	\includegraphics[height=3cm]{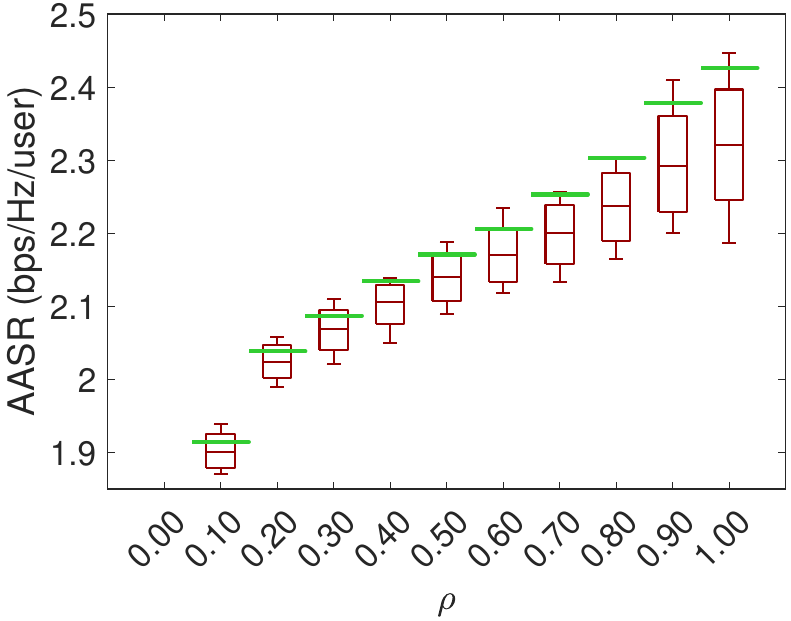}
        \label{fig:joint-PAPC-nominal-issues-0.1}
	}
    \caption{Nominal Pareto frontiers (i.e., nominal AASR performances; green) and true AASR performances (box plots). In (b)-(e), sensing-centric and communication-centric designs are obtained for $\rho = 0$ and $\rho = 1$, respectively; cf. \eqref{eq:waveform-design-tradeoff}. (Cf. Fig. \ref{fig:true-frontiers-a}.)}
    \label{fig:nominal-issues}
\end{figure}

In Fig. \ref{fig:nominal-issues}, the box plots of the true communication performances $R_{\mat H_0, \matb X}$ using the nominally optimal waveform $\matb X$ under $1000$ independent Monte--Carlo simulations (i.e., $1000$ different realizations of $\mat H_0$) are shown. The box plots display the minimum value, $5\%$ percentile, median, $95\%$ percentile, and the maximum value among all AASRs $R_{\mat H_0, \matb X}$. The green lines indicate the nominal characterization of sensing-communication performances, which, however, cannot fully depict the actual performance trade-off; cf. Fig \ref{fig:true-frontiers-a}.

In addition, the box plots in Fig. \ref{fig:nominal-issues} support that the joint design methods can improve communication performances by sacrificing sensing performances via changing the trade-off parameter $\rho$. Moreover, the joint design under the total power constraint (TPC) brings higher AASR than that under the per-antenna power constraint (PAPC).

\subsubsection{Robust Sensing-Centric Waveform Design \captext{\eqref{eq:robust-waveform-design}}}
The experimental results of sensing-centric robust waveform designs are shown in Fig. \ref{fig:sensing-centric}, where the box plots of the true communication performances $R_{\mat H_0, \mat X^*}$ using the robust waveforms $\mat X^*$ with different $\theta$s are shown. As we can see, in robust designs, the radius $\theta$ of the uncertainty set $\cal H$ should be elegantly specified: if $\theta$ is overly large, the associated performance bound would be extremely conservative (i.e., loose), while if $\theta$ is overly small, the performance lower bound cannot be correctly identified; recall Subsections \ref{subsec:price-robustness} and \ref{subsec:budget-uncertainty-set}.

\begin{figure}[!htbp]
    \centering
    \subfigure[Method \ref{method:max-form-solution} against $\theta$]{
	 	\includegraphics[height=3cm]{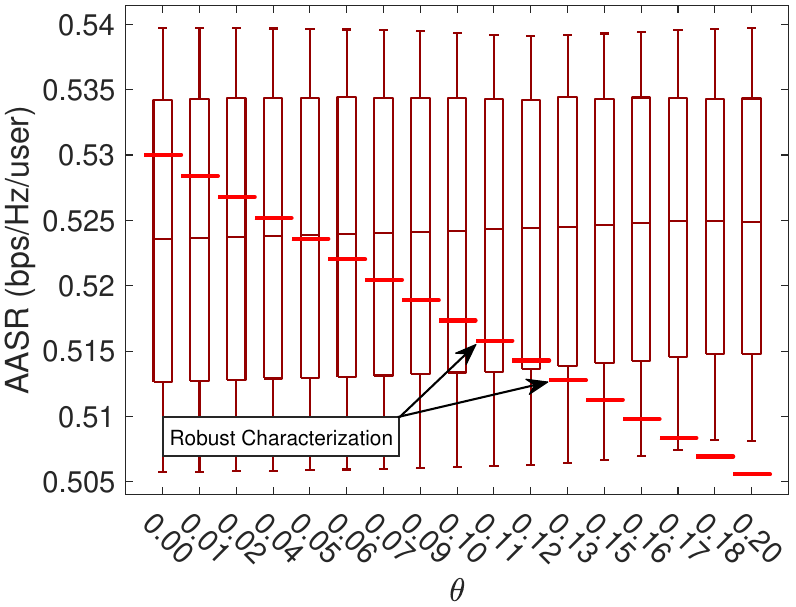}
        \label{fig:sensingcentric-robust-method1}
	}
	\subfigure[Method \ref{method:max-form-solution} ($\theta = 0.127$)]{
	 	\includegraphics[height=3cm]{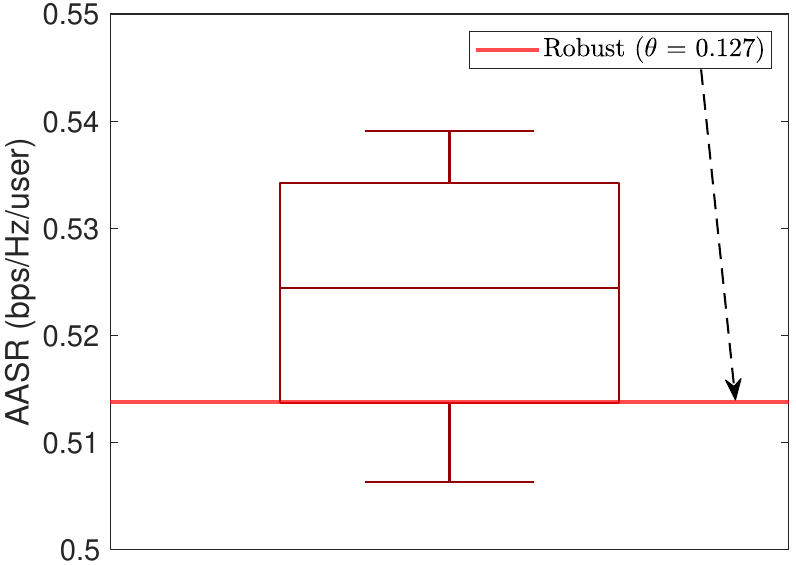}
        \label{fig:sensingcentric-robust-method1-0.127}
	}

     \subfigure[Method \ref{method:max-form-solution-another} against $\theta$]{
	 	\includegraphics[height=3cm]{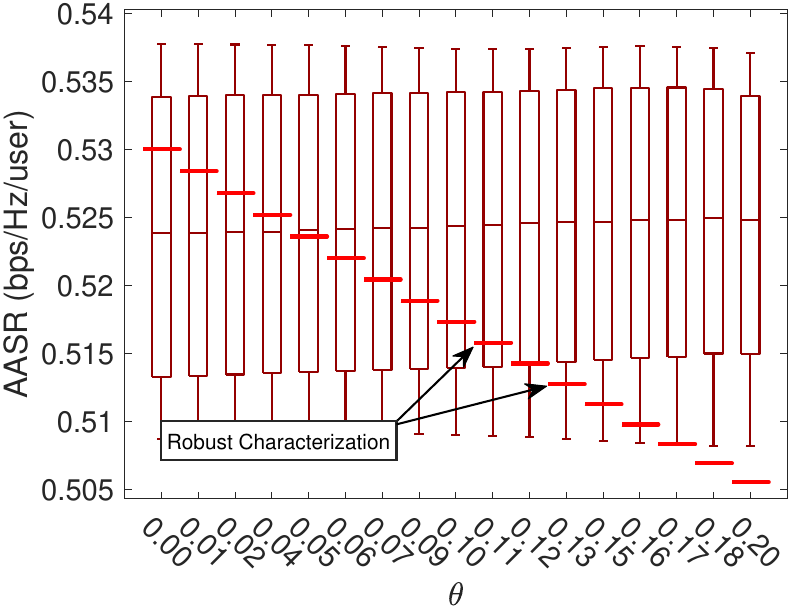}
        \label{fig:sensingcentric-robust-method2}
	}
	\subfigure[Method \ref{method:max-form-solution-another} ($\theta = 0.127$)]{
	 	\includegraphics[height=3cm]{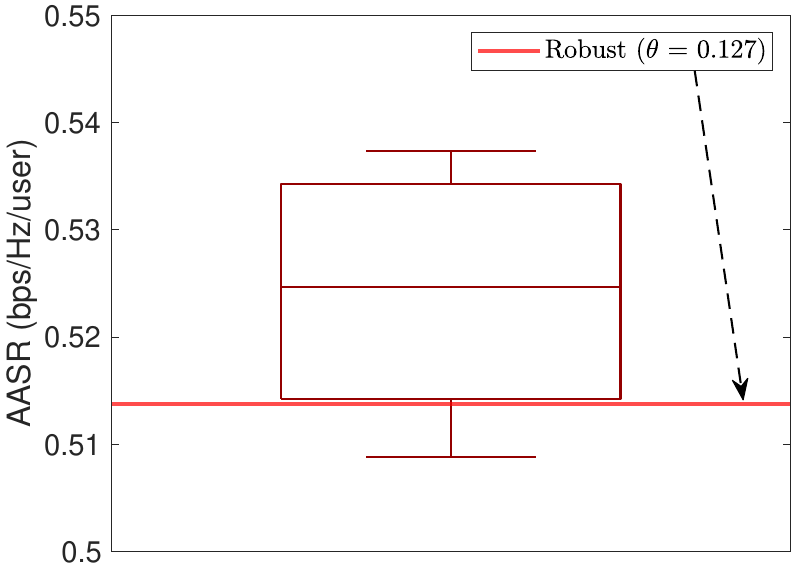}
        \label{fig:sensingcentric-robust-method2-0.127}
	}
    \caption{Conservative characterization under robust sensing-centric waveform design. Red step values show the robust performance characterizations $R_{\mat H^*, \mat X^*}$ under different $\theta$. When $\theta \approx 0.127$, $R_{\mat H^*, \mat X^*}$ correctly identifies the tightest conservative performance bounds of true AASRs, with probability $95\%$. NB: $\theta = 0$ gives nominal characterizations. (Cf. Fig. \ref{fig:true-frontiers-b}.)}
    \label{fig:sensing-centric}
\end{figure}

In Fig. \ref{fig:sensingcentric-beampattern}, we show the beam patterns of the three waveforms: the perfect-sensing waveform $\mat X_s$, the nominal waveform $\bar{\mat X}$, and the robust waveform $\mat X^*$ when $\theta = 0.127$. Since this is a sensing-centric design scheme, all waveforms are guaranteed to have the same beam pattern.
\begin{figure}[!htbp]
    \centering
    \subfigure[Method \ref{method:max-form-solution}]{
	 	\includegraphics[height=3cm,width=4.1cm]{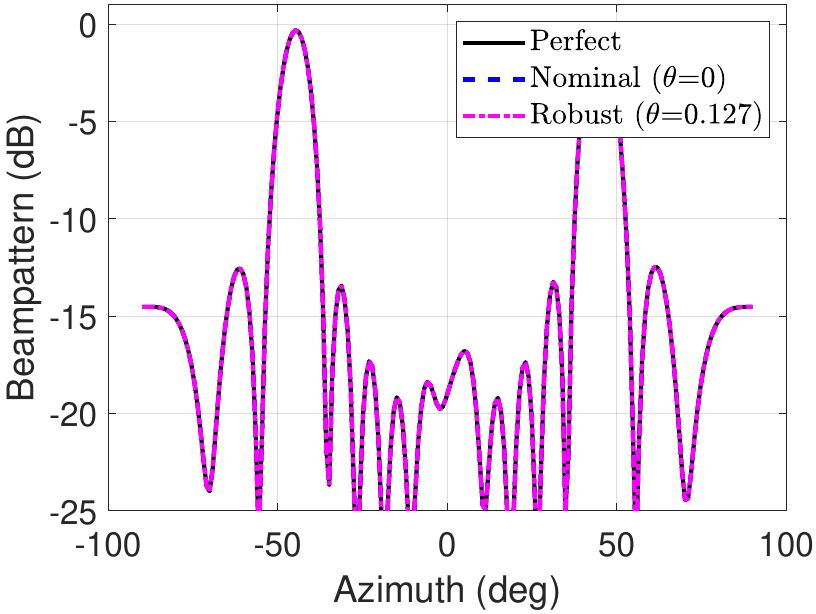}
        \label{fig:sensingcentric-beampattern-method1}
	}
	\subfigure[Method \ref{method:max-form-solution-another}]{
	 	\includegraphics[height=3cm,width=4.1cm]{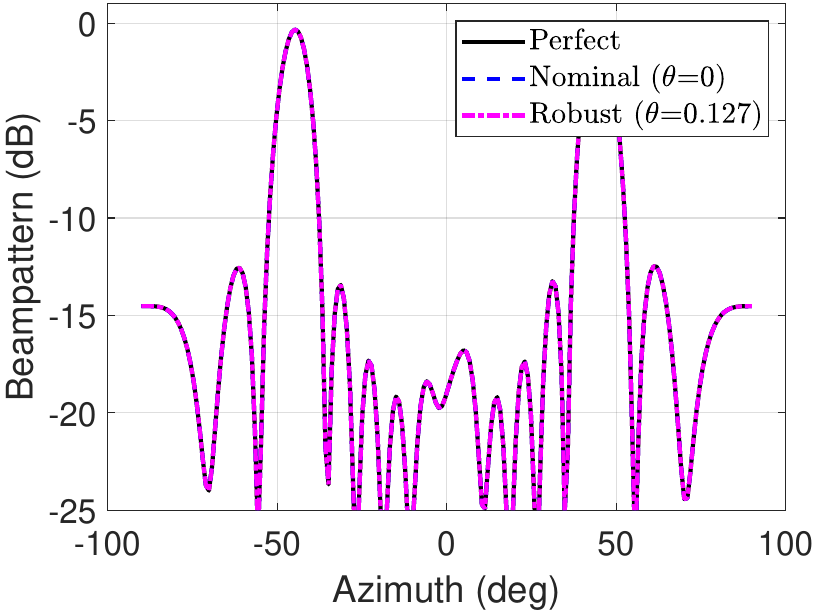}
        \label{fig:sensingcentric-beampattern-method2}
	}
    \caption{Beam patterns of sensing-centric ISAC waveforms (including perfect-sensing, nominal, and robust waveforms); the two sub-figures are identical.}
    \label{fig:sensingcentric-beampattern}
\end{figure}

\begin{figure}[!htbp]
    \centering
    \subfigure[TPC against $\theta$ ($\rho = 0.25$)]{
	 	\includegraphics[height=3cm]{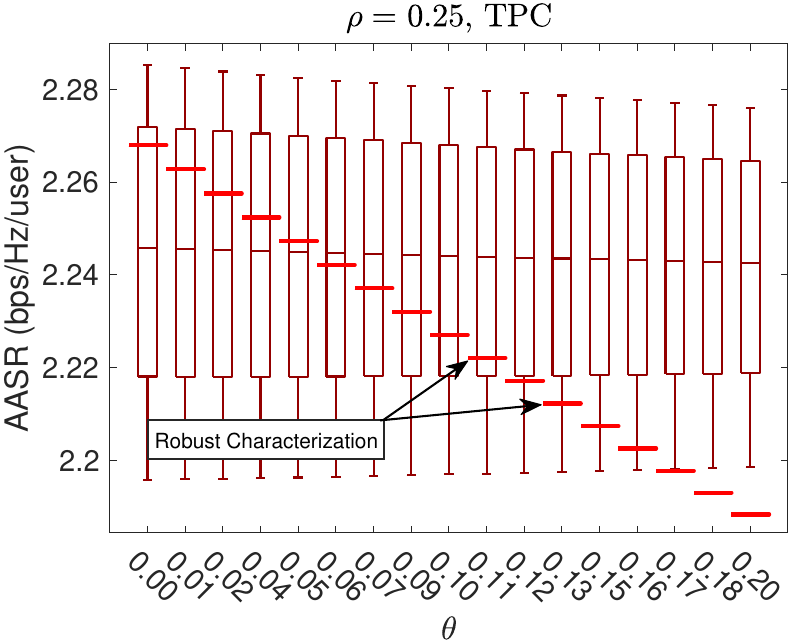}
        \label{fig:joint-TPC-robust-rho-0.25}
	}
	\subfigure[TPC ($\theta=0.16$, $\rho = 0.25$))]{
	 	\includegraphics[height=3cm]{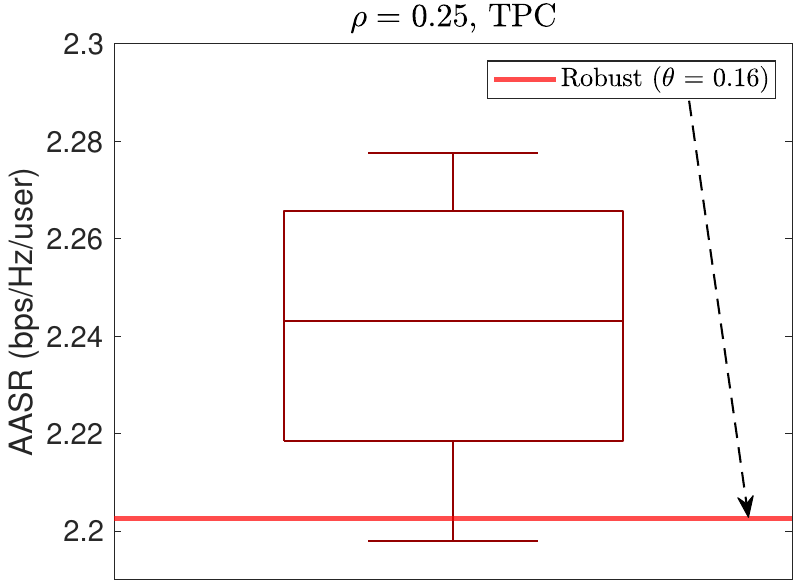}
        \label{fig:joint-TPC-robust-rho-0.25-theta-0.16}
	}
    
    \subfigure[TPC ($\theta = 0.16$, $\rho \le 0.1$)]{
	 	\includegraphics[height=3cm]{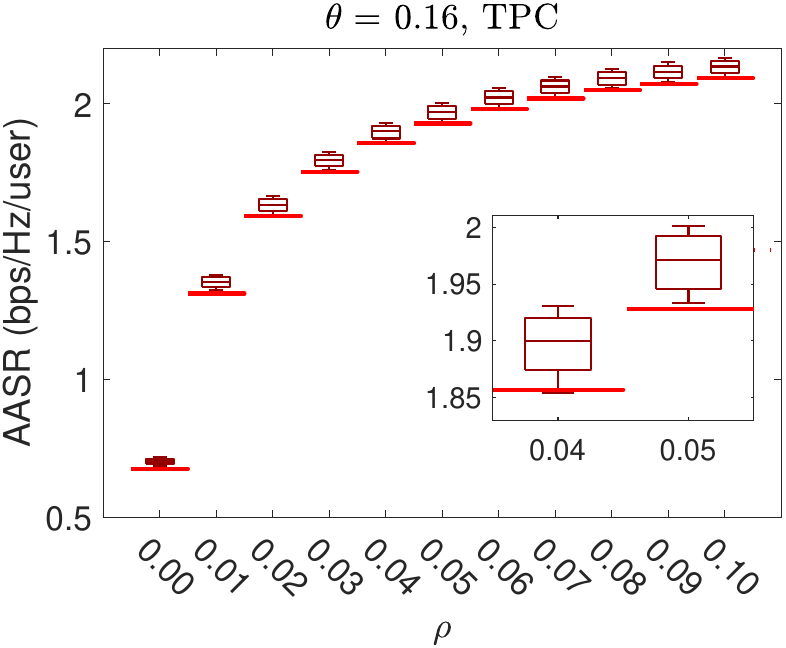}
        \label{fig:joint-TPC-robust-0.00-0.01}
	}~~~~
	\subfigure[TPC ($\theta = 0.16$, $\rho \ge 0.1$)]{
	 	\includegraphics[height=3cm]{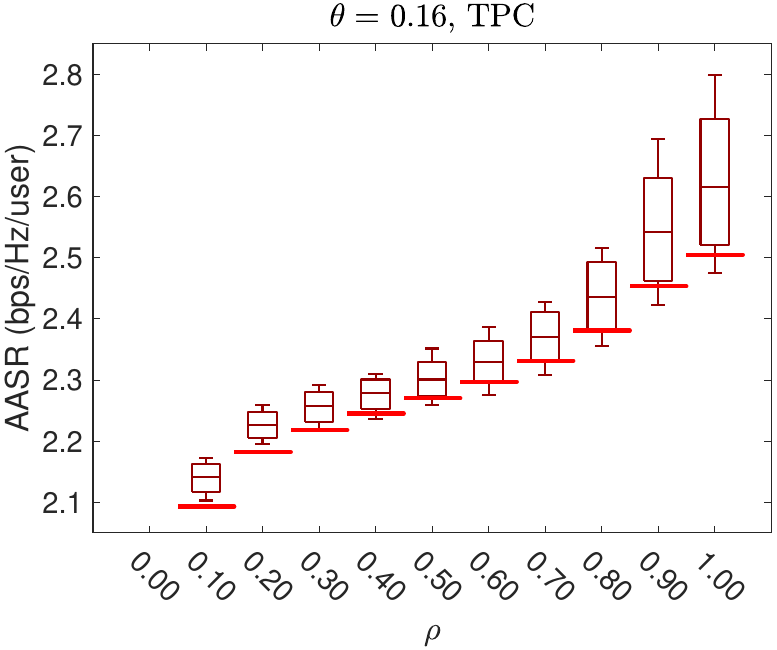}
        \label{fig:joint-TPC-robust-0.1-1.0}
	}
    \caption{Conservative characterization under robust joint waveform design subject to total power constraint (TPC). When $\theta = 0.16$, the tight conservative performance boundary is identified, with probability $95\%$. (Cf. Fig. \ref{fig:true-frontiers-b}.)}
    \label{fig:tradeoff-design-TPC}
\end{figure}

\begin{figure}[!htbp]
    \centering
    \subfigure[PAPC against $\theta$ ($\rho = 0.25$)]{
	 	\includegraphics[height=3cm]{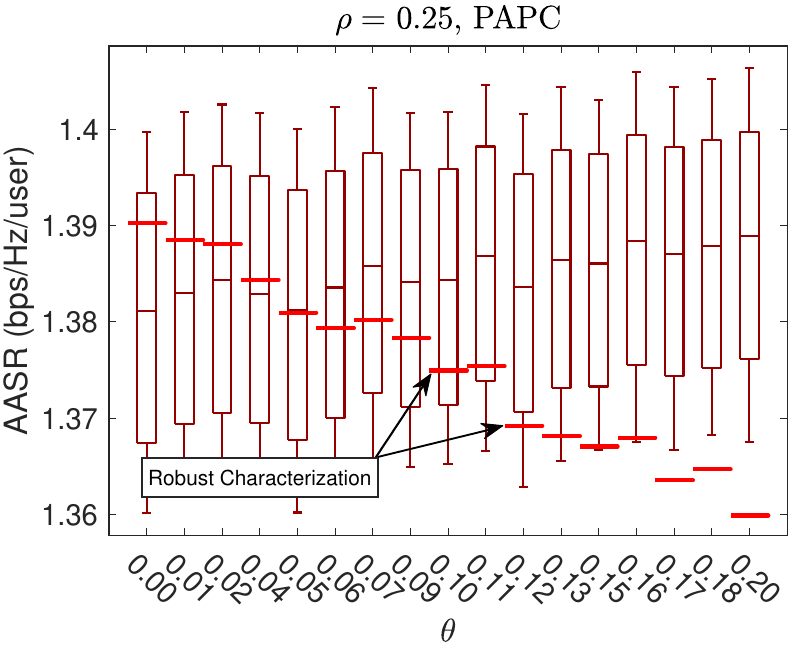}
        \label{fig:joint-PAPC-robust-rho-0.25}
	}
	\subfigure[PAPC ($\theta = 0.16$, $\rho = 0.25$)]{
	 	\includegraphics[height=3cm]{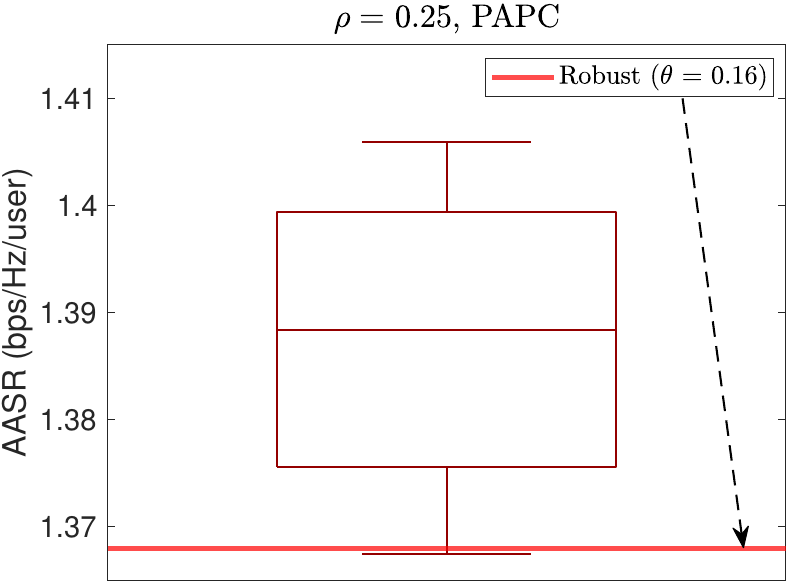}
        \label{fig:joint-PAPC-robust-rho-0.25-theta-0.16}
	}
    
    \subfigure[PAPC ($\theta = 0.16$, $\rho \le 0.1$)]{
	 	\includegraphics[height=3cm]{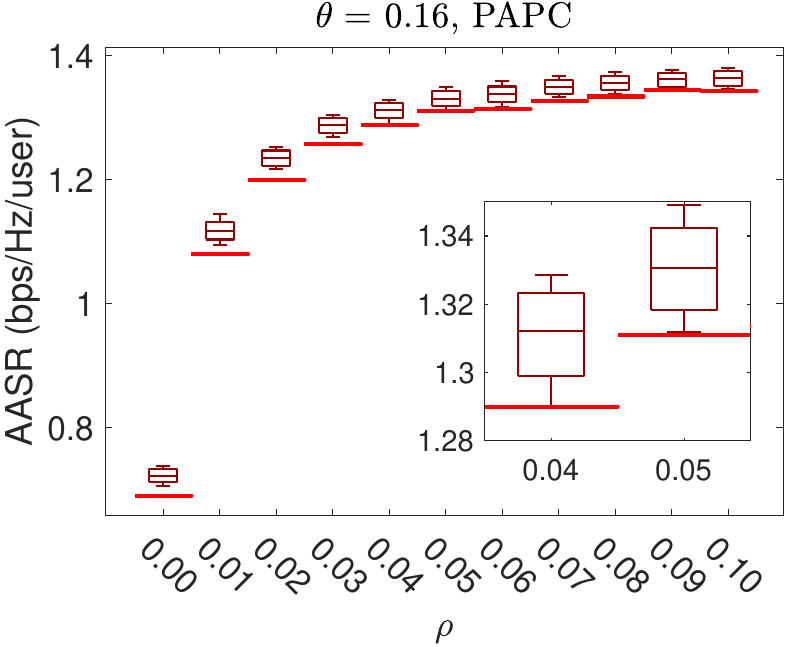}
        \label{fig:joint-PAPC-robust-0.00-0.01}
	}~~~~
	\subfigure[PAPC ($\theta = 0.16$, $\rho \ge 0.1$)]{
	 	\includegraphics[height=3cm]{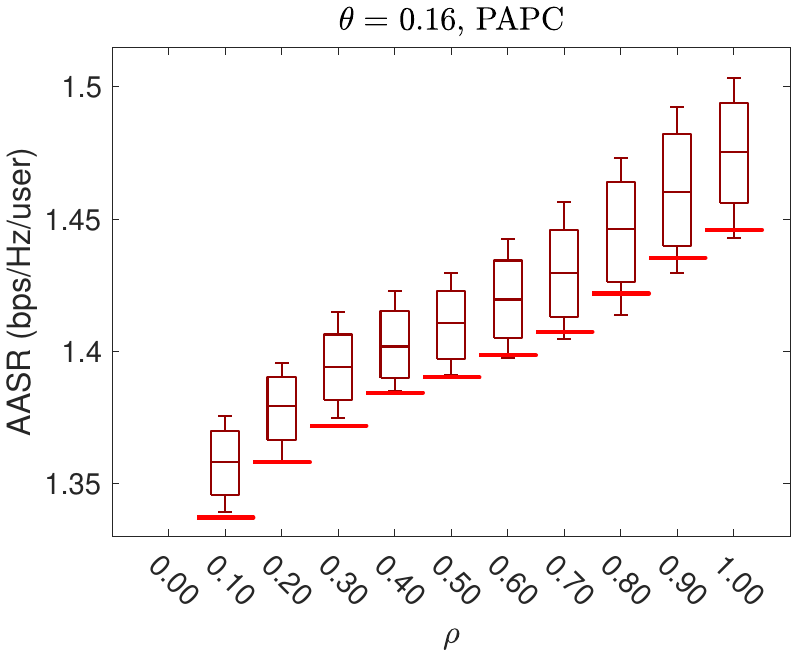}
        \label{fig:joint-PAPC-robust-0.1-1.0}
	}
    \caption{Conservative characterization under robust joint waveform design subject to per-antenna power constraint (PAPC). When $\theta = 0.16$, the tight conservative performance boundary is identified, with probability $95\%$. (Cf. Fig. \ref{fig:true-frontiers-b}.) Since high computational burdens exist at each iteration (cf. Table \ref{tab:running-times}), the solution algorithm is forced to terminate after limited iterations. As a result, the performances are not globally optimal and stable; see, e.g., (a).}
    \label{fig:tradeoff-design-PAPC}
\end{figure}

\subsubsection{Robust Joint Waveform Design \captext{\eqref{eq:robust-waveform-design-tradeoff}}}

The experimental results of robust joint waveform design, subject to TPC and PAPC, are shown in Figs. \ref{fig:tradeoff-design-TPC} and \ref{fig:tradeoff-design-PAPC}, respectively. In both cases, conservative (i.e., robust) performance boundaries are successfully identified. As joint design schemes, nominally optimal waveforms and robust waveforms are not guaranteed to have the same beam pattern as the perfect waveform for sensing; see Fig. \ref{fig:joint-beampattern}.
\begin{figure}[!htbp]
    \centering
    \subfigure[TPC ($\rho = 0.25$)]{
	 	\includegraphics[height=3cm]{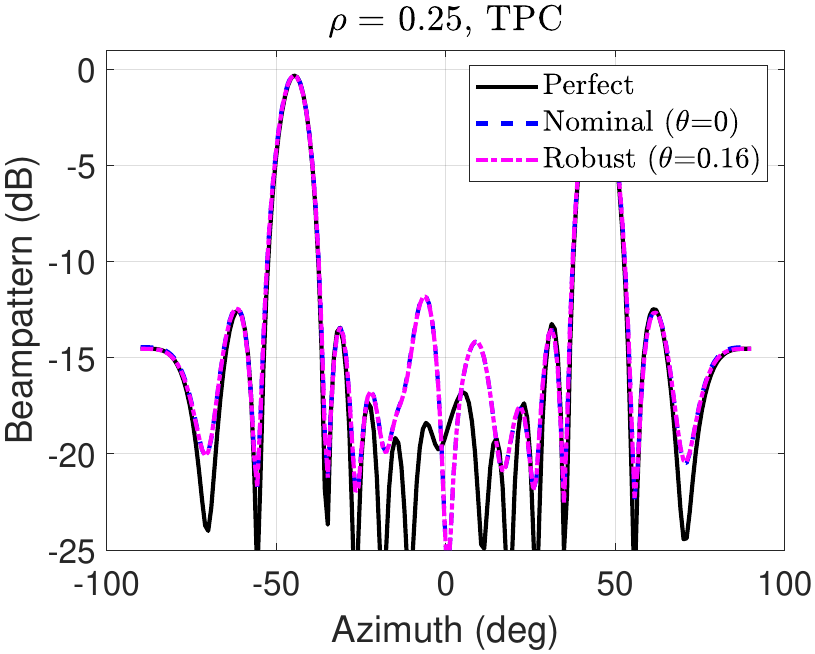}
        \label{fig:joint-beampattern-TPC-0.25}
	}
     \subfigure[TPC ($\rho = 0.75$)]{
	 	\includegraphics[height=3cm]{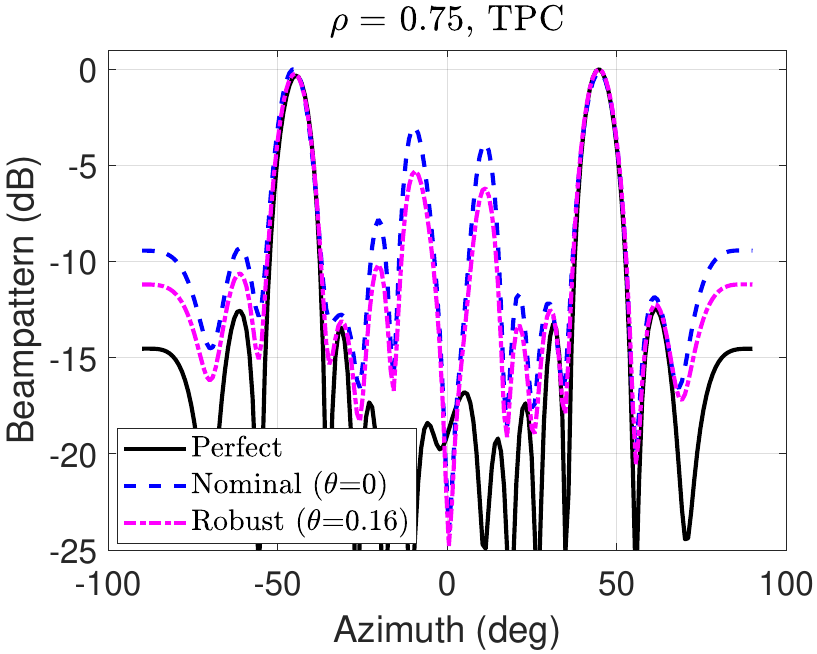}
        \label{fig:joint-beampattern-TPC-0.75}
	}
 
	\subfigure[PAPC ($\rho = 0.25$)]{
	 	\includegraphics[height=3cm]{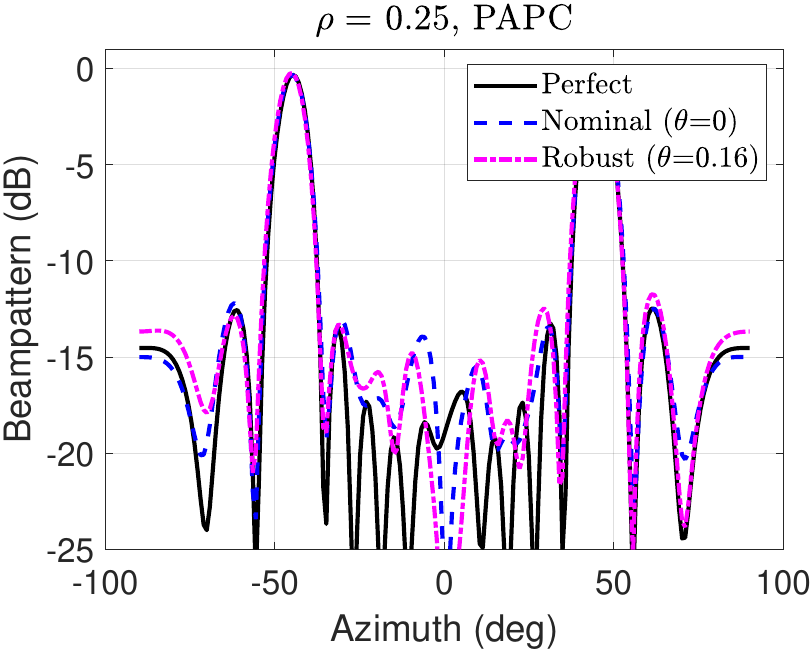}
        \label{fig:joint-beampattern-PAPC-0.25}
	}
	\subfigure[PAPC ($\rho = 0.75$)]{
	 	\includegraphics[height=3cm]{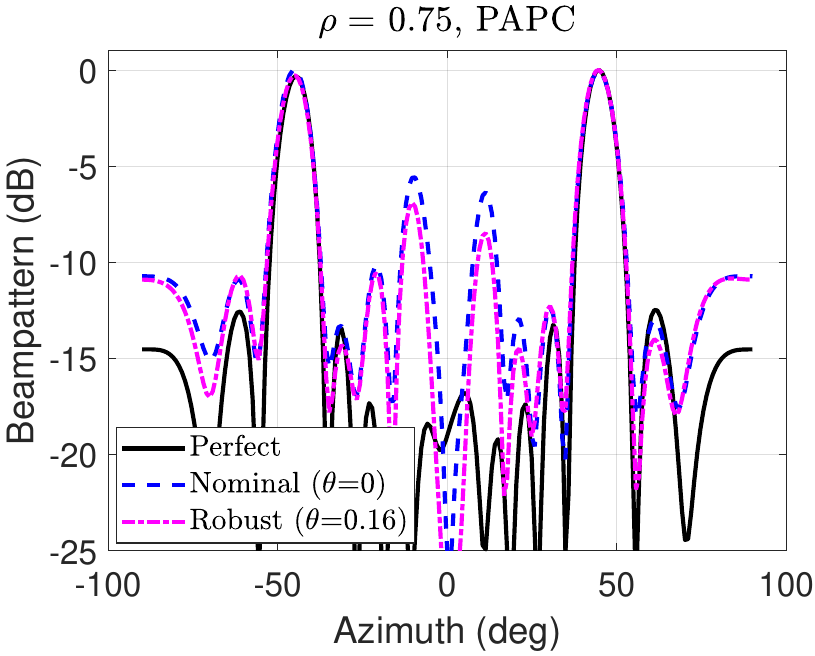}
        \label{fig:joint-beampattern-PAPC-0.75}
	}
    \caption{Beam patterns of joint-design ISAC waveforms (including perfect-sensing, nominal, and robust waveforms).}
    \label{fig:joint-beampattern}
\end{figure}

\subsubsection{Results on Algorithmic Convergence and Computation Burdens}\label{subsubsec:convergence-computation}

In Propositions \ref{prop:convex-quadratic-max} and \ref{thm:solution-remedy-problem}, iterative processes are involved. However, experiments show that the two processes converge extremely fast on average, almost within six or seven iterations. This is attributed to the existence of closed-form solutions at each iteration. One illustration is given in Fig. \ref{fig:convergence-check}.

\begin{figure}[!htbp]
    \centering
	\subfigure[Iteration in Proposition \ref{prop:convex-quadratic-max}]{
	 	\includegraphics[height=3cm]{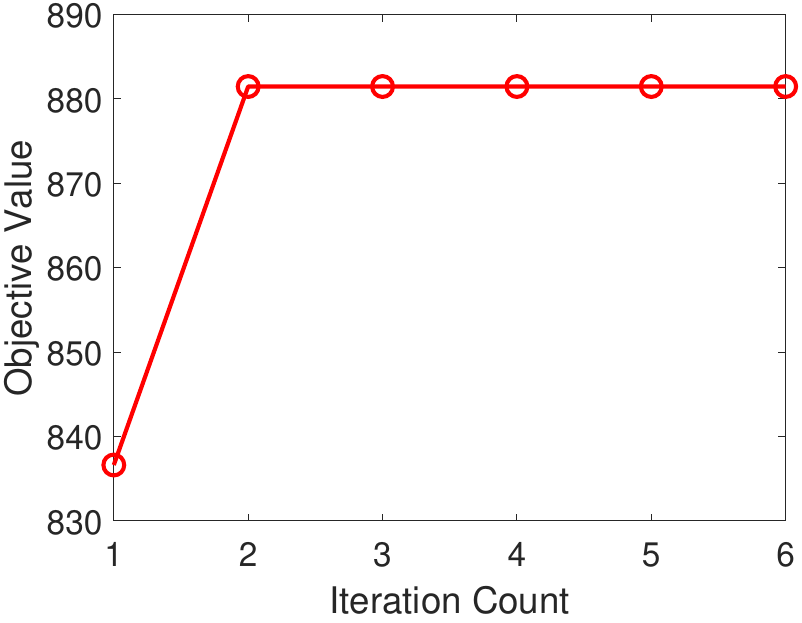}
        \label{fig:convex-max-converge}
	}
	\subfigure[Iteration in Proposition \ref{thm:solution-remedy-problem}]{
	 	\includegraphics[height=3cm]{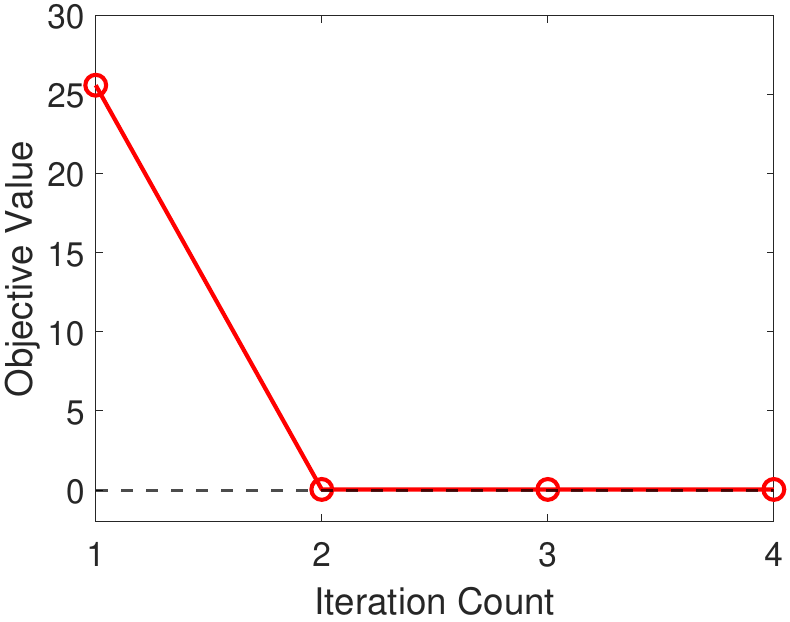}
        \label{fig:remedy-converge}
	}

    \caption{Empirical convergence illustrations of the iteration processes in Propositions \ref{prop:convex-quadratic-max} and \ref{thm:solution-remedy-problem}, terminating within six steps on average. Note that at each iteration step, closed-form solutions exist.}
    \label{fig:convergence-check}
\end{figure}

The average running times of involved methods are shown in Table \ref{tab:running-times}; see Table \ref{tab:solution-method} for the components of each method. Table \ref{tab:running-times} suggests the following: 1) non-robust nominal designs consume significantly fewer running times than their robust counterparts; 2) the joint design method under PAPC has significantly more computational burdens than the sensing-centric design and the joint design under TPC; 3) the sensing-centric design and the joint design under TPC are computationally comparable; 4) Methods \ref{method:max-form-solution}, \ref{method:max-form-solution-another}, and \ref{method:robust-joint-design} under TPC are computationally comparable.  
\begin{table}[!htbp]
\centering
\caption{Average Running Times of Design Methods; cf. Table \ref{tab:solution-method}}
\begin{tabular}{l|c|c|c|c|c|c|c}
\hline
         & \multicolumn{3}{c|}{\textbf{Sensing Centric}}  &  \multicolumn{2}{c|}{\textbf{Joint (TPC)}}  &  \multicolumn{2}{c}{\textbf{Joint (PAPC)}}   \\   
         \cline{2-8}
         &  N.    &   M. 1    &  M. 2        &   N.     &   M. 3      &   N.     &   M. 3   \\
\hline 
\textbf{Times}    &  0.36    &   5.83    &  4.60        &   0.69     &   5.76      &   132.35     &   180.17   \\
\hline
\multicolumn{8}{l}{
\tabincell{l}{
N.: \underline{N}ominal Designs;~~~~~~~~~~~~~M.: \underline{M}ethod (Robust Designs)\\
Time Unit: milliseconds (ms)
}}
\end{tabular}
\label{tab:running-times}
\end{table}

\subsubsection{Results on Different \captext{$\epsilon$}}
In previous experiments, we use $\epsilon = 0.05$; cf. \eqref{eq:simulation-engine}. In this subsection, we offer more empirical results against different values of $\epsilon$ to show the applicability of the proposed robust method for different uncertainty levels. As demonstrations, we adopt the sensing-centric design scheme; see Fig. \ref{fig:different-epsilon}. From Fig. \ref{fig:different-epsilon}, the following two observations can be outlined: First, the larger the uncertainty level $\epsilon$, the more dispersive the AASRs; second, for different uncertainty levels $\epsilon$, the proposed robust design method can provide conservative characterizations under appropriate $\theta$s (i.e., identify the tight performance lower bounds with probability $95\%$). For more discussions, see Appendix R in supplementary materials.

\begin{figure}[!htbp]
    \centering
	\subfigure[Nominal ($\epsilon = 0.001$)]{
	 	\includegraphics[height=3cm]{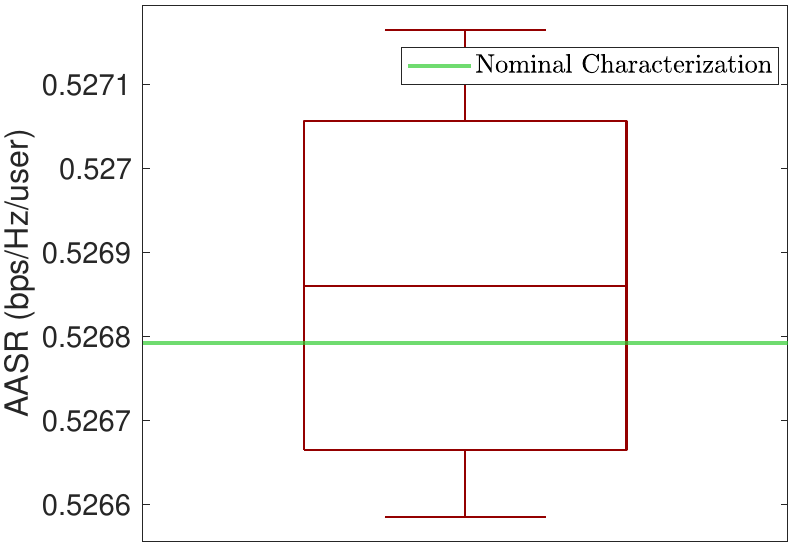}
	}
	\subfigure[Conservative ($\epsilon = 0.001$)]{
	 	\includegraphics[height=3cm]{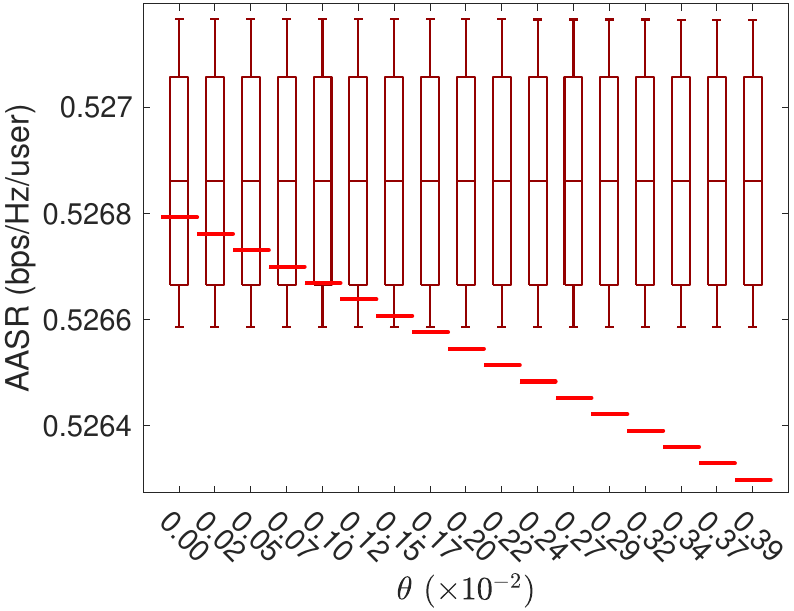}
	}

	\subfigure[Nominal ($\epsilon = 0.01$)]{
	 	\includegraphics[height=3cm]{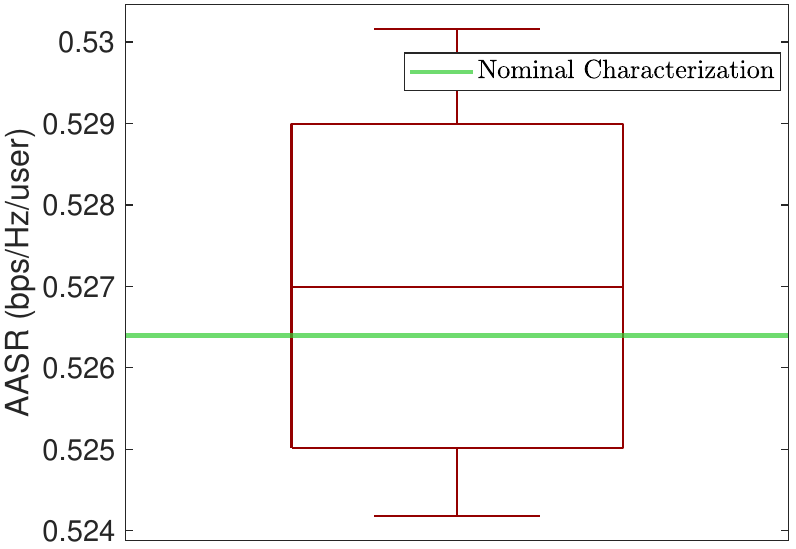}
	}
	\subfigure[Conservative ($\epsilon = 0.01$)]{
	 	\includegraphics[height=3cm]{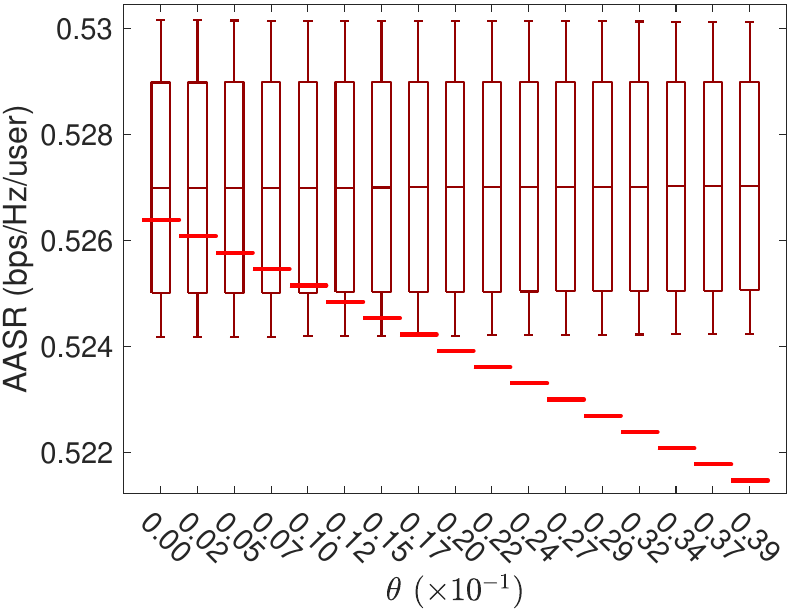}
	}

	\subfigure[Nominal ($\epsilon = 0.1$)]{
	 	\includegraphics[height=3cm]{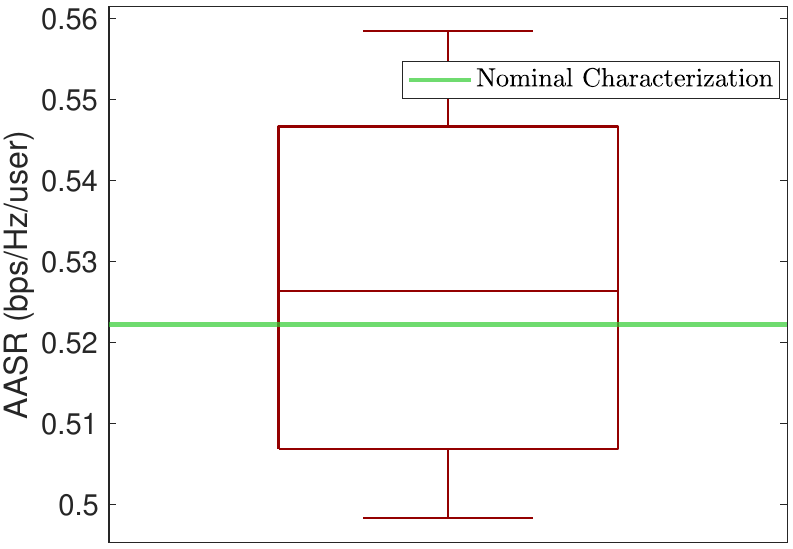}
	}
	\subfigure[Conservative ($\epsilon = 0.1$)]{
	 	\includegraphics[height=3cm]{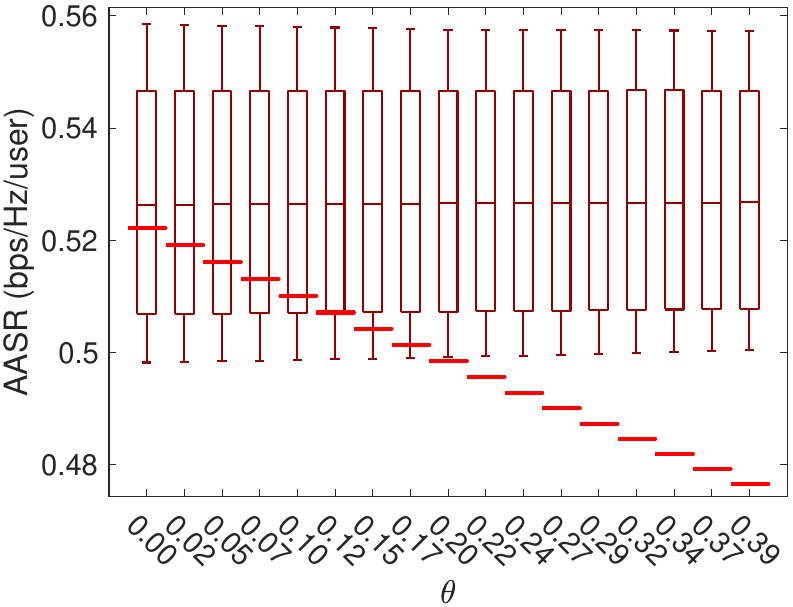}
	}

    \subfigure[Nominal ($\epsilon = 0.5$)]{
	 	\includegraphics[height=2.95cm]{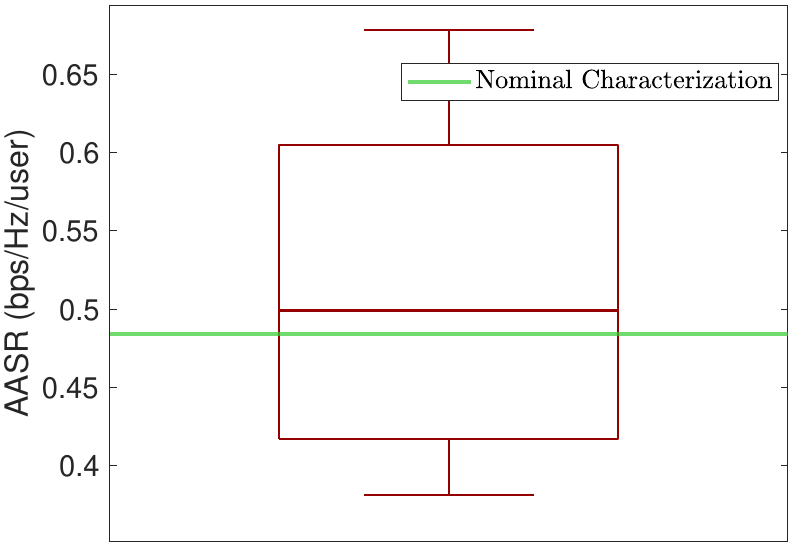}
	}
	\subfigure[Conservative ($\epsilon = 0.5$)]{
	 	\includegraphics[height=2.95cm]{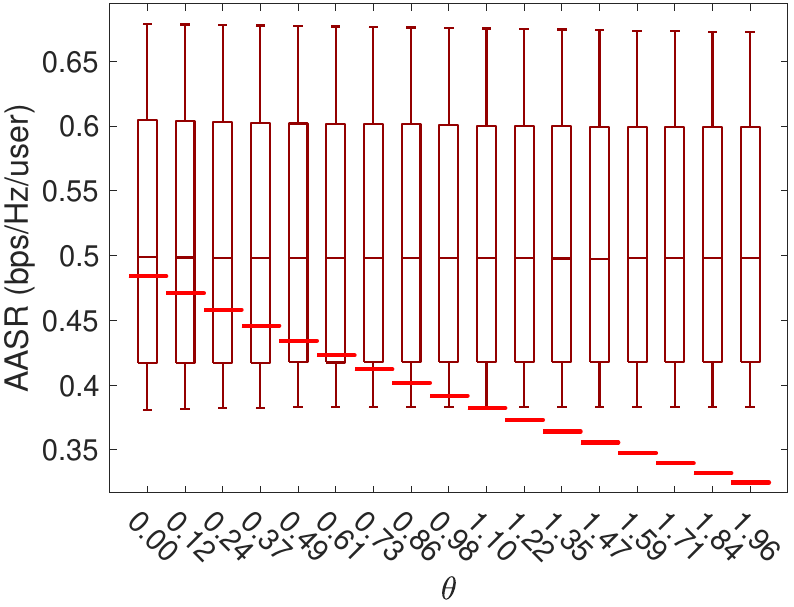}
	}

    \caption{Nominal and conservative characterizations against different values of $\epsilon$ under the sensing-centric design; cf. Figs. \ref{fig:nominal-issues-sensingcentric} and \ref{fig:sensing-centric}. To obtain tight performance lower bounds, when $\epsilon = 0.001$, $\theta = 0.12 \times 10^{-2}$; when $\epsilon = 0.01$, $\theta = 0.12 \times 10^{-1}$; when $\epsilon = 0.1$, $\theta = 0.12$; when $\epsilon = 0.5$, $\theta = 0.60$.}
    \label{fig:different-epsilon}
\end{figure}

\section{Conclusions}\label{sec:conclusion}
This paper investigates uncertainty-aware performance characterization and robust waveform design for ISAC. With uncertainties in nominal communication channels, the phenomenon of unreliable Pareto frontier is noticed and the concept of robust (i.e., conservative) performance characterization is proposed; see Fig. \ref{fig:true-frontiers}. We show that min-max robust waveform design formulations can obtain the conservative performance boundary; see Section \ref{sec:robust-formulation}. To solve these min-max formulations, an approximate solution framework is presented; see Section \ref{sec:solve-robust-problem}. The experiments validate that the proposed methods are effective in robust performance characterization for ISAC and are also computationally efficient; see Section \ref{sec:experiment}.

However, in real-world operations, the true radius $\theta$ of the uncertainty set is unknown, which can therefore be left as an empirically tunable parameter for investigated real-world ISAC systems unless it can be elegantly specified in the channel estimation stage; cf. Figs. \ref{fig:sensing-centric}, \ref{fig:tradeoff-design-TPC}, and \ref{fig:tradeoff-design-PAPC}. 

\bibliographystyle{IEEEtran}
\bibliography{References}

\begin{thebibliography}{10}
\providecommand{\url}[1]{#1}
\csname url@samestyle\endcsname
\providecommand{\newblock}{\relax}
\providecommand{\bibinfo}[2]{#2}
\providecommand{\BIBentrySTDinterwordspacing}{\spaceskip=0pt\relax}
\providecommand{\BIBentryALTinterwordstretchfactor}{4}
\providecommand{\BIBentryALTinterwordspacing}{\spaceskip=\fontdimen2\font plus
\BIBentryALTinterwordstretchfactor\fontdimen3\font minus
  \fontdimen4\font\relax}
\providecommand{\BIBforeignlanguage}[2]{{%
\expandafter\ifx\csname l@#1\endcsname\relax
\typeout{** WARNING: IEEEtran.bst: No hyphenation pattern has been}%
\typeout{** loaded for the language `#1'. Using the pattern for}%
\typeout{** the default language instead.}%
\else
\language=\csname l@#1\endcsname
\fi
#2}}
\providecommand{\BIBdecl}{\relax}
\BIBdecl

\bibitem{sturm2011waveform}
C.~Sturm and W.~Wiesbeck, ``Waveform design and signal processing aspects for
  fusion of wireless communications and radar sensing,'' \emph{Proceedings of
  the IEEE}, vol.~99, no.~7, pp. 1236--1259, 2011.

\bibitem{liu2020joint}
F.~Liu, C.~Masouros, A.~P. Petropulu, H.~Griffiths, and L.~Hanzo, ``Joint radar
  and communication design: Applications, state-of-the-art, and the road
  ahead,'' \emph{IEEE Trans. on Commun.}, vol.~68, no.~6, pp. 3834--3862, 2020.

\bibitem{zhang2021overview}
J.~A. Zhang, F.~Liu, C.~Masouros, R.~W. Heath, Z.~Feng, L.~Zheng, and
  A.~Petropulu, ``An overview of signal processing techniques for joint
  communication and radar sensing,'' \emph{IEEE Journal of Selected Topics in
  Signal Processing}, vol.~15, no.~6, pp. 1295--1315, 2021.

\bibitem{liu2022survey}
A.~Liu, Z.~Huang, M.~Li, Y.~Wan, W.~Li, T.~X. Han, C.~Liu, R.~Du, D.~K.~P. Tan,
  J.~Lu \emph{et~al.}, ``A survey on fundamental limits of integrated sensing
  and communication,'' \emph{IEEE Commun. Surveys \& Tutorials}, vol.~24,
  no.~2, pp. 994--1034, 2022.

\bibitem{zhou2022integrated}
W.~Zhou, R.~Zhang, G.~Chen, and W.~Wu, ``Integrated sensing and communication
  waveform design: A survey,'' \emph{IEEE Open Journal of the Communications
  Society}, vol.~3, pp. 1930--1949, 2022.

\bibitem{ferguson1996course}
T.~S. Ferguson, \emph{A Course in Large Sample Theory}.\hskip 1em plus 0.5em
  minus 0.4em\relax Chapman \& Hall, 1996.

\bibitem{ouyang2022performance}
C.~Ouyang, Y.~Liu, and H.~Yang, ``Performance of downlink and uplink integrated
  sensing and communications ({ISAC}) systems,'' \emph{IEEE Wireless
  Communications Letters}, vol.~11, no.~9, pp. 1850--1854, 2022.

\bibitem{ouyang2023integrated}
C.~Ouyang, Y.~Liu, H.~Yang, and N.~Al-Dhahir, ``Integrated sensing and
  communications: A mutual information-based framework,'' \emph{IEEE
  Communications Magazine}, vol.~61, no.~5, pp. 26--32, 2023.

\bibitem{xie2024sensing}
L.~Xie, F.~Liu, J.~Luo, and S.~Song, ``Sensing mutual information with random
  signals in gaussian channels: Bridging sensing and communication metrics,''
  \emph{arXiv preprint arXiv:2402.03919}, 2024.

\bibitem{kumari2019adaptive}
P.~Kumari, S.~A. Vorobyov, and R.~W. Heath, ``Adaptive virtual waveform design
  for millimeter-wave joint communication--radar,'' \emph{IEEE Trans. on Signal
  Processing}, vol.~68, pp. 715--730, 2019.

\bibitem{cover2006elements}
T.~M. Cover and J.~A. Thomas, \emph{Elements of Information Theory},
  2nd~ed.\hskip 1em plus 0.5em minus 0.4em\relax John Wiley \& Sons, 2006.

\bibitem{xiong2023fundamental}
Y.~Xiong, F.~Liu, Y.~Cui, W.~Yuan, T.~X. Han, and G.~Caire, ``On the
  fundamental tradeoff of integrated sensing and communications under gaussian
  channels,'' \emph{IEEE Trans. on Inform. Theory}, 2023.

\bibitem{mishra2019toward}
K.~V. Mishra, M.~B. Shankar, V.~Koivunen, B.~Ottersten, and S.~A. Vorobyov,
  ``Toward millimeter-wave joint radar communications: A signal processing
  perspective,'' \emph{IEEE Signal Processing Magazine}, vol.~36, no.~5, pp.
  100--114, 2019.

\bibitem{aditya2022sensing}
S.~Aditya, O.~Dizdar, B.~Clerckx, and X.~Li, ``Sensing using coded
  communications signals,'' \emph{IEEE Open Journal of the Communications
  Society}, vol.~4, pp. 134--152, 2022.

\bibitem{wei2023waveform}
Z.~Wei, J.~Piao, X.~Yuan, H.~Wu, J.~A. Zhang, Z.~Feng, L.~Wang, and P.~Zhang,
  ``Waveform design for {MIMO-OFDM} integrated sensing and communication
  system: An information theoretical approach,'' \emph{IEEE Trans. on Commun.},
  2023.

\bibitem{gaudio2020effectiveness}
L.~Gaudio, M.~Kobayashi, G.~Caire, and G.~Colavolpe, ``On the effectiveness of
  {OTFS} for joint radar parameter estimation and communication,'' \emph{IEEE
  Trans. on Wireless Commun.}, vol.~19, no.~9, pp. 5951--5965, 2020.

\bibitem{ouyang2016orthogonal}
X.~Ouyang and J.~Zhao, ``Orthogonal chirp division multiplexing,'' \emph{IEEE
  Trans. on Commun.}, vol.~64, no.~9, pp. 3946--3957, 2016.

\bibitem{huang2020majorcom}
T.~Huang, N.~Shlezinger, X.~Xu, Y.~Liu, and Y.~C. Eldar, ``Majorcom: A
  dual-function radar communication system using index modulation,'' \emph{IEEE
  Trans. on Signal Processing}, vol.~68, pp. 3423--3438, 2020.

\bibitem{ahmed2023sensing}
A.~Ahmed, E.~Aboutanios, and Y.~D. Zhang, ``{Sensing-Centric {ISAC} Signal
  Processing},'' in \emph{Integrated Sensing and Communications}.\hskip 1em
  plus 0.5em minus 0.4em\relax Springer, 2023, pp. 179--209.

\bibitem{liu2018toward}
F.~Liu, L.~Zhou, C.~Masouros, A.~Li, W.~Luo, and A.~Petropulu, ``Toward
  dual-functional radar-communication systems: Optimal waveform design,''
  \emph{IEEE Trans. on Signal Processing}, vol.~66, no.~16, pp. 4264--4279,
  2018.

\bibitem{guo2023bistatic}
B.~Guo, J.~Liang, G.~Wang, B.~Tang, and H.~So, ``Bistatic {MIMO} {DFRC} system
  waveform design via fractional programming,'' \emph{IEEE Trans. on Signal
  Processing}, 2023.

\bibitem{liao2023robust}
B.~Liao, X.~Xiong, and Z.~Quan, ``Robust beamforming design for dual-function
  radar-communication system,'' \emph{IEEE Transactions on Vehicular
  Technology}, 2023.

\bibitem{li2023optimal}
X.~Li, V.~C. Andrei, U.~J. M{\"o}nich, and H.~Boche, ``Optimal and robust
  waveform design for {MIMO-OFDM} channel sensing: A cramer-rao bound
  perspective,'' in \emph{2023 IEEE International Conference on
  Communications}, 2023.

\bibitem{liu2017robust}
F.~Liu, C.~Masouros, A.~Li, and T.~Ratnarajah, ``Robust {MIMO} beamforming for
  cellular and radar coexistence,'' \emph{IEEE Wireless Commun. Letters},
  vol.~6, no.~3, pp. 374--377, 2017.

\bibitem{zhao2022joint}
N.~Zhao, Y.~Wang, Z.~Zhang, Q.~Chang, and Y.~Shen, ``Joint transmit and receive
  beamforming design for integrated sensing and communication,'' \emph{IEEE
  Commun. Letters}, vol.~26, no.~3, pp. 662--666, 2022.

\bibitem{ren2023robust}
Z.~Ren, L.~Qiu, J.~Xu, and D.~W.~K. Ng, ``Robust transmit beamforming for
  secure integrated sensing and communication,'' \emph{IEEE Trans. on Commun.},
  2023.

\bibitem{luan2023robust}
M.~Luan, B.~Wang, Z.~Chang, T.~H{\"a}m{\"a}l{\"a}inen, and F.~Hu, ``Robust
  beamforming design for ris-aided integrated sensing and communication
  system,'' \emph{IEEE Trans. on Intell. Transport. Syst.}, 2023.

\bibitem{bazzi2023integrated}
A.~Bazzi and M.~Chafii, ``On integrated sensing and communication waveforms
  with tunable {PAPR},'' \emph{IEEE Transactions on Wireless Communications},
  2023.

\bibitem{fuhrmann2008transmit}
D.~R. Fuhrmann and G.~San~Antonio, ``Transmit beamforming for {MIMO} radar
  systems using signal cross-correlation,'' \emph{IEEE Trans. on Aerosp.
  Electron. Syst.}, vol.~44, no.~1, pp. 171--186, 2008.

\bibitem{stoica2008waveform}
P.~Stoica, J.~Li, and X.~Zhu, ``Waveform synthesis for diversity-based transmit
  beampattern design,'' \emph{IEEE Trans. on Signal Processing}, vol.~56,
  no.~6, pp. 2593--2598, 2008.

\bibitem{ahmed2012relaxation}
S.~Ahmed and I.~M. Jaimoukha, ``A relaxation-based approach for the orthogonal
  procrustes problem with data uncertainties,'' in \emph{Proceedings of 2012
  UKACC International Conference on Control}.\hskip 1em plus 0.5em minus
  0.4em\relax IEEE, 2012, pp. 906--911.

\bibitem{sergeyev2013introduction}
Y.~D. Sergeyev, R.~G. Strongin, and D.~Lera, \emph{Introduction to global
  optimization exploiting space-filling curves}.\hskip 1em plus 0.5em minus
  0.4em\relax Springer Science \& Business Media, 2013.

\bibitem{desale2015heuristic}
S.~Desale, A.~Rasool, S.~Andhale, and P.~Rane, ``Heuristic and meta-heuristic
  algorithms and their relevance to the real world: a survey,''
  \emph{International Journal of Computer Engineering in Research Trends}, vol.
  351, no.~5, pp. 2349--7084, 2015.

\bibitem{mart2018handbook}
R.~Mart, P.~M. Pardalos, and M.~G. Resende, \emph{Handbook of
  Heuristics}.\hskip 1em plus 0.5em minus 0.4em\relax Springer Publishing,
  2018.

\bibitem{kaveh2019metaheuristics}
A.~Kaveh and T.~Bakhshpoori, \emph{Metaheuristics: Outlines, {MATLAB} Codes,
  and Examples}.\hskip 1em plus 0.5em minus 0.4em\relax Springer, 2019.

\bibitem{taillard2023design}
E.~D. Taillard, \emph{Design of Heuristic Algorithms for Hard
  Optimization}.\hskip 1em plus 0.5em minus 0.4em\relax Springer Publishing,
  2023.

\bibitem{vazquez2010convergence}
E.~Vazquez and J.~Bect, ``Convergence properties of the expected improvement
  algorithm with fixed mean and covariance functions,'' \emph{Journal of
  Statistical Planning and Inference}, vol. 140, no.~11, pp. 3088--3095, 2010.

\bibitem{bull2011convergence}
A.~D. Bull, ``Convergence rates of efficient global optimization algorithms.''
  \emph{Journal of Machine Learning Research}, vol.~12, no.~10, 2011.

\bibitem{shahriari2015taking}
B.~Shahriari, K.~Swersky, Z.~Wang, R.~P. Adams, and N.~De~Freitas, ``Taking the
  human out of the loop: A review of {Bayesian} optimization,''
  \emph{Proceedings of the IEEE}, vol. 104, no.~1, pp. 148--175, 2015.

\bibitem{jones1998efficient}
D.~R. Jones, M.~Schonlau, and W.~J. Welch, ``Efficient global optimization of
  expensive black-box functions,'' \emph{Journal of Global optimization},
  vol.~13, pp. 455--492, 1998.

\bibitem{bjornson2019massive}
E.~Bj{\"o}rnson, L.~Sanguinetti, H.~Wymeersch, J.~Hoydis, and T.~L. Marzetta,
  ``Massive {MIMO} is a reality—what is next?: Five promising research
  directions for antenna arrays,'' \emph{Digital Signal Processing}, vol.~94,
  pp. 3--20, 2019.

\bibitem{eriksson2019scalable}
D.~Eriksson, M.~Pearce, J.~Gardner, R.~D. Turner, and M.~Poloczek, ``Scalable
  global optimization via local bayesian optimization,'' \emph{Advances in
  Neural Information Processing Systems}, vol.~32, 2019.

\bibitem{malherbe2017global}
C.~Malherbe and N.~Vayatis, ``Global optimization of lipschitz functions,'' in
  \emph{International Conference on Machine Learning}.\hskip 1em plus 0.5em
  minus 0.4em\relax PMLR, 2017, pp. 2314--2323.

\bibitem{ben2022algorithm}
A.~Ben-Tal and E.~Roos, ``An algorithm for maximizing a convex function based
  on its minimum,'' \emph{INFORMS Journal on Computing}, vol.~34, no.~6, pp.
  3200--3214, 2022.

\bibitem{hansen1995lipschitz}
P.~Hansen and B.~Jaumard, \emph{Lipschitz Optimization}.\hskip 1em plus 0.5em
  minus 0.4em\relax Springer, 1995.

\bibitem{golub2013matrix}
G.~H. Golub and C.~F. Van~Loan, \emph{Matrix Computations}.\hskip 1em plus
  0.5em minus 0.4em\relax JHU press, 2013.

\bibitem{zhang2020measuring}
H.~Zhang, R.~He, B.~Ai, S.~Cui, and H.~Zhang, ``Measuring sparsity of wireless
  channels,'' \emph{IEEE Transactions on Cognitive Communications and
  Networking}, vol.~7, no.~1, pp. 133--144, 2020.

\bibitem{khawar2015target}
A.~Khawar, A.~Abdelhadi, and C.~Clancy, ``Target detection performance of
  spectrum sharing {MIMO} radars,'' \emph{IEEE Sensors Journal}, vol.~15,
  no.~9, pp. 4928--4940, 2015.

\end{thebibliography}

 \begin{IEEEbiography}[{\includegraphics[width=1in,height=1.2in,clip,keepaspectratio]{wsx}}]
{Shixiong Wang} (Member, IEEE) received the B.Eng. degree in detection, guidance, and control technology and the M.Eng. degree in systems and control engineering from the School of Electronics and Information, Northwestern Polytechnical University, China, in 2016 and 2018, respectively, and the Ph.D. degree from the Department of Industrial Systems Engineering and Management, National University of Singapore, Singapore, in 2022.
 He has been a Post-Doctoral Research Associate with the Intelligent Transmission and Processing Laboratory, Imperial College London, London, U.K., since May 2023. He was a Post-Doctoral Research Fellow with the Institute of Data Science, National University of Singapore, from March 2022 to March 2023. His research interests include statistics and optimization theories, with applications in signal processing (especially optimal estimation theory) and machine learning (especially generalization error theory).
 \end{IEEEbiography}

 \begin{IEEEbiography}
 [{\includegraphics[width=1in,height=1.2in,clip,keepaspectratio]{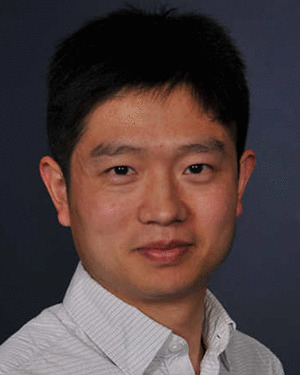}}] 
 {Wei Dai} (Member, IEEE) received the Ph.D. degree from the University of Colorado Boulder, Boulder, Colorado, in 2007. He is currently a Senior Lecturer (Associate Professor) in the Department of Electrical and Electronic Engineering, Imperial College London, London, UK. From 2007 to 2011, he was a Postdoctoral Research Associate with the University of Illinois Urbana-Champaign, Champaign, IL, USA. His research interests include electromagnetic sensing, biomedical imaging, wireless communications, and information theory.
 \end{IEEEbiography}

 \begin{IEEEbiography}
 [{\includegraphics[width=0.95in,height=1.1in,clip]{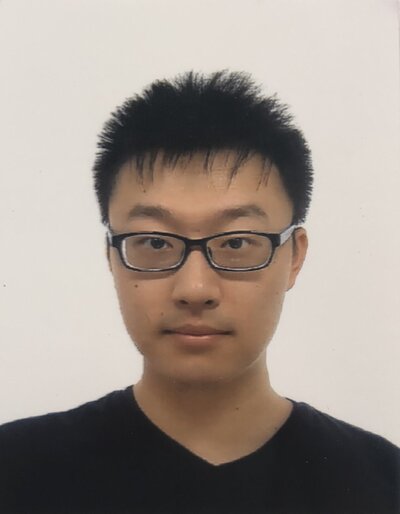}}] 
 {Haowei Wang} is currently a research scientist at Rice-Rick Digitalization PTE. Ltd. Before joining Rice-Rick, he received the B.Eng. degree in industrial engineering from Nanjing University, China, in 2016 and the Ph.D. degree in industrial and systems engineering from National University of Singapore, Singapore, in 2021. His research interest includes simulation optimization and Bayesian optimization under uncertainties.
 \end{IEEEbiography}

 \begin{IEEEbiography}
 [{\includegraphics[width=1in,height=1.2in,clip,keepaspectratio]{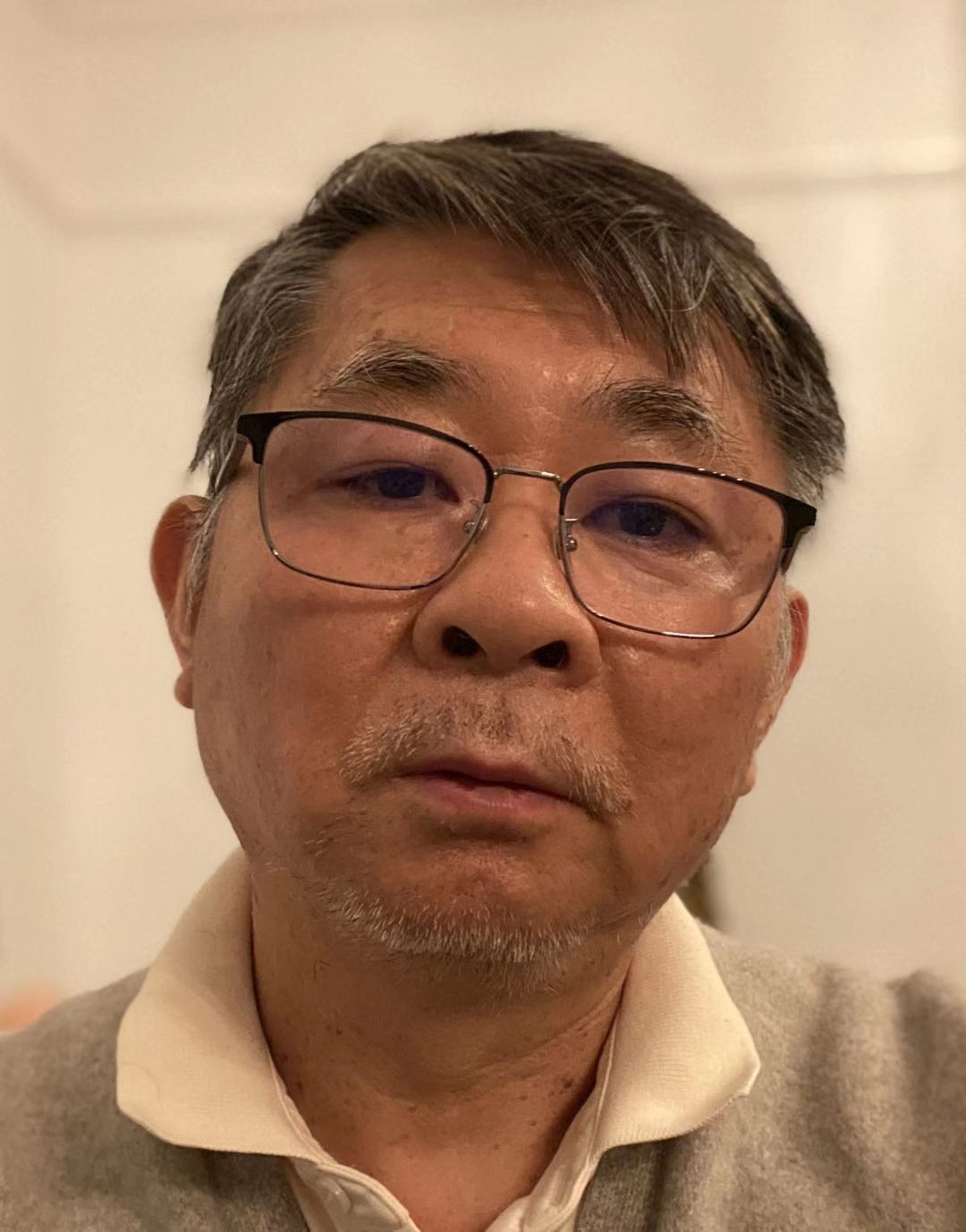}}] 
 {Geoffrey Ye Li} (Fellow, IEEE) is currently a Chair Professor at Imperial College London, UK. Before joining Imperial in 2020, he was a Professor at Georgia Institute of Technology, USA, for 20 years and a Principal Technical Staff Member with AT\&T Labs – Research (previous Bell Labs) in New Jersey, USA, for five years. He made fundamental contributions to orthogonal frequency division multiplexing (OFDM) for wireless communications, established a framework on resource cooperation in wireless networks, and introduced deep learning to communications. In these areas, he has published around 700 journal and conference papers in addition to over 40 granted patents. His publications have been cited over 69,000 times with an H-index of 119. He has been listed as a Highly Cited Researcher by Clarivate/Web of Science almost every year.

 Dr. Geoffrey Ye Li was elected to IEEE Fellow and IET Fellow for his contributions to signal processing for wireless communications. He won 2024 IEEE Eric E. Sumner Award, 2019 IEEE ComSoc Edwin Howard Armstrong Achievement Award, and several other awards from IEEE Signal Processing, Vehicular Technology, and Communications Societies.
 \end{IEEEbiography}


\newpage


\setcounter{page}{1}

\textbf{~~~~~~~~~~~~Online Supplementary Materials}\\

This document contains appendices to the paper, including:
\begin{itemize}
    \item Extensive discussions on Model \eqref{eq:perfect-communication}; see Appendix \ref{append:extended-model};
    \item Proofs of all lemmas, theorems, and propositions;
    \item An intuitive explanation of Method \ref{method:max-form-solution}; see Appendix \ref{subsubsec:motivation};
    \item The specifics of how the simulations in Section \ref{sec:experiment} are conducted; see Appendix \ref{append:simulation-engine}.
\end{itemize}

\appendices
\section{Extensions of Model \captext{\eqref{eq:perfect-communication}}}\label{append:extended-model}
For the signal model $\rmat Y = \mat H \mat X + \rmat W$ in \eqref{eq:communication-model}, it can be written as 
\[
\rmat Y = \mat S + (\mat H \mat X - \mat S) + \rmat W.
\]
To recover $\mat S$ from $\rmat Y$, we aim to minimize the interference signals $\mat H \mat X - \mat S$ over waveforms $\mat X$. This is another reason to minimize MUI energy $\|\mat H \mat X - \mat S\|^2_F$ as in \eqref{eq:perfect-communication}. Below we discuss the case of multi-antenna users and the case of multi-carrier.
\begin{itemize}
    \item \textbf{Multi-Antenna Case}. Suppose that we have $K$ downlink users and each user is equipment with $R$ receive antennas. Then, the channel matrix of each user is $\mat H_k \in \C^{R \times N}$ for every $k \in [K]$, where $N$ is the number of transmit antennas at the base station. As a result, for each user $k$, the base-band signal model is $\rmat Y_k = \mat H_k \mat X + \rmat W_k$, where $\mat X \in \C^{N \times L}$ is the transmitted waveform and $L$ is the frame length. By constructing $\mat Y$, $\mat H$, and $\mat W$ as
    \[
        \rmat Y \defeq 
        \left[
        \begin{array}{cccc}
           \rmat Y_1 \\
           \rmat Y_2 \\
           \vdots \\
           \rmat Y_K
        \end{array}
        \right]
    ,~~
        \mat H \defeq 
        \left[
        \begin{array}{cccc}
           \mat H_1 \\
           \mat H_2 \\
           \vdots \\
           \mat H_K
        \end{array}
        \right]
    ,~~
        \rmat W \defeq 
        \left[
        \begin{array}{cccc}
           \rmat W_1 \\
           \rmat W_2 \\
           \vdots \\
           \rmat W_K
        \end{array}
        \right],
    \]
    the integrated base-band signal model $\rmat Y = \mat H \mat X + \rmat W$ can be obtained. Note that in this case, $\|\mat H \mat X - \mat S\|^2_F$ no longer physically means the multi-user interference (MUI) energy. However, as in \eqref{eq:perfect-communication}, minimizing $\|\mat H \mat X - \mat S\|^2_F$ with respect to $\mat X$ is still the technical focus to improve communication performance (i.e., to reduce the restoration error of the information matrix $\mat S$).

    \item \textbf{Multi-Carrier Case}. Suppose that we have $R$ sub-carriers and for each sub-carrier $r \in [R]$, the base-band signal model is $\rmat Y_r = \mat H_r \mat X_r + \rmat W_r$, where $\rmat Y_r \in \C^{K \times L}$, $\mat H_r \in \C^{K \times N}$, $\mat X_r \in \C^{N \times L}$, and $\rmat W_r \in \C^{K \times L}$; $K$ is the number of downlink single-antenna users, $L$ is the frame length, and $N$ is the number of transmit antennas at the base station. Since every sub-carrier has an independent base-band signal model, all signal processing operations can be applied separately for every base-band model indexed by $r \in [R]$.
\end{itemize}

\section{Proof of Lemma \ref{lem:max-form}}\label{append:proof-max-form-lemma}
\begin{proof}
The first inequality in the lemma is due to the weak min-max property (also known as the min-max inequality), which is unconditionally true for any $\phi$, $\cal H$, and $\cal X$. The second inequality is due to the feasibility of the solution $\mat X^*$ in $\cal X$. This completes the proof.
\stp
\end{proof}

\section{Proof of Lemma \ref{lem:gap}}\label{append:proof-gap}
\begin{proof}
We have
\[
\begin{array}{l}
    \displaystyle \max_{\mat H} \phi(\mat H, \matb X) - \max_{\mat H} \min_{\mat X}  \phi(\mat H, \mat X) \vspace{0.3em} \\
    \quad \quad = \displaystyle \max_{\mat H} \phi(\mat H, \matb X) - \phi(\matb H, \matb X) + \displaystyle \min_{\mat X} \phi(\matb H, \mat X) \\
    \quad \quad \quad \quad  - \displaystyle \max_{\mat H} \min_{\mat X}  \phi(\mat H, \mat X) \\
    \quad \quad \le \big|\displaystyle \max_{\mat H} \phi(\mat H, \matb X) - \phi(\matb H, \matb X)\big| + \\
    \quad \quad \quad \quad \big|\displaystyle \max_{\mat H} \min_{\mat X}  \phi(\mat H, \mat X) - \displaystyle \min_{\mat X} \phi(\matb H, \mat X)\big| \vspace{0.3em}\\
    \quad \quad \le \displaystyle \max_{\mat H} \big|\phi(\mat H, \matb X) - \phi(\matb H, \matb X)\big| + \\
    \quad \quad \quad \quad \displaystyle \max_{\mat H} \big| \min_{\mat X}  \phi(\mat H, \mat X) - \displaystyle \min_{\mat X} \phi(\matb H, \mat X)\big| \\
    \quad \quad \le L_2 \cdot \displaystyle \max_{\mat H} \| \mat H - \matb H \| + L_1 \cdot \displaystyle \max_{\mat H} \| \mat H - \matb H \| \\
    \quad \quad = (L_1 + L_2) \cdot \theta.
\end{array}
\]
This completes the proof.
\stp
\end{proof}

\section{Proof of Theorem \ref{thm:max-form}}\label{append:proof-max-form}
\begin{proof}
First, we consider the max-min counterpart of \eqref{eq:robust-waveform-design}:
\begin{equation}\label{eq:robust-waveform-design-dual}
    \begin{array}{cl}
       \displaystyle \max_{\mat H} \min_{\mat X}   &  \|\mat H \mat X - \mat S\|^2_F \\
       \st  & \displaystyle \frac{1}{L} \mat X \mat X^\H = \mat R,\\
       & \|\mat H - \bar{\mat H}\| \le \theta.
    \end{array}
\end{equation}
For every feasible $\mat H$, the inner sub-problem
$ \min_{\mat X \in \cal X}  \|\mat H \mat X - \mat S\|^2_F$
is solved by
\begin{equation*}\label{eq:sensing-centric-X_H}
    \mat X^*_{\mat H} = \sqrt{L} \cdot \mat F \cdot \mat U \mat I_{N \times L} \mat V^\H,
\end{equation*}
where $\mat U \mat \Sigma \mat V^\H \overset{\text{SVD}}{=} \mat F^\H \mat H^\H\mat S$ and $\mat I_{N \times L} \defeq [\mat I_N, \mat 0_{N \times (L - N)}]$; the $N \times (L - N)$ zero matrix is denoted by $\mat 0_{N \times (L - N)}$; see \cite[Eq. (15)]{liu2018toward}. Note that the optimal solution $\mat X^*_{\mat H}$ depends on $\mat H$, and $\mat X^*_{\mat H}$ may not be unique given $\mat H$. Plugging in $\mat X^*_{\mat H}$ back to \eqref{eq:robust-waveform-design-dual} yields \eqref{eq:max-form}.

Second, according to Lemma \ref{lem:max-form} and Condition \eqref{eq:sensing-centric-thm-condition}, the strong min-max property holds, that is,
\begin{equation}\label{eq:sensing-centric-game}
\begin{array}{l}
    \displaystyle \min_{\mat X \in \cal X}  \max_{\mat H  \in \cal H} \|\mat H \mat X - \mat S\|^2_F = \displaystyle \max_{\mat H  \in \cal H} \min_{\mat X \in \cal X}   \|\mat H \mat X - \mat S\|^2_F \\
    \quad\quad\quad\quad = \displaystyle \max_{\mat H  \in \cal H} \|\sqrt{L} \cdot \mat H \cdot \mat F \cdot \mat U \mat I_{N \times L} \mat V^\H - \mat S\|^2_F.
\end{array}
\end{equation}
This completes the proof. \stp
\end{proof}

\section{Proof of Proposition \captext{\ref{prop:f-continuous}}}\label{append:f-continuous}
One may verify that it is difficult to prove the continuity of $f$ directly using the definition in \eqref{eq:obj-f} because the SVD of a matrix might not be unique; specifically, given $\mat H$, there may exist multiple $\mat U$s and $\mat V$s such that $\mat U \mat \Sigma \mat V^\H = \mat F^\H \mat H^\H\mat S$. The complication arises when $\mat U$ and $\mat V$ correspond to zero singular value(s) in $\mat \Sigma$. We, therefore, investigate the continuity of $f$ from its original definition.
\begin{proof}
    Recall that 
    $
    f(\mat H) \defeq \min_{\mat X \in \cal X} \|\mat H \mat X - \mat S\|^2_F
    $
    where $\cal X \defeq \{\mat X:~\mat X \mat X^\H = L\mat R\}$. 
    First, note that for every $\mat X \in \cal X$ and every $\mat H_1, \mat H_2 \in \cal H$, there exists an upper bound $0 < B_1 < \infty$ such that
    $
        \|\mat H_1 \mat X - \mat S\|_F +\|\mat H_2 \mat X - \mat S\|_F \le B_1
    $. 
    A loose choice can be $B_1 \defeq 2\sqrt{LP_\T} \cdot (\|\bar{\mat H}\|_F + B\theta) + 2\|\mat S\|_F$ where $0 < B < \infty$ is a real-valued constant such that $\|\mat H_1 - \bar{\mat H}\|_F \le B\|\mat H_1 - \bar{\mat H}\|$; the existence of $B$ is guaranteed due to the equivalence of norms on a finite-dimensional linear space. Just note that $|\|\mat H_1\|_F - \|\bar{\mat H}\|_F| \le \|\mat H_1 - \bar{\mat H}\|_F \le B\|\mat H_1 - \bar{\mat H}\| \le B\theta$. The same argument also holds for $\mat H_2$. Hence, for every $\mat H_1, \mat H_2 \in \cal H$, we have
    \[
        \begin{array}{l}
            |f(\mat H_1) - f(\mat H_2)| \\
            \quad\quad = \Big|\displaystyle \min_{\mat X \in \cal X} \|\mat H_1 \mat X - \mat S\|^2_F - \displaystyle \min_{\mat X \in \cal X} \|\mat H_2 \mat X - \mat S\|^2_F\Big| \\
            \quad\quad \le \displaystyle \max_{\mat X \in \cal X} \Big| \|\mat H_1 \mat X - \mat S\|^2_F - \|\mat H_2 \mat X - \mat S\|^2_F\Big| \\
            \quad\quad = \displaystyle \max_{\mat X \in \cal X} \Big| \Big(\|\mat H_1 \mat X - \mat S\|_F +\|\mat H_2 \mat X - \mat S\|_F\Big) \cdot \\
            \quad \quad \quad \quad \quad \quad \Big(\|\mat H_1 \mat X - \mat S\|_F - \|\mat H_2 \mat X - \mat S\|_F\Big) \Big| \\
            \quad\quad \le \displaystyle \max_{\mat X \in \cal X} \Big| \|\mat H_1 \mat X - \mat S\|_F +\|\mat H_2 \mat X - \mat S\|_F \Big| \cdot \\
            \quad \quad \quad \quad \quad \quad \Big|\|\mat H_1 \mat X - \mat S\|_F - \|\mat H_2 \mat X - \mat S\|_F \Big| \\
            \quad\quad \le B_1 \cdot \displaystyle \max_{\mat X \in \cal X} \Big\| \mat H_1 \mat X - \mat H_2 \mat X\Big\|_F \\
            \quad\quad \le B_1 \cdot \Big\| \mat H_1 - \mat H_2\Big\|_F \cdot \displaystyle \max_{\mat X \in \cal X} \Big\|\mat X\Big\|_F \\
           \quad\quad  = B_1 \sqrt{LP_\T} \cdot \| \mat H_1 - \mat H_2 \|_F \\
           \quad\quad \le B_1 \sqrt{LP_\T} \cdot B \| \mat H_1 - \mat H_2 \|.
        \end{array}
    \]
    
    Note that $\|\mat X\|_F = \sqrt{\Tr \mat X^\H \mat X} = \sqrt{ L \cdot \Tr \mat R} = \sqrt{L P_\T}$.
    \stp
\end{proof}

\section{Proof of Proposition \captext{\ref{prop:f-other-properties}}}\label{append:f-other-properties}
\begin{proof}
The non-convexity and non-concavity of the objective function $f(\mat H) = \min_{\mat X:~\mat X \mat X^\H = L \mat R} \|\mat H \mat X - \mat S\|^2_F$ can be verified through the definitions of convexity and concavity by constructing counterexamples, completing the proof. \stp
\end{proof}

The example below specifically justifies the claim above.
\begin{example}
Let the nominal channel be $\bar H \defeq 0.9 + 0.5j$, $\epsilon \defeq 0.05$, $R \defeq 1$, and $S \defeq 1 + 1j$. Suppose that $H_1$ and $H_2$ are generated according to the following formulas:
\[
\begin{array}{cc}
    H_1 &= \bar H + \epsilon \cdot \Delta_1, \\
    H_2 &= \bar H + \epsilon \cdot \Delta_2,
\end{array}
\]
where $\Delta_1$ and $\Delta_2$ are mutually-independent standard complex Gaussian variables. One may verify that $f(0.5 H_1 + 0.5 H_2) \le 0.5 f(H_1) + 0.5 f(H_2)$ for some realizations of $H_1$ and $H_2$, while $f(0.5 H_1 + 0.5 H_2) \ge 0.5 f(H_1) + 0.5 f(H_2)$ for other realizations. Therefore, $f$ is neither concave nor convex.
\stp
\end{example}

\section{Proof of Proposition \captext{\ref{prop:f-upper-bound}}}\label{append:f-upper-bound}
\begin{proof}
    The upper bound is obtained by plugging $\bar{\mat X}$ into the optimization problem $f(\mat H) \defeq \min_{\mat X:~\mat X \mat X^\H = L \mat R} \|\mat H \mat X - \mat S\|^2_F$ because $\bar{\mat X}$ is a feasible solution. The lower bound is obtained by the reverse triangle inequality, i.e., $\|\mat H \mat X - \mat S\|_F \ge |\|\mat H \mat X\|_F - \|\mat S\|_F|$; note that $\|\mat H \mat X\|_F = \sqrt{L\Tr[\mat H^\H \mat H \mat R]}$ because $\mat X \mat X^\H = L \mat R$. The upper bound is positive-definite quadratic (thus convex) in $\opvec(\mat H)$ because $\bar{\mat X} \bar{\mat X}^\H = L \mat R$ and $\mat R$ is positive definite. Since at the center $\bar{\mat H}$ of $\cal H$ we have $f(\bar{\mat H}) = \overline{f}(\bar{\mat H})$, the upper bound is tight. Also, the reverse triangle inequality is tight, and therefore, the lower bound is tight. The third claim is obvious. We show the uniform bound in the fourth claim below.
    We have 
    \[
        \begin{array}{cl}
            \bar f(\mat H) - f(\mat H) &= \bar f(\mat H) - \bar{f}(\bar{\mat H}) + f(\bar{\mat H}) - f(\mat H) \\
            & \le |\bar f(\mat H) - \bar{f}(\bar{\mat H})| + |f(\bar{\mat H}) - f(\mat H)| \\
            & \le L_f \|\mat H - \bar{\mat H}\| + L_f \|\mat H - \bar{\mat H}\| \\
            & \le 2L_f \cdot \theta.
        \end{array}
    \]
    This completes the proof.
    \stp
\end{proof}

\section{Proof of Proposition \captext{\ref{prop:max-form-solution}}}\label{append:max-form-solution}
\begin{proof}
    The Hessian of the objective function in terms of $\opvec(\mat H)$ is
    \[
    \begin{array}{l}
    L((\mat F \mat A)^\T \otimes \mat I_K)^\H ( (\mat F \mat A)^\T \otimes \mat I_K) \\
    \quad \quad \quad = L\{[(\mat F \mat A)^\T]^\H \otimes \mat I^\H_K\} \{ (\mat F \mat A)^\T \otimes \mat I_K\} \\
    \quad \quad \quad = L[(\mat F \mat A)^\T]^\H (\mat F \mat A)^\T \otimes \mat I^\H_K \mat I_K \\
    \quad \quad \quad = L[(\mat F \mat A)(\mat F \mat A)^\H]^\T \otimes \mat I^\H_K \mat I_K \\
    \quad \quad \quad = L \mat R^\T \otimes \mat I_K \succ \mat 0,
    \end{array}
    \]
    because the beampattern-inducing matrix $\mat R$ is positive definite.
    Hence, the objective function of \eqref{eq:max-form-trans-vec} is positive definite in $\opvec(\mat H)$, and therefore, convex in $\opvec(\mat H)$. The convexity of the objective function in terms of $\opvec(\mat A)$ can be shown in a similar way: just note that the objective function can only be shown to be positive semi-definite because the Hessian $L(\mat I_L \otimes \mat H \mat F)^\H(\mat I_L \otimes \mat H \mat F)$ is not necessarily positive definite. 
    
    The feasible region of $\mat H$ is convex because the norm constraint is convex and the equality constraint is linear. The feasible region of $\mat A$ is convex because it is built by a linear equality constraint. This completes the proof.
    \stp
\end{proof}

\section{Intuitive Understanding of Solution Method \captext{\ref{method:max-form-solution}}}\label{subsubsec:motivation}
Motivated by Lemma \ref{lem:max-form} and Condition \eqref{eq:sensing-centric-thm-condition}, we start with employing the upper bound $\overline{f}(\mat H)$ in Proposition \ref{prop:f-upper-bound} because $\overline{f}(\mat H)$ is close to $f(\mat H)$ if the radius $\theta$ of the uncertainty set $\cal H$ is small; recall from Subsection \ref{subsec:budget-uncertainty-set} that the radius of uncertainty set should be controlled to small. The other reason to employ $\bar{\mat X}$ and $\overline{f}(\mat H)$ is that $\bar{\mat X}$ exists for any specified uncertainty set $\cal H$ with any radius $\theta \ge 0$ because $\cal H$ is centered at $\bar{\mat H}$.

\textbf{Step 1}. Maximize the upper bound $\overline{f}(\mat H)$ to obtain $\mat H^*$:
    \[
        \mat H^* = \argmax_{\mat H \in \cal H} \|\mat H \bar{\mat X} - \mat S\|^2_F.
    \]
\textbf{Interpretation}. Step 1 gives a feasible approximation solution $(\bar{\mat X},\mat H^*)$ to Problem \eqref{eq:max-form}, and therefore Problem \eqref{eq:robust-waveform-design}, in the sense of Fact \ref{fact:upperbound}. To be specific, we have
    \begin{enumerate}[a)]
        \item Bound of the Truly Optimal Cost and the True Cost: 
        $$
        \min_{\mat X \in \cal X} \|\mat H_0 \mat X - \mat S\|^2_F \le \|\mat H_0 \bar{\mat X} - \mat S\|^2_F \le \|\mat H^* \bar{\mat X} - \mat S\|^2_F;
        $$
        \item Bound of the Nominally Optimal Cost: 
        $$
        \min_{\mat X \in \cal X} \|\bar{\mat H} \mat X - \mat S\|^2_F = \|\bar{\mat H} \bar{\mat X} - \mat S\|^2_F \le \|\mat H^* \bar{\mat X} - \mat S\|^2_F;
        $$
    \end{enumerate}
    where the nominally optimal solution $\bar{\mat X} \defeq  \sqrt{L} \mat F \bar{\mat U} \mat I_{N \times L} \bar{\mat V}^\H$ solves the nominal waveform design problem $\min_{\mat X \in \cal X} \|\bar{\mat H} \mat X - \mat S\|^2_F$ and $\cal X = \{\mat X:\mat X \mat X^\H = L \mat R\}$.
However, maximizing the upper-bound $\overline{f}(\mat H)$ gives an extremely conservative solution. Specifically, the robust cost $\|\mat H^* \bar{\mat X} - \mat S\|^2_F$ would be overly large than the true cost $\|\mat H_0 \bar{\mat X} - \mat S\|^2_F$ and the truly optimal cost $\min_{\mat X \in \cal X} \|\mat H_0 \mat X - \mat S\|^2_F$. Hence, a refinement is needed.

\textbf{Step 2}.
    Refine the robust cost $\|\mat H^* \bar{\mat X} - \mat S\|^2_F$: i.e.,
    \[
        \mat X^* = \argmin_{\mat X \in \cal X} \|\mat H^* \mat X - \mat S\|^2_F.
    \]
\textbf{Consequence}. However, the resulting cost $\|\mat H^* \mat X^* - \mat S\|^2_F$ cannot be guaranteed to upper bound the truly optimal cost, the true cost, and the nominally optimal cost. To be specific, it is unnecessary to have
    \begin{enumerate}[a)]
        \item Bound of the Truly Optimal Cost and the True Cost: 
        $$
        \min_{\mat X \in \cal X} \|\mat H_0 \mat X - \mat S\|^2_F \le \|\mat H_0 \mat X^* - \mat S\|^2_F \stackrel{\text{\scriptsize ?}}{\le} \|\mat H^* \mat X^* - \mat S\|^2_F.
        $$
        \item Bound of the Nominally Optimal Cost: 
        $$
        \min_{\mat X \in \cal X} \|\bar{\mat H} \mat X - \mat S\|^2_F = \|\bar{\mat H} \bar{\mat X} - \mat S\|^2_F \stackrel{\text{\scriptsize ?}}{\le} \|\mat H^* \mat X^* - \mat S\|^2_F.
        $$
    \end{enumerate}
(Note that $\mat H^*$ is obtained by the maximization at $\matb X$.) Thus, we propose a remedy strategy.

\textbf{Step 3}. Design a mechanism to let $\mat X^* = \bar{\mat X}$.\\
\textbf{Interpretation}. If it technically holds that $\mat X^* = \bar{\mat X}$ (or $\mat X^* \approx \bar{\mat X}$), then the conservativeness of the solution $(\mat H^*, \bar{\mat X})$ will be controlled, and equivalently, the feasibility of the solution $(\mat H^*, {\mat X}^*)$ in the sense of Fact \ref{fact:upperbound} will be guaranteed.

The three algorithmic steps above provide an intuitive understanding of Method \ref{method:max-form-solution}. 

\section{Proof of Lemma \captext{\ref{lem:stacking}}}\label{append:real-domain-representation}
\begin{proof}
We have
$
(\mat \Gamma + \mat \Theta j)(\vec a + \vec b j) = (\mat \Gamma \vec a - \mat \Theta \vec b) + (\mat \Theta \vec a + \mat \Gamma \vec b )j
$. 
Hence, the stacking scheme to construct real quantities from complex quantities is given in the statement of the lemma.
This completes the proof.
\stp
\end{proof}

\section{Proof of Lemma \captext{\ref{lem:positive-definite}}}\label{append:positive-definite}
\begin{proof}
According to Proposition \ref{prop:max-form-solution}, we immediately have $\mat C^\T \mat C \succ \mat 0$, and therefore, $p(\vec h)$ is positive definite in $\vec h$. In addition, we have
$
   p(\vec h) - \frac{\mu}{2} \vec h^\T \vec h = \vec h^\T (\mat C ^\T \mat C - \frac{\mu}{2} \mat I_{2KN}) \vec h - 2 \vec s^\T \mat C \vec h + \vec s^\T \vec s
$, 
where $\mu > 0$ is a positive number. Since $\mat C^\T \mat C \succ \mat 0$, there exist $\mu > 0$ such that 
$
\mat C ^\T \mat C - \frac{\mu}{2} \mat I_{2KN} \succeq \mat 0
$. 
As a result, the function $p(\vec h) - \frac{\mu}{2} \vec h^\T \vec h$ can be a convex function for some $\mu > 0$, which means that the objective function $p(\vec h)$ of \eqref{eq:H-compact} is strongly convex.
\stp
\end{proof}

\section{Proof of Proposition \captext{\ref{prop:convex-quadratic-max}}}\label{append:convex-quadratic-max}
\begin{proof}
    According to 
    [49, Thm.~1], a point $\vec y$ is a globally optimal solution to \eqref{eq:H-compact} if and only if
    $
        \ip{\nabla p(\vec y)}{~\vec h - \vec y} \le 0
    $, for every $\vec h$ such that $\|\vec h - \bar{\vec h}\| \le \theta$, 
    where $\nabla p(\vec y)$ denotes the gradient of $p$ evaluated at $\vec y$. Therefore, if we have
    $
        \max_{\vec y \in \cal Y} \max_{\vec h:\|\vec h - \bar{\vec h}\| \le \theta} \ip{\nabla p(\vec y)}{~\vec h - \vec y}  \le 0
    $, 
    for some dedicated $\cal Y$, then every $\vec y$ that solves the above optimization is a global maximum of \eqref{eq:H-compact}. This proposition, which is adapted from [49, Algo.~1] for Problem \eqref{eq:H-compact}, formalizes the above intuition. The global optimality and convergence are therefore guaranteed by [49, Thm.~4]. For rigorous and complete technical proof, see [49]; just note that in \eqref{eq:sub-prob-1}, the equality constraint can be changed to its convex inequality counterpart because the optima of linear objective functions lie on the boundary of feasible regions.  \stp
\end{proof}

\section{Proof of Proposition \captext{\ref{prop:sub-prob-1}}}\label{append:sub-prob-1}
\begin{proof}
Since $\mat C^\T \mat C$ is positive definite, the invertible $\mat M$ exists; see Appendix \ref{append:positive-definite}. Problem \eqref{eq:sub-prob-1} is equivalent to
    \[
        \begin{array}{ccl}
            \vec y_k =  & \displaystyle \argmax_{\vec y \in \R^{2KN}} &(\vec h^\T_{k-1} \mat C^\T \mat C - \vec s^\T \mat C) \cdot \vec y \\
            & \st &\|\mat M \vec y - \mat M^{-\T} \mat C^\T \vec s\|_2 = \gamma.
        \end{array}
    \]
The above display is further equivalent to
    \[
        \begin{array}{cl}
            \displaystyle \max_{\vec z \in \R^{2KN}} & (\vec h^\T_{k-1} \mat C^\T \mat C - \vec s^\T \mat C) \mat M^{-1} \cdot (\gamma \vec z + \mat M^{-\T} \mat C^\T \vec s), \\
            
            \st & \|\vec z\|_2 = 1.
        \end{array}
    \]
    Due to Cauchy–Schwarz inequality, the maximum is
    \[
        \vec z^* = \frac{\mat M^{-\T} (\mat C^\T \mat C \vec h_{k-1} -  \mat C^\T \vec s)}{\| \mat M^{-\T} (\mat C^\T \mat C \vec h_{k-1} -  \mat C^\T \vec s) \|_2}.
    \]
    Then, a maximum of \eqref{eq:sub-prob-1} is 
    $
        \vec y^* =  \mat M^{-1} \cdot (\gamma \vec z^* + \mat M^{-\T} \mat C^\T \vec s)
    $. 
    This completes the proof.
    \stp
\end{proof}

\section{Proof of Proposition \captext{\ref{prop:sub-prob-2}}}\label{append:sub-prob-2}
\begin{proof}
    Problem \eqref{eq:sub-prob-2} is equivalent to
    \[
        \begin{array}{cl}
            \displaystyle \max_{\vec z \in \R^{2KN}} & (\vec y^\T_k \mat C^\T \mat C - \vec s^\T \mat C) \cdot (\theta \vec z + \bar{\vec h}),~\st~\|\vec z\|_2 = 1.
        \end{array}
    \]
    Due to Cauchy–Schwarz inequality, the maximum is
    \[
        \vec z^* = \frac{\mat C^\T \mat C \vec y_k - \mat C^\T \vec s}{\|\mat C^\T \mat C \vec y_k - \mat C^\T \vec s\|_2}.
    \]
    As a result, a maximum of \eqref{eq:sub-prob-2} is 
    $
        \vec h^* =  \theta \vec z^* + \bar{\vec h}
    $. 
    \stp
\end{proof}

\section{Proof of Proposition \captext{\ref{thm:solution-remedy-problem}}}\label{append:UV-solution-remedy}
\begin{proof}
In terms of $\mat U$, we can rewrite \eqref{eq:max-form-trans-vec-sol-UV} as
\begin{equation}\label{eq:max-form-trans-vec-sol-UV-compact-U}
    \begin{array}{cl}
        \displaystyle \min_{\mat U}  \|\mat U \mat A_1 - \mat B_1\|^2_F,~~\st~~\mat U \mat U^\H = \mat I_N,
    \end{array}
\end{equation}
where $\mat A_1$ and $\mat B_1$ are defined in Proposition \ref{thm:solution-remedy-problem}. Problem \eqref{eq:max-form-trans-vec-sol-UV-compact-U} is a standard orthogonal Procrustes problem whose closed-form solution is given in Proposition \ref{thm:solution-remedy-problem}; see technical details in [50].

In terms of $\mat V$, we can rewrite \eqref{eq:max-form-trans-vec-sol-UV} as
\begin{equation}\label{eq:max-form-trans-vec-sol-UV-compact-V}
    \begin{array}{cl}
        \displaystyle \min_{\mat V}  \|\mat V \mat A_2 - \mat B_2\|^2_F,~~\st~~\mat V \mat V^\H = \mat I_L,
    \end{array}
\end{equation}
where $\mat A_2$ and $\mat B_2$ are defined in Proposition \ref{thm:solution-remedy-problem}. Problem \eqref{eq:max-form-trans-vec-sol-UV-compact-V} is a standard orthogonal Procrustes problem whose closed-form solution is given in Proposition \ref{thm:solution-remedy-problem}; see technical details in [50].

In terms of $\mat \Sigma$, Problem \eqref{eq:max-form-trans-vec-sol-UV} is a positive-definite-quadratic convex program. Note that the space $\Omega_{N \times L}$ of the matrices with diagonal, non-negative, and real entries is convex. The gradient of the objective function of \eqref{eq:max-form-trans-vec-sol-UV} with respect to $\mat \Sigma$ is 
\[
2\mat U^\H(\mat U \mat \Sigma \mat V^\H - \mat F^\H \mat H^{*\H}\mat S) \mat V = 2 \mat \Sigma - 2 \mat U^\H  \mat F^\H \mat H^{*\H} \mat S\mat V.
\] 
Therefore, the optimal solution is given by the projected point of $\mat U^\H \mat F^\H \mat H^{*\H} \mat S \mat V$ onto $\Omega_{N \times L}$. 

The convergence proof of the iteration process is straightforward. We rewrite \eqref{eq:max-form-trans-vec-sol-UV} in shorthand as
\[
    \displaystyle \min_{\mat U, \mat \Sigma, \mat V} \psi(\mat U, \mat \Sigma, \mat V),
\]
where the constraints of the variables $(\mat U, \mat \Sigma, \mat V)$ are implicitly defined by \eqref{eq:max-form-trans-vec-sol-UV}. Let $(\mat U^0, \mat \Sigma^0, \mat V^0)$ denote the initial values of the variables and $(\mat U^r, \mat \Sigma^r, \mat V^r)$ the values at the $r^\th$ iteration. Following the defined iteration process in Proposition \ref{thm:solution-remedy-problem}, we have
\[
\begin{array}{cl}
    \psi(\mat U^0, \mat \Sigma^0, \mat V^0) & \ge \psi(\mat U^1, \mat \Sigma^0, \mat V^0) \\
    &\ge \psi(\mat U^1, \mat \Sigma^0, \mat V^1) \\
    &\ge \psi(\mat U^1, \mat \Sigma^1, \mat V^1) \\
    & \vdots \\
    &\ge \psi(\mat U^r, \mat \Sigma^r, \mat V^r),
\end{array}
\]
for every $r \ge 1$. Therefore, the sequence $\{\psi(\mat U^r, \mat \Sigma^r, \mat V^r)\}$, which is indexed by $r$, is decreasing as $r$ increases. Since $\psi(\mat U, \mat \Sigma, \mat V) \ge 0$ for every feasible $(\mat U, \mat \Sigma, \mat V)$, according to the monotone convergence theorem, $\psi(\mat U^r, \mat \Sigma^r, \mat V^r)$ monotonically converges to a non-negative value as $r$ goes to infinity.

Note that $(\mat U^r, \mat \Sigma^r, \mat V^r)$ is not guaranteed to converge because at the $r^\th$ iteration, the values of $(\mat U^r, \mat \Sigma^r, \mat V^r)$ may not be unique. This completes the proof.
\stp
\end{proof}

\section{Proof of Proposition \captext{\ref{prop:g-continuous}}}\label{append:g-continuous}
\begin{proof}
According to Lemma \ref{lem:max-form}, the upper bound function $\bar g(\mat H)$ is straightforward to obtain because $\matb X$ is a feasible solution in $\cal X$.

The Lipschitz continuity of $g$ is shown as follows. For every $\mat H_1, \mat H_2 \in \cal H$, we have
\[
\begin{array}{l}
   |g(\mat H_1) - g(\mat H_2)| \\
   \quad\quad = \Big|\displaystyle \min_{\mat X \in \cal X} \rho \|\mat H_1 \mat X - \mat S\|^2_F  + (1 - \rho) \|\mat X  - \mat X_s\|^2_F - \\
   \quad \quad \quad \quad \displaystyle \min_{\mat X \in \cal X} \rho \|\mat H_2 \mat X - \mat S\|^2_F  + (1 - \rho) \|\mat X  - \mat X_s\|^2_F\Big| \\
   \quad\quad \le \displaystyle \max_{\mat X \in \cal X} \Big| \big[\rho \|\mat H_1 \mat X - \mat S\|^2_F  + (1 - \rho) \|\mat X  - \mat X_s\|^2_F\big] - \\
   \quad \quad \quad \quad \big[\rho \|\mat H_2 \mat X - \mat S\|^2_F  + (1 - \rho) \|\mat X  - \mat X_s\|^2_F\big] \Big| \\
   \quad\quad = \rho \cdot \displaystyle \max_{\mat X \in \cal X} \Big|  \|\mat H_1 \mat X - \mat S\|^2_F - \|\mat H_2 \mat X - \mat S\|^2_F\Big| \\
   \quad\quad \le \rho \cdot L_f \cdot \| \mat H_1 - \mat H_2 \|.
\end{array}
\]

Using the same manipulations, $\bar g (\mat H)$ can also be shown to be $\rho L_f$-Lipschitz continuous. This completes the proof. \stp
\end{proof}

\section{Details on Experiments}\label{append:simulation-engine}
In this appendix, we detail the logic flow upon which the shared source codes are written.
\begin{algorithm}[htbp]
    \caption{Simulation Engine}
    \label{algo:simulation-engine}
    \begin{flushleft}
        \justifying
        \textbf{Definition}: Let $I$ denote the number of Monte--Carlo episodes.\\
        \textbf{Remark}: The notation $0:0.01:0.2$ means a discrete vector starting at $0$, ending with $0.2$, and uniformly spaced with $0.01$.
    \end{flushleft}
    \begin{algorithmic}[1]
        \Require $I = 1000$
            \State // \bfit{Stage 1: Engine Initialization}
            \State $N \leftarrow 16$, $K \leftarrow 4$, $L \leftarrow 30$, $f \leftarrow 5.9\times 10^{9}$, $P_{\T} \leftarrow 2.5$, $\theta \leftarrow 0:0.01:0.2$
            \State Generate $\mat H_{\text{ref}}$
            \State Design Perfect-Sensing Waveform $\mat X_s$
            \State // \bfit{Stage 2: Offline Design Using Practically Available Nominal Channel}
            \State Generate Constellation $\mat S$
            \State Generate Nominal Channel $\bar{\mat H}$ Using \eqref{eq:simulation-engine}
            \State Design Nominally Optimal Waveform $\bar{\mat X}$ Using $\bar{\mat H}$
            \State Calculate Nominally Estimated AASR $R_{\bar{\mat H}, \bar{\mat X}}$
            \State Design Robust Waveform $\mat X^*$ Using $\bar{\mat H}$ for Every $\theta$ (NB: $\mat X^*$ depends on $\theta$)
            \State Calculate Robustly Estimated AASR $R_{\mat H^*, \mat X^*}$
            \State // \bfit{Stage 3: Online Test Using Random And Practically Unknown True Channel}
            \State $i \leftarrow 0$
            \While {true}
                \State // \bfit{Calculate True AASRs}
                \State Uniformly Generate $\mat H_0$ Using \eqref{eq:simulation-engine}
                \State Calculate True AASR $R_{\mat H_0, \bar{\mat X}}$ at Nominally Optimal Waveform $\bar{\mat X}$
                \State Calculate True AASR $R_{\mat H_0, \mat X^*}$ at Robust Waveform $\mat X^*$ for Every $\theta$ (NB: $\mat X^*$ depends on $\theta$)
                \State // \bfit{Next Episode}
                \State $i \leftarrow i + 1$
                \State // \bfit{End of Simulation}
                \If {$i > I$}
                    \State {\bf break while}
                \EndIf
            \EndWhile
        \Ensure AASRs for all $I$ episodes (used for box plots)
    \end{algorithmic}
\end{algorithm}

\section{Additional Results on Different \captext{$\epsilon$}}\label{append:add-results-epsilon}
The gap defined in Lemma \ref{lem:gap}, i.e., the value of $\|\mat H^* \matb X - \mat S\|^2_F - \|\mat H^* \mat X^* - \mat S\|^2_F$ (recall Methods \ref{method:max-form-solution} and \ref{method:max-form-solution-another}), is shown in Table \ref{tab:gap}. The gap is statistically small even for large $\epsilon$.
\begin{table}[!htbp]
\centering
\caption{Gap in Lemma \ref{lem:gap}}
\begin{tabular}{ccccccc}
\hline
 $\epsilon $  &  0.1  & 0.25  & 0.5 \\
 \hline
 \tabincell{c}{Gap}          &  $0.055 \pm 0.031$  & $0.234 \pm 0.089$ &  $0.706 \pm 0.188$ \\
 \hline
 \multicolumn{4}{l}{\footnotesize Note: Format: mean $\pm$ std; when $\epsilon < 0.1$, the gap is almost zero.}
\end{tabular}
\label{tab:gap}
\end{table}

{\small
\section*{References}
\begin{hangparas}{.25in}{1}
\noindent [49] R. Enhbat, “An algorithm for maximizing a convex function over a simple set,” Journal of Global Optimization, vol. 8, pp. 379–391, 1996.

\noindent [50] R. Everson, “Orthogonal, but not orthonormal, procrustes problems,” Advances in Computational Mathematics, vol. 3, no. 4, 1998.
\end{hangparas}
}
\end{document}